\documentclass{article}

\usepackage[english]{babel}
\usepackage{natbib}
\usepackage[letterpaper,top=2cm,bottom=2cm,left=3cm,right=3cm,marginparwidth=1.75cm]{geometry}
\usepackage{amsmath, amsthm, amssymb}
\usepackage{graphicx}
\usepackage[colorlinks=true, allcolors=blue]{hyperref}
\usepackage{subcaption}
\usepackage{mathtools}
\usepackage{wrapfig}

\usepackage{algorithm}
\usepackage{algpseudocode}

\newcommand{\vecmat}[1]{#1}%{\boldsymbol{#1}}
\newcommand{\blind}{1}%JASA-template

\DeclareMathOperator*{\argmax}{argmax}
\DeclareMathOperator*{\tr}{tr}
\newtheorem{theorem}{Theorem}[section]
\newtheorem{proposition}[theorem]{Proposition}

\title{Iterative Methods for Vecchia-Laplace Approximations for Latent Gaussian Process Models}
\author{
  Pascal Kündig\footnotemark[1]  \footnotemark[3] \footnotemark[4]\\
  \and
  Fabio Sigrist\footnotemark[2] \footnotemark[1] 
}

\begin{document}
\date{}
\maketitle

\begin{abstract}
Latent Gaussian process (GP) models are flexible probabilistic non-parametric function models. Vecchia approximations are accurate approximations for GPs to overcome computational bottlenecks for large data, and the Laplace approximation is a fast method with asymptotic convergence guarantees to approximate marginal likelihoods and posterior predictive distributions for non-Gaussian likelihoods. Unfortunately, the computational complexity of combined Vecchia-Laplace approximations grows faster than linearly in the sample size when used in combination with direct solver methods such as the Cholesky decomposition. Computations with Vecchia-Laplace approximations can thus become prohibitively slow precisely when the approximations are usually the most accurate, i.e., on large data sets. In this article, we present iterative methods to overcome this drawback. Among other things, we introduce and analyze several preconditioners, derive new convergence results, and propose novel methods for accurately approximating predictive variances. We analyze our proposed methods theoretically and in experiments with simulated and real-world data. In particular, we obtain a speed-up of an order of magnitude compared to Cholesky-based calculations and a threefold increase in prediction accuracy in terms of the continuous ranked probability score compared to a state-of-the-art method on a large satellite data set. All methods are implemented in a free C++ software library with high-level Python and R packages.%\footnotemark[5]
\end{abstract}
\footnotetext[1]{Lucerne University of Applied Sciences and Arts}
\footnotetext[2]{Seminar for Statistics, ETH Zurich}
\footnotetext[3]{University of Basel}
\footnotetext[4]{Corresponding author: pascal.kuendig@hslu.ch}
%\footnotetext[5]{https://github.com/fabsig/GPBoost}

\section{Introduction}

Gaussian processes (GPs) \citep{cressie2015statistics, williams2006gaussian} are a flexible class of probabilistic non-parametric models with numerous successful applications in statistics, machine learning, and other disciplines. Traditionally, estimation and prediction are done using the Cholesky decomposition of the covariance matrix which has $O(n^3)$ computational complexity. Multiple methods have been proposed to overcome this computational bottleneck for large data, see \citet{heaton2019case} for a review. In spatial statistics, Vecchia approximations \citep{vecchia1988estimation, datta2016hierarchical, katzfuss2017general} have recently ``emerged as a leader among the sea of approximations" \citep{guinness2019gaussian}. %In summary, Vecchia approximations approximate Cholesky factors of precision matrices using sparse matrices. These approximate sparse Cholesky factors can be computed with computational cost and memory requirements that grow linear in the sample size.

When GPs are used in models with non-Gaussian likelihoods, marginal likelihoods and posterior predictive distributions can usually not be calculated in closed-form, and an approximation such as Laplace's method, expectation propagation, or variational Bayes methods must be used. The Laplace approximation \citep{tierney1986accurate} is typically the computationally most efficient method, but, depending on the likelihood, it can be inaccurate for small sample sizes \citep{nickisch2008approximations}. In particular, estimated variance parameters can be downward biased. However, the Laplace approximation converges asymptotically to the correct quantity, and the approximation is thus expected to become accurate for large data sets. Supporting this argument, we show in Figure \ref{fig:LaplaceBias} the estimated marginal variance parameter obtained with a Vecchia-Laplace approximation for varying sample sizes $n$ on simulated data with a Bernoulli likelihood and the setting described in Section \ref{simSetting}. Estimation is done using the iterative methods introduced in this paper and repeated on $100$ simulated data sets for every $n$. Figure \ref{fig:LaplaceBias} shows that the downward bias of the variance parameter vanishes as $n$ grows.%\footnote{Note that the convergence rate of the Laplace approximation for Gaussian processes depends on the ``actual dimension" of the latent Gaussian process \citep{rue2009approximate}. In a typical infill asymptotic regime and assuming that range parameters do not decrease with the sample size, the Laplace approximation thus converges to the correct quantity.}
\begin{figure}[ht!]
    \centering
    \includegraphics[width=0.8\linewidth]{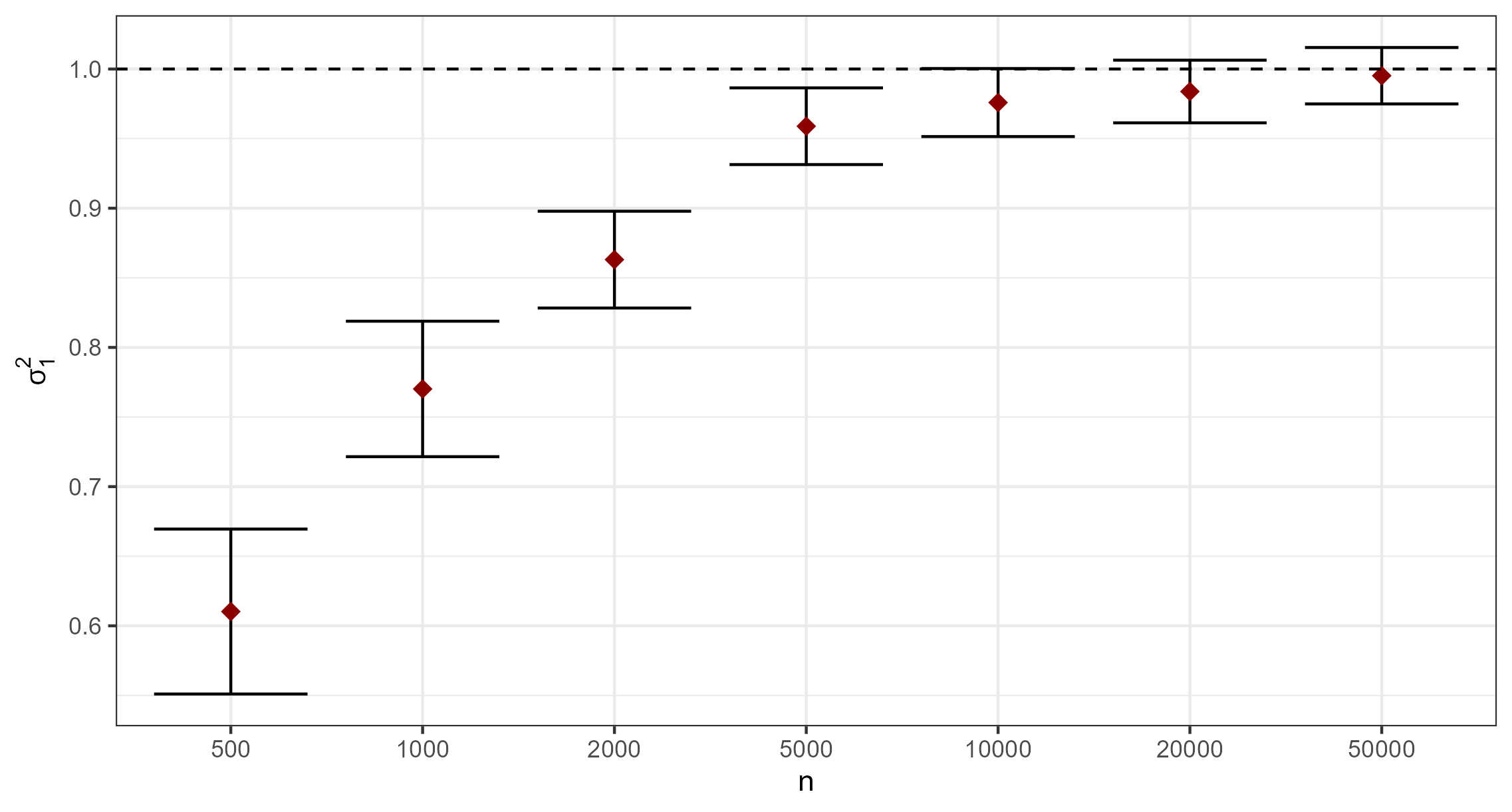}
    \caption{Estimated variance parameter $\sigma^2_1$ obtained with a Vecchia-Laplace approximation vs. different sample sizes $n$ for binary data. The red rhombi represent means and the whiskers are $\pm 2 \times$ standard errors. The dashed line indicates the true parameter $\sigma^2_1=1$.}
    \label{fig:LaplaceBias}
\end{figure}
Unfortunately, the computational costs of combined Vecchia-Laplace approximations still grow fast in the sample size when used in combination with direct solver methods such as the Cholesky decomposition, which is the go-to approach to date. Specifically, the computational complexity is $O(n^\frac{3}{2})$ in two dimensions and at least $O(n^2)$ in higher dimensions \citep{zilber2021vecchia}. I.e., computations with Vecchia-Laplace approximations can become prohibitively slow when the approximations are usually the most accurate. 

In this article, we show how parameter estimation and computation of posterior predictive distributions can be done with Vecchia-Laplace approximations using iterative methods \citep{saad2003iterative} such as the conjugate gradient (CG) method and stochastic Lanczos quadrature (SLQ) \citep{ubaru2017fast} which rely on matrix-vector multiplications with sparse triangular matrices. We propose and analyze several preconditioners to accelerate convergence and reduce variances of stochastic approximations. Gradients are calculated using stochastic trace estimation (STE) with almost no computational overhead once likelihoods are evaluated. To increase the precision of stochastic gradients, we use a novel control variate-based form of variance reduction. Moreover, we introduce and compare different methods for calculating predictive variances. We analyze our methods theoretically and in experiments using simulated and real-world data. Our theoretical results show that the CG method with our proposed preconditioners is expected to converge fast. Supporting this argument, we obtain a reduction in runtime of about one order of magnitude compared to the Cholesky decomposition with identical accuracy for parameter estimation and prediction. Further, the application of our methods on a large satellite water vapor data set leads to a threefold increase in prediction accuracy in terms of the continuous ranked probability score compared to a current state-of-the-art method \citep{zilber2021vecchia}.

\subsection{Relation to existing work}
Iterative methods have been used to reduce the computational cost for GPs in both statistics \citep{stein2013stochastic, aune2013iterative, aune2014parameter, majumder2022kryging} and machine learning \citep{gardner2018gpytorch, wenger2022preconditioning}. \citet{schaefer2021sparse} apply the CG method to solve linear systems when applying a Vecchia approximation to a latent GP in a model with a Gaussian likelihood. Recently, \citet{gelsinger2023log} propose to use the CG method and STE in an EM algorithm when applying a Laplace approximation and a spectral approximation for log-Gaussian Cox process models. None of the previous works have applied iterative methods to Vecchia-Laplace approximations. Instead of inverting matrices in Newton's method for finding the mode, \citet{zilber2021vecchia} show that the solution of the corresponding linear system can be seen as a pseudo-posterior mean and propose to apply a Vecchia approximation to the joint distribution of the latent GP and Gaussian pseudo-data for calculating this pseudo-posterior mean. Vecchia approximations are thus not applied to approximate a statistical model but to approximate a numerical computation. %Experiments in this article show that the approach of \citet{zilber2021vecchia} is less accurate and slower than our proposed methods. %\footnote{Or to obtain a valid statistical model, depending on the point of view.}%\citet{zilber2021vecchia} also show how approximate marginal likelihoods and predictive distributions can be calculated using similar arguments.

\section{Vecchia-Laplace approximations}\label{summary_VL}
Let $b(\cdot) \sim GP(0,c_{\theta}(\cdot,\cdot))$ be a zero-mean latent Gaussian process (GP) on a domain $\mathcal{D} \subset \mathbb{R}^d$ with a covariance function $\text{Cov}(b(\vecmat{s}),b(\vecmat{s}')) = c_{\theta}(\vecmat{s},\vecmat{s}')$, $\vecmat{s},\vecmat{s}'\in \mathcal{D}$, that depends on a set of parameters $\vecmat{\theta}\in\Theta\subset \mathbb{R}^q$. 
Assume further that a response variable \(\vecmat{y} \in \mathbb{R}^n\) observed at locations \(\vecmat{s}_1, \dots, \vecmat{s}_n \in \mathcal{D}\) follows a distribution with a density \(p(\vecmat{y} \mid \vecmat{\mu}, \vecmat{\xi})\), where \(\vecmat{\mu} \in \mathbb{R}^n\) are (potentially link function-transformed) parameters, and \(\vecmat{\xi} \in \Xi \subset \mathbb{R}^r\) are auxiliary parameters, such as a shape parameter for a gamma likelihood. Assume further that \(p(\vecmat{y} \mid \vecmat{\mu}, \vecmat{\xi})\) is log-concave in \(\vecmat{\mu}\). The parameter $\vecmat{\mu}=\vecmat{F}(\vecmat{X})+\vecmat{b}$ is the sum of a deterministic function $\vecmat{F}(\vecmat{X})$ and $\vecmat{b} = (b(\vecmat{s}_1),\dots,b(\vecmat{s}_n))^T\sim \mathcal{N}(\vecmat{0},\vecmat{\Sigma})$, where $\vecmat{X}\in\mathbb{R}^{n\times p}$ contains predictor variables and $\Sigma_{ij} = c_{\theta}(\vecmat{s_i},\vecmat{s_j}), 1\leq i,j\leq n$. Typically, $\vecmat{F}(\vecmat{X})=\vecmat{X}\vecmat{\beta}, \vecmat{\beta} \in \mathbb{R}^p$, but $F(\cdot)$ can also be modeled using machine learning methods such as tree-boosting \citep{sigrist2022latent}. We assume conditional independence of $\vecmat{y}$ given $\vecmat{\mu}$: $p(\vecmat{y}|\vecmat{\mu},\vecmat{\xi})=\prod_{i=1}^n p(y_i|\mu_i,\vecmat{\xi})$. The marginal likelihood is given by $p(\vecmat{y}|\vecmat{F},\vecmat{\theta},\vecmat{\xi})=\int p(\vecmat{y}|\vecmat{\mu},\vecmat{\xi})p(\vecmat{b}|\vecmat{\theta})d\vecmat{b}$. In the case of non-Gaussian likelihoods, there is typically no analytic expression for $p(\vecmat{y}|\vecmat{F},\vecmat{\theta},\vecmat{\xi})$ and an approximation has to be used.

\subsection{Vecchia approximations}\label{vecchia_approx}
We apply a Vecchia approximation to the latent GP: $p(\vecmat{b}|\vecmat{\theta}) \approx \prod_{i=1}^n p(b_i|\vecmat{b}_{N(i)},\vecmat{\theta})$, where $\vecmat{b}_{N(i)}$ are subsets of $(b_1,\dots,b_{i-1})$, and $N(i) \subseteq \{1,\dots,i-1\}$ with $|N(i)| \leq m$. If $i > m+1$, $N(i)$ is often chosen as the $m$ nearest neighbors of $\vecmat{s}_i$ among $\vecmat{s}_1,\dots,\vecmat{s}_{i-1}$. Note that $p(b_i|\vecmat{b}_{N(i)},\vecmat{\theta}) = \mathcal{N}(\vecmat{A}_i \vecmat{b}_{N(i)},D_i)$, where $\vecmat{A}_i = \vecmat{\Sigma}_{i,N(i)} \vecmat{\Sigma}_{N(i)}^{-1}$, $D_i = \Sigma_{i,i}  - \vecmat{\Sigma}_{N(i),i}^T \vecmat{\Sigma}_{N(i)}^{-1} \vecmat{\Sigma}_{N(i),i}$, $\vecmat{\Sigma}_{N(i),i}$ is a sub-vector of $\vecmat{\Sigma}$ with the $i$-th column and row indices $N(i)$, and $\vecmat{\Sigma}_{N(i)}$ denotes sub-matrix of $\vecmat{\Sigma}$ consisting of rows and columns $N(i)$. Defining a sparse lower triangular matrix $\vecmat{B}\in\mathbb{R}^{n\times n}$ with $1$'s on the diagonal, off-diagonal entries $\vecmat{B}_{i,N_{(i)}} = -\vecmat{A}_i$, and $0$ otherwise, and a diagonal matrix $\vecmat{D}\in\mathbb{R}^{n\times n}$ with $D_i$ on the diagonal, one obtains the approximation $\vecmat{b} \overset{\text{approx}}{\sim} \mathcal{N}(\vecmat{0},\tilde{\vecmat{\Sigma}})$, where $\tilde{\vecmat{\Sigma}}^{-1} = \vecmat{B}^T\vecmat{D}^{-1}\vecmat{B}$ is sparse. We denote the density of this approximate distribution as $\tilde p(\vecmat{b}|\vecmat{\theta})$. Calculating a Vecchia approximation has $O(nm^3)$ computational and $O(nm)$ memory cost. Often, accurate approximations are obtained for small $m$'s. %\citet{datta2016hierarchical} recommend a value between $10$ and $15$.

\subsection{Vecchia-Laplace approximations}\label{VLA}
A Vecchia-Laplace approximation is obtained by combining a Vecchia with a Laplace approximation. Specifically, a Vecchia-Laplace approximation to the negative log-marginal likelihood $-\log(p(\vecmat{y}|\vecmat{F},\vecmat{\theta},\vecmat{\xi}))$ modulo constant terms is given by
\begin{equation}\label{VLA_loss}
L^{VLA}(\vecmat{y},\vecmat{F},\vecmat{\theta},\vecmat{\xi})=-\log p(\vecmat{y}|\vecmat{\mu}^*,\vecmat{\xi}) + \frac{1}{2} {\vecmat{b}^*}^T\tilde{\vecmat{\Sigma}}^{-1} \vecmat{b}^* + \frac{1}{2}\log\det\left(\tilde{\vecmat{\Sigma}} \vecmat{W}+\vecmat{I}_n\right),
\end{equation}
where $\vecmat{\mu}^* = \vecmat{F}(\vecmat{X})+\vecmat{b}^*$ and $\vecmat{b}^*=\argmax_{\vecmat{b}}\log p(\vecmat{y}|\vecmat{\mu},\vecmat{\xi}) - \frac{1}{2} \vecmat{b}^T\tilde{\vecmat{\Sigma}}^{-1} \vecmat{b}$ is the mode of $p(\vecmat{y}|\vecmat{\mu},\vecmat{\xi})\tilde p(\vecmat{b}|\vecmat{\theta})$, and $\vecmat{W}\in\mathbb{R}^{n\times n}$ is diagonal with $W_{ii}=-\frac{\partial^2 \log p(y_i| \mu_i,\vecmat{\xi})}{\partial \mu_i^2}\Big|_{\vecmat{\mu}=\vecmat{\mu}^*}$. The mode $\vecmat{b}^*$ is usually found with Newton's method \citep{williams2006gaussian}, and one iteration is given by
\begin{equation}\label{newton}
%\begin{split}
\vecmat{b}^{*t+1}
%&=b-\left(\frac{\partial^2\Psi(b)}{\partial b^2}\right)^{-1}
%\left(\frac{\partial\Psi(b)}{\partial b} \right) \\
=\left(\vecmat{W} + \tilde{\vecmat{\Sigma}}^{-1}\right)^{-1}\left(\vecmat{W} \vecmat{b}^{*t}+\frac{\partial\log p(\vecmat{y}|\vecmat{\mu}^{*t},\vecmat{\xi})}{\partial \vecmat{b}}\right), ~~ t = 0, 1, \dots
%\end{split}
\end{equation}
Since $p(\vecmat{y}|\vecmat{F},\vecmat{\theta},\vecmat{\xi}) = p(\vecmat{y}|\vecmat{F},\vecmat{b},\vecmat{\xi})\tilde p(\vecmat{b}|\vecmat{\theta}) / p(\vecmat{b}|\vecmat{y},\vecmat{\theta},\vecmat{\xi})$, a Vecchia-Laplace approximation for $p(\vecmat{y}|\vecmat{F},\vecmat{\theta},\vecmat{\xi})$ is equivalent to the approximation $p(\vecmat{b}|\vecmat{y},\vecmat{\theta},\vecmat{\xi}) \approx \mathcal{N}\left(\vecmat{b}^*,\left(\vecmat{W}+\tilde{\vecmat{\Sigma}}^{-1}\right)^{-1}\right)$. If a first- or second-order method is used for minimizing $L^{VLA}(\vecmat{y},\vecmat{F},\vecmat{\theta},\vecmat{\xi})$, gradients with respect to $\vecmat{\theta}$, $\vecmat{F}$, and $\vecmat{\xi}$ are needed. We include these in Appendix \ref{gradientsVLA}. Note that gradients with respect to $\vecmat{F}$ are used, for instance, for linear models since $\frac{\partial L^{VLA}(\vecmat{y},\vecmat{F},\vecmat{\theta},\vecmat{\xi})}{\partial \vecmat{\beta}}=\vecmat{X}^T\frac{\partial L^{VLA}(\vecmat{y},\vecmat{F},\vecmat{\theta},\vecmat{\xi})}{\partial \vecmat{F}}$ when $\vecmat{F}(\vecmat{X})=\vecmat{X\beta}$.

% or also for tree-boosting based models such as the one of \citet{sigrist2022latent}.
 
% an approximation for the marginal likelihood is given by:
% \begin{equation}\label{vecchiaLaplace}
% p(\vecmat{y}|\vecmat{F},\vecmat{\theta},\vecmat{\xi})\approx  p(\vecmat{y}|\vecmat{\mu}^*,\vecmat{\xi})\tilde p(\vecmat{b}^*|\vecmat{\theta}) \det\left(\vecmat{W}+\tilde{\vecmat{\Sigma}}^{-1}\right)^{- 1/2}(2\pi)^{n/2},
% \end{equation}

%By the term ``Laplace approximation", we refer to both the integral approximation of the marginal likelihood and the corresponding Gaussian approximation to the posterior distribution in this article since one implies the other, and it is clear from the context what is meant. 

\subsubsection{Prediction with Vecchia-Laplace approximations}\label{prediction}
The goal is to predict either the latent variable $\vecmat{\mu}_p\in \mathbb{R}^{n_p}$ or the response variable $\vecmat{y}_p \in \mathbb{R}^{n_p}$ at $n_p$ locations $\vecmat{s}_{p,1},\dots,\vecmat{s}_{p,n_p} \in \mathcal{D}$ using the posterior predictive distributions $p(\vecmat{\mu}_p|\vecmat{y},\vecmat{\theta},\vecmat{\xi})=\int p(\vecmat{\mu}_p|\vecmat{b},\vecmat{\theta})p(\vecmat{b}|\vecmat{y},\vecmat{\theta},\vecmat{\xi})d\vecmat{b}$ and $p(\vecmat{y}_p|\vecmat{y},\vecmat{\theta},\vecmat{\xi})=\int p(\vecmat{y}_p|\vecmat{\mu}_p,\vecmat{\xi}) p(\vecmat{\mu}_p|\vecmat{y},\vecmat{\theta},\vecmat{\xi})d\vecmat{\mu}_p$, respectively. A Vecchia approximation is applied to the joint distribution of the GP at the training and prediction locations $(\vecmat{b},\vecmat{\mu}_p) \in \mathbb{R}^{n+n_p}$ \citep{katzfuss2020vecchia}. We focus on the case where the observed locations appear first in the ordering which has the advantage that the nearest neighbors found for estimation can be reused. We obtain $\vecmat{\mu}_p|\vecmat{b} \sim \mathcal{N}(\vecmat{\nu}_p,\vecmat{\Xi}_p)$, where $\vecmat{\nu}_p = \vecmat{F}(\vecmat{X}_p)-\vecmat{B}_p^{-1}\vecmat{B}_{po}\vecmat{b}$, $\vecmat{X}_p \in \mathbb{R}^{n_p\times p}$ is the predictive design matrix, and $\vecmat{\Xi}_p = \vecmat{B}_p^{-1}\vecmat{D}_p\vecmat{B}_p^{-T}$, where $\vecmat{B}_{po} \in \mathbb{R}^{n_p\times n}$, $\vecmat{B}_p \in \mathbb{R}^{n_p \times n_p}$, $\vecmat{D}_p^{-1} \in \mathbb{R}^{n_p \times n_p}$ are the following sub-matrices of the approximate precision matrix $\Tilde{Cov}\bigl((\vecmat{b},\vecmat{\mu}_p)\bigl)^{-1}$ \citep[see, e.g.,][Proposition 3.3]{sigrist2022gaussian}:
\begin{equation*}
\begin{split}
    \Tilde{Cov}\bigl((\vecmat{b},\vecmat{\mu}_p)\bigl)^{-1} =  \begin{pmatrix} \vecmat{B} & \vecmat{0}\\ \vecmat{B}_{po} &\vecmat{B}_p \end{pmatrix}^T \begin{pmatrix} \vecmat{D}^{-1} &\vecmat{0}\\ \vecmat{0} &\vecmat{D}_p^{-1} \end{pmatrix} \begin{pmatrix} \vecmat{B} & \vecmat{0}\\ \vecmat{B}_{po} &\vecmat{B}_p \end{pmatrix}.
\end{split}
\end{equation*}
Combining $p(\vecmat{b}|\vecmat{y},\vecmat{\theta},\vecmat{\xi}) \approx \mathcal{N}\left(\vecmat{b}^*,\left(\vecmat{W}+\tilde{\vecmat{\Sigma}}^{-1}\right)^{-1}\right)$ and $\vecmat{\mu}_p|\vecmat{b} \sim \mathcal{N}(\vecmat{\nu}_p,\vecmat{\Xi}_p)$, we obtain 
\begin{equation}\label{latentpostpredictive}
    \begin{split}
        p(\vecmat{\mu}_p|\vecmat{y},\vecmat{\theta},\vecmat{\xi}) \approx \mathcal{N}\left(\vecmat{\omega}_p,\vecmat{\Omega}_p\right), ~~\vecmat{\omega}_p &= \vecmat{F}(\vecmat{X}_p)-\vecmat{B}_p^{-1}\vecmat{B}_{po}\vecmat{b}^*,\\
        \vecmat{\Omega}_p &= \vecmat{B}_p^{-1} \vecmat{D}_p \vecmat{B}_p^{-T} + \vecmat{B}_p^{-1}\vecmat{B}_{po}(\vecmat{W} + \tilde{\vecmat{\Sigma}}^{-1})^{-1} \vecmat{B}_{po}^T \vecmat{B}_p^{-1}.
    \end{split}
\end{equation}
The integral in $p(\vecmat{y}_p|\vecmat{y},\vecmat{\theta},\vecmat{\xi})$ is analytically intractable for most likelihoods and must be approximated using numerical integration or by simulating from $p(\vecmat{\mu}_p|\vecmat{y},\vecmat{\theta},\vecmat{\xi})\approx \mathcal{N}\left(\vecmat{\omega}_p,\vecmat{\Omega}_p\right)$.

\section{Iterative methods for Vecchia-Laplace approximations}\label{VLAwithItMethods}
Parameter estimation and prediction with Vecchia-Laplace approximations involve several time-consuming operations. First, calculating linear solves $(\vecmat{W}+\tilde{\vecmat{\Sigma}}^{-1}) \vecmat{u} = \vecmat{b}$, $\vecmat{b} \in \mathbb{R}^n$ (i) in Newton's method for finding the mode, see \eqref{newton}, (ii) for implicit derivatives of the log-marginal likelihood, see Appendix \ref{gradientsVLA}, and (iii) for predictive distributions, see \eqref{latentpostpredictive}. Second, calculating the log-determinant $\log\det(\tilde{\vecmat{\Sigma}} \vecmat{W}+\vecmat{I}_n)$ in $L^{VLA}(\vecmat{y},\vecmat{F},\vecmat{\theta},\vecmat{\xi})$. And third, calculating trace terms such as $\tr\left((\vecmat{W}+\tilde{\vecmat{\Sigma}}^{-1})^{-1}\frac{\partial (\vecmat{W}+\tilde{\vecmat{\Sigma}}^{-1})}{\partial{\theta}_k}\right)$ for the derivatives of log-determinants, see Appendix \ref{gradientsVLA}. Traditionally, these operations are performed using a Cholesky decomposition of $\vecmat{W}+\tilde{\vecmat{\Sigma}}^{-1}$. We propose to do these operations with iterative methods as follows. 

For linear solves with the matrix $\vecmat{W}+\tilde{\vecmat{\Sigma}}^{-1}$, we use the preconditioned conjugate gradient (CG) method, which solves a linear system $(\vecmat{W}+\tilde{\vecmat{\Sigma}}^{-1}) \vecmat{u} = \vecmat{b}$ by iteratively doing matrix-vector multiplications with $\vecmat{W}+\tilde{\vecmat{\Sigma}}^{-1} = \vecmat{W}+\vecmat{B}^T\vecmat{D}^{-1}\vecmat{B}$. This can be done fast since $\vecmat{B}$ is a sparse triangular matrix. The convergence properties of the CG method can be drastically improved by preconditioning. This means solving the equivalent system 
\begin{equation}\label{pcls1}
    \vecmat{P}^{-\frac{1}{2}}(\vecmat{W}+\tilde{\vecmat{\Sigma}}^{-1})\vecmat{P}^{-\frac{T}{2}}\vecmat{P}^{\frac{T}{2}}\vecmat{u} = \vecmat{P}^{-\frac{1}{2}}\vecmat{b},
\end{equation}
where $\vecmat{P}$ is a symmetric positive definite matrix. The preconditioned CG algorithm is included in Appendix \ref{appendix:CGalgo}. The convergence rate of the preconditioned CG method depends on the condition number of $\vecmat{P}^{-\frac{1}{2}}(\vecmat{W}+\tilde{\vecmat{\Sigma}}^{-1})\vecmat{P}^{-\frac{T}{2}}$. In practice, numerical convergence is often achieved with $l \ll n$ iterations, and linear solves with $\vecmat{W}+\tilde{\vecmat{\Sigma}}^{-1}$ can thus be calculated in $O(ln)$ time. In Section \ref{section:conv_theory}, we derive convergence rates for the CG method with two preconditioners proposed in this article. Alternatively, since $\vecmat{W} + \tilde{\vecmat{\Sigma}}^{-1} = \tilde{\vecmat{\Sigma}}^{-1}(\tilde{\vecmat{\Sigma}} + \vecmat{W}^{-1})\vecmat{W}$, we can also solve%, or, equivalently, $(\vecmat{W} + \tilde{\vecmat{\Sigma}}^{-1})^{-1} = W^{-1}(\tilde{\vecmat{\Sigma}} + \vecmat{W^{-1}})^{-1}\tilde{\vecmat{\Sigma}}$ 
\begin{equation}\label{pcls2}
\vecmat{P}^{-\frac{1}{2}}(\tilde{\vecmat{\Sigma}} + \vecmat{W}^{-1})\vecmat{P}^{-\frac{T}{2}}\vecmat{P}^{\frac{T}{2}} \vecmat{W} \vecmat{u} = \vecmat{P}^{-\frac{1}{2}}\tilde{\vecmat{\Sigma}} \vecmat{b}.
\end{equation}
Multiplication with $\tilde{\vecmat{\Sigma}}$ also costs $O(n)$ as this only involves sparse triangular solves.

% For prediction, the calculation of the posterior predictive (co-)variances $\vecmat{\Omega}_p$ in \eqref{latentpostpredictive} requires solving $n_p$ linear equation systems with the matrix $\vecmat{W}+\tilde{\vecmat{\Sigma}}^{-1}$. For large $n_p$ and $n$, solving these systems with the CG method can thus be computationally prohibitive. In Section \ref{section:iterativePred}, we present several alternative methods to approximate predictive (co-)variances.

Different techniques have been proposed in the literature to calculate the log-determinant of a large, symmetric positive definite matrix. \citet{dong2017scalable} compare stochastic Lanczos quadrature (SLQ) \citep{ubaru2017fast}, Chebyshev approximations, and other methods and conclude that SLQ achieves the highest accuracy and fastest runtime. In this article, we use SLQ in combination with the preconditioned CG method to approximate $\log\det(\tilde{\vecmat{\Sigma}}\vecmat{W}+\vecmat{I}_n)$. Specifically, we first note that
\begin{equation}\label{logdetsplit}
    \log\det(\tilde{\vecmat{\Sigma}}\vecmat{W}+\vecmat{I}_n) = -\log\det(\vecmat{D}^{-1}) + \log\det(\vecmat{P}) + \log\det(\vecmat{P}^{-\frac{1}{2}}(\vecmat{W}+\tilde{\vecmat{\Sigma}}^{-1})\vecmat{P}^{-\frac{T}{2}}),
% \begin{split}
%     \log\det(\tilde{\vecmat{\Sigma}}\vecmat{W}+\vecmat{I}_n) &= \log\det(\tilde{\vecmat{\Sigma}}) + \log\det(\vecmat{W}+\tilde{\vecmat{\Sigma}}^{-1})\\
%                                      &= -\log\det(\vecmat{D}^{-1}) + \log\det(\vecmat{P}) + \log\det(\vecmat{P}^{-\frac{1}{2}}(\vecmat{W}+\tilde{\vecmat{\Sigma}}^{-1})\vecmat{P}^{-\frac{T}{2}}),
% \end{split}
\end{equation}
and
\begin{equation}\label{logdetsplit2}
    \log\det(\tilde{\vecmat{\Sigma}}\vecmat{W}+\vecmat{I}_n) = \log\det(\vecmat{W}) + \log\det(\vecmat{P}) + \log\det(\vecmat{P}^{-\frac{1}{2}}(\vecmat{W}^{-1}+\tilde{\vecmat{\Sigma}})\vecmat{P}^{-\frac{T}{2}}).
\end{equation}
We then use SLQ to approximate the last terms in \eqref{logdetsplit} and \eqref{logdetsplit2}. In brief, SLQ uses the relation $\log(\det(\vecmat{A}))=\tr(\log(\vecmat{A}))$ for a symmetric positive-definite matrix $\vecmat{A}$ which allows for applying stochastic trace estimation (STE), and $\log(\vecmat{A})$ is approximated with the Lanczos algorithm. When using \eqref{logdetsplit}, we obtain the following approximation:
\begin{equation}\label{def_SLQ}
    \log\det(\vecmat{P}^{-\frac{1}{2}}(\vecmat{W}+\tilde{\vecmat{\Sigma}}^{-1})\vecmat{P}^{-\frac{T}{2}}) \approx \frac{1}{t} \sum_{i=1}^t\|\vecmat{P}^{-\frac{1}{2}}\vecmat{z}_i\|_2^2 \vecmat{e}_1^T \log(\tilde{\vecmat{T}}_i) \vecmat{e}_1 \approx \frac{n}{t} \sum_{i=1}^t \vecmat{e}_1^T \log(\tilde{\vecmat{T}}_i) \vecmat{e}_1,
\end{equation}
where $\vecmat{z}_1,\dots, \vecmat{z}_t \in \mathbb{R}^{n}$ are i.i.d. random vectors with $\mathbb{E}[\vecmat{z}_i]=0$ and $\mathbb{E}[\vecmat{z}_i\vecmat{z}_i^T]=\vecmat{P}$, $\tilde{\vecmat{Q}}_i \tilde{\vecmat{T}}_i \tilde{\vecmat{Q}}_i^T \approx \vecmat{P}^{-\frac{1}{2}}(\vecmat{W}+\tilde{\vecmat{\Sigma}}^{-1})\vecmat{P}^{-\frac{T}{2}}$ is the partial Lanczos decomposition obtained after $l$ steps of the Lanczos algorithm with $\vecmat{P}^{-\frac{1}{2}}\vecmat{z}_i$ as start vector, $\tilde{\vecmat{Q}}_i \in \mathbb{R}^{n\times l}$, $\tilde{\vecmat{T}}_i\in \mathbb{R}^{l\times l}$, and $\vecmat{e}_1=(1,0,\dots,0)^T$. A derivation of this approximation is given in Appendix \ref{deriv_SLQ}. In this article, we use Gaussian random vectors $\vecmat{z}_i\sim \mathcal{N}(\vecmat{0},\vecmat{P})$. In Section \ref{section:conv_theory}, we derive probabilistic error bounds for the approximation in \eqref{def_SLQ} which depend on the condition number of $\vecmat{P}^{-\frac{1}{2}}(\vecmat{W}+\tilde{\vecmat{\Sigma}}^{-1})\vecmat{P}^{-\frac{T}{2}}$. \citet{wenger2022preconditioning} highlight that a decomposition as in \eqref{logdetsplit} with a preconditioner performs variance reduction. Intuitively, the more accurate the preconditioner $\vecmat{P} \approx (\vecmat{W}+\tilde{\vecmat{\Sigma}}^{-1})$, the smaller $\log\det(\vecmat{P}^{-\frac{1}{2}}(\vecmat{W}+\tilde{\vecmat{\Sigma}}^{-1})\vecmat{P}^{-\frac{T}{2}})$, and thus the smaller the variance of its stochastic approximation.%= \tr(\log(\vecmat{W}+\tilde{\vecmat{\Sigma}}^{-1})-\log(\vecmat{P}))

As in \citet{gardner2018gpytorch}, we use a technique from \citet{saad2003iterative} to calculate the partial Lanczos tridiagonal matrices $\tilde{\vecmat{T}}_1,\dots,\tilde{\vecmat{T}}_t$ from the coefficients of the preconditioned CG algorithm when solving $(\vecmat{W}+\tilde{\vecmat{\Sigma}}^{-1})^{-1}\vecmat{z}_1,\dots, (\vecmat{W}+\tilde{\vecmat{\Sigma}}^{-1})^{-1}\vecmat{z}_t$ $t$ times, see Appendix \ref{appendix:CGalgo}. In doing so, we avoid running the Lanczos algorithm, which brings multiple advantages: Numerical instabilities are not an issue, storing $\tilde{\vecmat{Q}}_i$ is not necessary, and the linear solves $(\vecmat{W}+\tilde{\vecmat{\Sigma}}^{-1})^{-1}\vecmat{z}_i$ can be reused in the STE for calculating derivatives of the log-determinant:
\begin{equation}\label{grad_STE}
    \frac{\partial\log\det(\vecmat{W}+\tilde{\vecmat{\Sigma}}^{-1})}{\partial\theta_k} 
    \approx \frac{1}{t}\sum_{i=1}^t ((\vecmat{W}+\tilde{\vecmat{\Sigma}}^{-1})^{-1}\vecmat{z}_i)^T\frac{\partial (\vecmat{W}+\tilde{\vecmat{\Sigma}}^{-1})}{\partial\theta_k}\vecmat{P}^{-1} \vecmat{z}_i.
\end{equation}
The derivation of this approximation is given in Appendix \ref{deriv_SLQ}. Gradients can thus be calculated with minimal computational overhead once the likelihood is calculated. Further, we do variance reduction by using a control variate based on $\frac{\partial\log\det(\vecmat{P})}{\partial\theta_k}$:
\begin{equation*}
    \frac{\partial\log\det(\vecmat{A})}{\partial\theta_k} \approx \; c\; \tr\left(\vecmat{P}^{-1}\frac{\partial \vecmat{P}}{\partial\theta_k}\right) +\frac{1}{t}\sum_{i=1}^t \underbrace{\left(\vecmat{A}^{-1}\vecmat{z}_i\right)^{T}\frac{\partial \vecmat{A}}{\partial\theta_k}\vecmat{P}^{-1}\vecmat{z}_i}_\text{=:$h(\vecmat{z}_i)$} - c\underbrace{\left(\vecmat{P}^{-1} \vecmat{z}_i\right)^T\frac{\partial \vecmat{P}}{\partial\theta_k}\vecmat{P}^{-1}\vecmat{z}_i}_\text{=:$r(\vecmat{z}_i)$},
\end{equation*}
where $\vecmat{A}=\vecmat{W}+\tilde{\vecmat{\Sigma}}^{-1}$, $\frac{\partial\vecmat{A}}{\partial\theta_k}=\frac{\partial\tilde{\vecmat{\Sigma}}^{-1}}{\partial\theta_k}$, and $c=\widehat{\text{Cov}}(h(\vecmat{z}_i),r(\vecmat{z}_i))/\widehat{\text{Var}}(r(\vecmat{z}_i))$. This is similar to the variance reduction in \citet{wenger2022preconditioning} who, however, use $c=1$. In contrast, we consider $r(\vecmat{z}_i)$ as a control variate and use the optimal weight $c$. In Appendix \ref{PAppendix}, we show in detail how to calculate derivatives of the log-determinant with respect to $\vecmat{F}$, $\vecmat{\theta}$, and $\vecmat{\xi}$ using STE and variance reduction for different preconditioners presented in the following.

\subsection{Preconditioners} \label{section:preconditioners}
Preconditioners accelerate the convergence speed of the CG method and reduce the variance of the approximations of the log-determinant and its derivatives. To practically use a matrix $\vecmat{P}$ as a preconditioner, we need to construct it, perform linear solves with it, calculate $\log \det(\vecmat{P})$, and sample from $\mathcal{N}(\vecmat{0},\vecmat{P})$ in a fast way. Below, we introduce several preconditioners. The preconditioners in Sections \ref{sec_vadu}, \ref{sec_zirc}, and \ref{sec_other} are applied using the versions in \eqref{pcls1} and \eqref{logdetsplit}, whereas the preconditioner in Section \ref{sec_lrac} is used in the alternative expressions in \eqref{pcls2} and \eqref{logdetsplit2}. %Note that the preconditioned CG method only performs linear solves with $\vecmat{P}$ and not with a square root $\vecmat{P}^{\frac{1}{2}}$ even though $\vecmat{P}^{\frac{1}{2}}$ appears, e.g., in \eqref{pcls1}.

\subsubsection{VADU and LVA preconditioners}\label{sec_vadu}
The ``\textbf{V}ecchia \textbf{a}pproximation with \textbf{d}iagonal \textbf{u}pdate" (VADU) preconditioner is given by
\begin{equation}\label{P_VADU}
    \vecmat{P}_{\text{VADU}} = \vecmat{B}^T(\vecmat{W} + \vecmat{D}^{-1})\vecmat{B}.
\end{equation}
Similarly, we also consider the ``\textbf{l}atent \textbf{V}ecchia \textbf{a}pproximation" (LVA) preconditioner $\vecmat{P}_{\text{LVA}} = \vecmat{\vecmat{B}}^T \vecmat{D}^{-1} \vecmat{B}$. In the following, we shed some light on the properties of the VADU and LVA preconditioners using intuitive arguments. Note that
\begin{equation}\label{P_VADU_prec}
\vecmat{P}_{\text{VADU}}^{-\frac{1}{2}}(\vecmat{W}+\tilde{\vecmat{\Sigma}}^{-1})\vecmat{P}_{\text{VADU}}^{-\frac{T}{2}} = (\vecmat{W} + \vecmat{D}^{-1})^{-\frac{1}{2}}\vecmat{B}^{-T}\vecmat{W}\vecmat{B}^{-1}(\vecmat{W} + \vecmat{D}^{-1})^{-\frac{1}{2}} + (\vecmat{D}\vecmat{W} + \vecmat{I}_n)^{-1},
\end{equation}
where $\vecmat{I}_n$ denotes the identity matrix. Heuristically, $\vecmat{P}_{\text{VADU}}^{-\frac{1}{2}}(\vecmat{W}+\tilde{\vecmat{\Sigma}}^{-1})\vecmat{P}_{\text{VADU}}^{-\frac{T}{2}}$ is thus close to $\vecmat{I}_n$ and has a low condition number if either (i) $D_{i}W_{ii}$ is small and $W_{ii} + D_{i}^{-1}$ is large for all $i$ or (ii) $\vecmat{B}\approx \vecmat{I}_n$. In particular, if $\vecmat{B} = \vecmat{I}_n$, then $\vecmat{P}_{\text{VADU}}^{-\frac{1}{2}}(\vecmat{W}+\tilde{\vecmat{\Sigma}}^{-1})\vecmat{P}_{\text{VADU}}^{-\frac{T}{2}} = \vecmat{I}_n$. Concerning (i), $D_{i}$ defined in Section \ref{vecchia_approx} is the predictive variance of $b_i$ conditional on its neighbors, and $D_{i}$ is thus small when the covariance function $c_{\theta}(\cdot,\cdot)$ decreases slowly with the distance among neighboring locations. In particular, $D_{i}$ usually becomes smaller with increasing $n$ when holding the domain $\mathcal{D}$ fixed. Further, $W_{ii}$ defined in Section \ref{VLA} will usually not systematically increase with $n$. For instance, for a binary likelihood with a probit and logit link, $W_{ii}\leq 1$ and $W_{ii}\leq 0.25$, respectively. We thus expect this preconditioner to work well for large $n$. On the other hand, when the covariance function $c_{\theta}(\cdot,\cdot)$ decreases quickly with the distance among neighboring locations, $D_{i}$ is usually not small, but $\vecmat{B}$ will be close to the identity $\vecmat{I}_n$. %This can be seen by noting that the Kriging weights $\vecmat{A}_i$ defined in Section \ref{vecchia_approx} will be small in this case, and by the fact that $\vecmat{B}^{-1}=\vecmat{I}_n + \sum_{k=1}^n(-1)^k(-\vecmat{A})^k$, where $-\vecmat{A}$ denotes the strictly lower part of $\vecmat{B}$.
% Fabio: for a proof of the statement for B^{-1} = I + ..., see https://math.stackexchange.com/questions/47543/getting-the-inverse-of-a-lower-upper-triangular-matrix.
%In particular, if $\vecmat{B}^{-1} = \vecmat{I}_n$, then $\vecmat{P}_{\text{VADU}}^{-\frac{1}{2}}(\vecmat{W}+\tilde{\vecmat{\Sigma}}^{-1})\vecmat{P}_{\text{VADU}}^{-\frac{T}{2}} = \vecmat{I}_n$.
%Both $\vecmat{P}_{\text{VADU}}$ and $\vecmat{P}_{\text{LVA}}$ are used with the versions in \eqref{pcls1} and \eqref{logdetsplit}.

Further, $\vecmat{P}_{\text{LVA}}^{-\frac{1}{2}}(\vecmat{W}+\tilde{\vecmat{\Sigma}}^{-1})\vecmat{P}_{\text{LVA}}^{-\frac{T}{2}} = \vecmat{D}^{\frac{1}{2}}\vecmat{B}^{-T}\vecmat{W}\vecmat{B}^{-1}\vecmat{D}^{\frac{1}{2}} + \vecmat{I}_n$ holds. Similar arguments as for $\vecmat{P}_{\text{VADU}}$ can thus be made that $\vecmat{P}_{\text{LVA}}$ is expected to perform well when the covariance function $c_{\theta}(\cdot,\cdot)$ decreases slowly with the distance among neighboring locations. However, when the covariance function $c_{\theta}(\cdot,\cdot)$ decreases quickly with the distance among neighboring locations and thus $\vecmat{B} \approx \vecmat{I}_n$, $\vecmat{P}_{\text{LVA}}$ is not expected to work as well as $\vecmat{P}_{\text{VADU}}$. This can be seen by considering $\vecmat{B} = \vecmat{I}_n$ for which $\vecmat{P}_{\text{LVA}}^{-\frac{1}{2}}(\vecmat{W}+\tilde{\vecmat{\Sigma}}^{-1})\vecmat{P}_{\text{LVA}}^{-\frac{T}{2}} = \vecmat{D}^{\frac{1}{2}}\vecmat{W}\vecmat{D}^{\frac{1}{2}} + \vecmat{I}_n$. For this reason, we generally prefer $\vecmat{P}_{\text{VADU}}$ over $\vecmat{P}_{\text{LVA}}$. In line with this, empirical experiments reported in Section \ref{CG_P_Comparison} show that $\vecmat{P}_{\text{VADU}}$ performs slightly better. 

In Section \ref{section:conv_theory}, we provide theory concerning the convergence of the preconditioned CG method and SLQ with the $\vecmat{P}_{\text{VADU}}$ and $\vecmat{P}_{\text{LVA}}$ preconditioners which support the above intuitive arguments. Calculating $\vecmat{P}_{\text{VADU}}^{-1}\vecmat{b}$ and $\vecmat{P}_{\text{LVA}}^{-1}\vecmat{b}$ has complexity $O(n)$ as it mainly consists of two triangular solves with $\vecmat{B}$, $\log\det(\vecmat{P}_{\text{VADU}}) = \log\det\left(\vecmat{W} + \vecmat{D}^{-1}\right)$ and $\log\det(\vecmat{P}_{\text{LVA}}) = \log\det\left(\vecmat{D}^{-1}\right)$,
and simulating from $\mathcal{N}(\vecmat{0},\vecmat{P}_{\text{VADU}})$ and $\mathcal{N}(\vecmat{0},\vecmat{P}_{\text{LVA}})$ has complexity $O(n)$.

\subsubsection{LRAC preconditioner}\label{sec_lrac}
A ``\textbf{l}ow-\textbf{r}ank \textbf{a}pproximation for the \textbf{c}ovariance matrix" (LRAC) preconditioner is given by
\begin{equation}\label{P_LRAC}
    \vecmat{P}_{\text{LRAC}} = \vecmat{L}_k\vecmat{L}_k^T + \vecmat{W}^{-1}, ~~ \text{where}~ \vecmat{L}_k \in \mathbb{R}^{n\times k} ~\text{and}~\tilde{\vecmat{\Sigma}} \approx \vecmat{L}_k\vecmat{L}_k^T.
\end{equation}
We use the pivoted Cholesky decomposition \citep{harbrecht2012low} to obtain $\vecmat{L}_k$ as this is a state-of-the-art approach for Gaussian likelihoods \citep{gardner2018gpytorch}. Since calculating entries of the approximated $\tilde{\vecmat{\Sigma}}$ is slow for large $n$, we use the original covariance matrix $\vecmat{\Sigma}$ to calculate $\vecmat{L}_k$. %In Appendix \ref{P_LRACAppendix}, it is shown how to calculate derivatives of the log-determinant using STE and variance reduction with the LRAC preconditioner.
%When the eigenvalues of a matrix decay sufficiently fast exponentially, the pivoted Cholesky decomposition approximates the matrix $\vecmat{\Sigma}$ in an exponentially decaying manner in $k$ for $n$ fixed \citep[Theorem 3.2]{harbrecht2012low}. 
Running $k$ iterations of the pivoted Cholesky decomposition algorithm costs $O(nk^2)$, whereby only the diagonal elements of the matrix $\vecmat{\Sigma}$ and $k$ of its rows need to be accessed. Further, for linear solves and to calculate the log-determinant, we can use the Woodbury matrix identity and the matrix determinant lemma, and the costs thus scale linearly with $n$. We have also tried adding $\text{diag}(\tilde{\vecmat{\Sigma}}) - \text{diag}(\vecmat{L}_k\vecmat{L}_k^T)$ to the diagonal of $\vecmat{P}_{\text{LRAC}}$ such that $\tilde{\vecmat{\Sigma}} + \vecmat{W}^{-1}$ and the preconditioner have the same diagonal. However, this did not improve the properties of the preconditioner (results not tabulated).

\subsubsection{ZIRC preconditioner}\label{sec_zirc}
\citet{schaefer2021sparse} apply a zero fill-in incomplete Cholesky factorization as preconditioner for solving linear systems for Vecchia approximations with Gaussian likelihoods. We consider a slight modification of this denoted as ``\textbf{z}ero fill-in \textbf{i}ncomplete \textbf{r}everse \textbf{C}holesky" (ZIRC) preconditioner and given by $\vecmat{P}_{\text{ZIRC}} = \vecmat{L}^T\vecmat{L} \approx \vecmat{W}+\tilde{\vecmat{\Sigma}}^{-1}$, where $\vecmat{L} \in \mathbb{R}^{n\times n}$ is a sparse lower triangular matrix that has the same sparsity pattern as either $\vecmat{B}$ or $\tilde{\vecmat{\Sigma}}^{-1}$. See Algorithm \ref{ZIRC_algo} in Appendix \ref{appendix:ReverseIncompChol} for a precise definition. Unfortunately, we often observe breakdowns \citep{scott2014positive}., i.e., (clearly) negative numbers in the calculations of square roots, with both sparsity patterns. For the sparsity pattern of $\tilde{\vecmat{\Sigma}}^{-1}$, breakdowns occur less frequently, but it is slower. %Such breakdowns are a known issue of incomplete Cholesky factorization schemes  %  This is applied with the versions in \eqref{pcls1} and \eqref{logdetsplit}.

\subsubsection{Other preconditioners}\label{sec_other}
We consider three additional preconditioners. First, we use a diagonal preconditioner $\vecmat{P}_1$ with the diagonal given by $(\vecmat{P}_1)_{ii} = (\vecmat{W} + \vecmat{B}^T \vecmat{D}^{-1} \vecmat{B})_{ii}$. Next, we apply a low-rank pivoted Cholesky approximation \citep{harbrecht2012low} to the precision matrix $\tilde{\vecmat{\Sigma}}^{-1} \approx \vecmat{L}_k\vecmat{L}_k^T$, $\vecmat{L}_k \in \mathbb{R}^{n\times k}$, and obtain $\vecmat{P}_2 = \vecmat{W} + \vecmat{L}_k\vecmat{L}_k^T$. However, the spectrum of precision matrices is usually flatter than that of covariance matrices and the first few largest eigenvalues are less dominant. The final preconditioner we consider is a low-rank approximation to $\tilde{\vecmat{\Sigma}}^{-1}$ using sub-matrices of $\vecmat{B}$ and $\vecmat{D}^{-1}$. It is given by $\vecmat{P}_3 = \vecmat{W} + \vecmat{B}_k^T \vecmat{D}_k^{-1} \vecmat{B}_k$, where $\vecmat{B}_k \in \mathbb{R}^{k \times n}$ is the sub-matrix of $\vecmat{B}$ that contains the $k$ rows $S_k$ for which $\sum_{i\in S_k}D_{ii}^{-1}(\sum_{j=1}^{n} B_{ij}^2)$ is maximal and $\vecmat{D}_k^{-1} \in \mathbb{R}^{k \times k}$ is the corresponding sub-matrix of $\vecmat{D}^{-1}$. One can show that $\vecmat{B}_k^T \vecmat{D}_k^{-1} \vecmat{B}_k$ is the low-rank approximation of $\tilde{\vecmat{\Sigma}}^{-1}$ for which the trace norm $\tr(\tilde{\vecmat{\Sigma}}^{-1} - \vecmat{B}_k^T \vecmat{D}_k^{-1} \vecmat{B}_k)$ is minimal among all low-rank approximations constructed by taking sub-matrices of $\vecmat{B}$ and $\vecmat{D}^{-1}$. This holds since $\tr(\vecmat{B}^T \vecmat{D}^{-1} \vecmat{B})  - \tr(\vecmat{B}_k^T \vecmat{D}_k^{-1} \vecmat{B}_k)$ is minimal when $\tr(\vecmat{B}_k^T \vecmat{D}_k^{-1} \vecmat{B}_k) = \tr(\vecmat{D}_k^{-1} \vecmat{B}_k \vecmat{B}_k^T)$ is maximal and $ \tr(\vecmat{D}_k^{-1} \vecmat{B}_k \vecmat{B}_k^T) = \sum_{i=1}^{k} (\vecmat{D}_k^{-1})_{ii}\Bigl(\sum_{j=1}^{n}(\vecmat{B}_k)_{ij}^2\Bigl) = \sum_{i \in S_k} D_{ii}^{-1}\Bigl(\sum_{j=1}^{n} B_{ij}^2\Bigl)$.

\subsection{Convergence theory} \label{section:conv_theory}
Theorem \ref{conv_VADU} provides lower bounds on the smallest eigenvalues and upper bounds on the largest eigenvalues of the preconditioned matrices $\vecmat{P}^{-\frac{1}{2}}(\vecmat{W}+\tilde{\vecmat{\Sigma}}^{-1})\vecmat{P}^{-\frac{T}{2}}$ as well as upper bounds on the convergence rate of the CG method with the VADU and LVA preconditioners.

\begin{theorem}\label{conv_VADU}
Let $\vecmat{u}_{l+l'}$ denote the approximate solution of $(\vecmat{W}+\tilde{\vecmat{\Sigma}}^{-1})\vecmat{u} = \vecmat{b}$ in iteration $(l + l')$, $l,l'\in \mathbb{N}$ , $l<n$, of the preconditioned CG method, and let $\lambda_n(\vecmat{A})\leq \dots\leq \lambda_1(\vecmat{A})$ denote the eigenvalues of a symmetric matrix $\vecmat{A}\in \mathbb{R}^{n\times n}$. The following holds:
\begin{equation*}
\begin{split}
   \lambda_n(\vecmat{P}_{\text{VADU}}^{-\frac{1}{2}}(\vecmat{W}+\tilde{\vecmat{\Sigma}}^{-1})\vecmat{P}_{\text{VADU}}^{-\frac{T}{2}}) & > \min((D_iW_{ii} + 1)^{-1}),\\
\lambda_l(\vecmat{P}_{\text{VADU}}^{-\frac{1}{2}}(\vecmat{W}+\tilde{\vecmat{\Sigma}}^{-1})\vecmat{P}_{\text{VADU}}^{-\frac{T}{2}}) &\leq (\lambda_l(\tilde{\vecmat{\Sigma}})\max(W_{ii}) + 1)\max((D_iW_{ii} + 1)^{-1}), 
\end{split}
\end{equation*}
\begin{equation}\label{CG_conv_VADU}
\|\vecmat{u}_{l+l'} - \vecmat{u}\|_{\vecmat{M}} \leq 2 \left(\frac{\sqrt{\gamma^{\text{VADU}}_l} - 1}{\sqrt{\gamma^{\text{VADU}}_l} + 1}\right)^{l'}\|\vecmat{u}_0 - \vecmat{u}\|_{\vecmat{M}} \leq 2 \left(\frac{\sqrt{\tilde \gamma^{\text{VADU}}_l} - 1}{\sqrt{\tilde\gamma^{\text{VADU}}_l} + 1}\right)^{l'}\|\vecmat{u}_0 - \vecmat{u}\|_{\vecmat{M}},
\end{equation}
where 
\begin{equation*}
\begin{split}
\gamma^{\text{VADU}}_l &= \frac{(\lambda_{l+1}(\tilde{\vecmat{\Sigma}})\max(W_{ii}) + 1)\max((D_iW_{ii} + 1)^{-1})}{(\lambda_n(\tilde{\vecmat{\Sigma}})\min(W_{ii}) + 1)\min((D_iW_{ii} + 1)^{-1})}, ~~
\tilde\gamma^{\text{VADU}}_l = \frac{\lambda_{l+1}(\tilde{\vecmat{\Sigma}})\max(W_{ii}) + 1}{\min((D_iW_{ii} + 1)^{-1})},
\end{split}
\end{equation*}
and $\vecmat{M}= \vecmat{P}_{\text{VADU}}^{-\frac{1}{2}}(\vecmat{W}+\tilde{\vecmat{\Sigma}}^{-1})\vecmat{P}_{\text{VADU}}^{-\frac{T}{2}}$. For the LVA preconditioner, the following holds:
\begin{equation*}
\lambda_n(\vecmat{P}_{\text{LVA}}^{-\frac{1}{2}}(\vecmat{W}+\tilde{\vecmat{\Sigma}}^{-1})\vecmat{P}_{\text{LVA}}^{-\frac{T}{2}}) > 1,~~
\lambda_l(\vecmat{P}_{\text{LVA}}^{-\frac{1}{2}}(\vecmat{W}+\tilde{\vecmat{\Sigma}}^{-1})\vecmat{P}_{\text{LVA}}^{-\frac{T}{2}}) \leq \lambda_l(\tilde{\vecmat{\Sigma}})\max(W_{ii}) + 1,
\end{equation*}
\begin{equation}\label{CG_conv_LVA}
\|\vecmat{u}_{l+l'} - \vecmat{u}\|_{\vecmat{M}} \leq 2 \left(\frac{\sqrt{\gamma^{\text{LVA}}_l} - 1}{\sqrt{\gamma^{\text{LVA}}_l} + 1}\right)^{l'}\|\vecmat{u}_0 - \vecmat{u}\|_{\vecmat{M}} \leq 2 \left(\frac{\sqrt{\tilde \gamma^{\text{LVA}}_l} - 1}{\sqrt{\tilde\gamma^{\text{LVA}}_l} + 1}\right)^{l'}\|\vecmat{u}_0 - \vecmat{u}\|_{\vecmat{M}},
\end{equation}
where 
\begin{equation*}
\gamma^{\text{LVA}}_l = \frac{\lambda_{l+1}(\tilde{\vecmat{\Sigma}})\max(W_{ii}) + 1}{\lambda_n(\tilde{\vecmat{\Sigma}})\min(W_{ii}) + 1},~~
\tilde\gamma^{\text{LVA}}_l = \lambda_{l+1}(\tilde{\vecmat{\Sigma}})\max(W_{ii}) + 1,
\end{equation*}
and $\vecmat{M}= \vecmat{P}_{\text{LVA}}^{-\frac{1}{2}}(\vecmat{W}+\tilde{\vecmat{\Sigma}}^{-1})\vecmat{P}_{\text{LVA}}^{-\frac{T}{2}}$.
\end{theorem}

A proof of Theorem \ref{conv_VADU} can be found in Appendix \ref{proofs}. Note that $\gamma^{\text{VADU}}_1$ and $\gamma^{\text{LVA}}_1$ are upper bounds on the condition number of the preconditioned matrices $\vecmat{P}_{\text{VADU}}^{-\frac{1}{2}}(\vecmat{W}+\tilde{\vecmat{\Sigma}}^{-1})\vecmat{P}_{\text{VADU}}^{-\frac{T}{2}}$ and $\vecmat{P}_{\text{LVA}}^{-\frac{1}{2}}(\vecmat{W}+\tilde{\vecmat{\Sigma}}^{-1})\vecmat{P}_{\text{LVA}}^{-\frac{T}{2}}$, respectively. This theorem provides further arguments, in addition to those made in Section \ref{sec_vadu}, that both $\vecmat{P}_{\text{VADU}}$ and $\vecmat{P}_{\text{LVA}}$ are expected to perform well. The spectrum of covariance matrices $\tilde{\vecmat{\Sigma}}$ often contains relatively few large eigenvalues and many small ones, and $\lambda_{l+1}(\tilde{\vecmat{\Sigma}})$ is often small already for small $l$. Thus, if after $l$ iterations of the preconditioned CG method, the largest eigenvalues of $\vecmat{P}^{-\frac{1}{2}}(\vecmat{W}+\tilde{\vecmat{\Sigma}}^{-1})\vecmat{P}^{-\frac{T}{2}}$ are ``removed" and the ``residual condition number" $(\lambda_{l+1}(\vecmat{P}^{-\frac{1}{2}}(\vecmat{W}+\tilde{\vecmat{\Sigma}}^{-1})\vecmat{P}^{-\frac{T}{2}}))/(\lambda_n(\vecmat{P}^{-\frac{1}{2}}(\vecmat{W}+\tilde{\vecmat{\Sigma}}^{-1})\vecmat{P}^{-\frac{T}{2}}))\leq\gamma^{\text{VADU}/\text{LVA}}_l$ is small, the convergence becomes fast. In addition, Theorem \ref{conv_VADU} shows that the smallest eigenvalues of the preconditioned matrices are larger than $ \min((D_iW_{ii} + 1)^{-1})$ and $1$ for $\vecmat{P}_{\text{VADU}}$ and $\vecmat{P}_{\text{LVA}}$, respectively. In essence, the theorem thus shows that these preconditioners remove the small eigenvalues of $\vecmat{W}+\tilde{\vecmat{\Sigma}}^{-1}$ which can be very small and close to zero, in particular, for large $n$. This is beneficial because, first, the convergence of CG algorithms is more delayed by small eigenvalues compared to large ones \citep{nishimura2022prior} and because more small eigenvalues are typically added to $\vecmat{W}+\tilde{\vecmat{\Sigma}}^{-1}$ when $n$ increases. Further, we recall that $D_i$ decreases towards $0$ when $n$ goes to infinity under an in-fill asymptotic regime, and the difference in the convergence rates for $\vecmat{P}_{\text{VADU}}$ and $\vecmat{P}_{\text{LVA}}$ thus vanishes.

The following theorem provides stochastic additive error bounds for the SLQ approximation in \eqref{def_SLQ} with the VADU and LVA preconditionrers as well as a multiplicative error bound for $\vecmat{P}_{\text{LVA}}$. A proof of Theorem \ref{acc_SLQ} can be found in Appendix \ref{proofs}.

\begin{theorem}\label{acc_SLQ}
    Let $\kappa$ and $\lambda_n$ denote the condition number and smallest eigenvalue of $\vecmat{P}^{-\frac{1}{2}}(\vecmat{W}+\tilde{\vecmat{\Sigma}}^{-1})\vecmat{P}^{-\frac{T}{2}}$, $C_{nt} = \frac{Q_{\chi^2_{nt}}({1-\eta/2})}{nt}$, where $Q_{\chi^2_{nt}}(\cdot)$ is the quantile function of a $\chi^2-$distribution with $nt$ degrees of freedom, and $\widehat \Gamma = \frac{1}{t} \sum_{i=1}^t\|\vecmat{P}^{-\frac{1}{2}}\vecmat{z}_i\|_2^2 \vecmat{e}_1^T \log(\tilde{\vecmat{T}}_i) \vecmat{e}_1$ an SLQ approximation. If the SLQ method is run with $l\geq l(\kappa) = \frac{\sqrt{3\kappa}}{4}\log\left( \frac{C_{nt}20\log(2(\kappa+1))\sqrt{2\kappa+1}}{\epsilon}\right)$ preconditioned CG steps and $t\geq t(\kappa) =  \frac{32}{\epsilon^2} \log(\kappa + 1)^2\log\left(\frac{4}{\eta}\right)$ number of random vectors, the following holds for any preconditioner $\vecmat{P}$:
\begin{equation}\label{bound_slq}
    P(|\widehat \Gamma - \log\det(\vecmat{P}^{-\frac{1}{2}}(\vecmat{W}+\tilde{\vecmat{\Sigma}}^{-1})\vecmat{P}^{-\frac{T}{2}})|\leq \epsilon n) \geq 1 - \eta,
\end{equation}
 where $\epsilon, \eta \in (0,1)$. The bound \eqref{bound_slq} holds for the VADU and LVA preconditioners also if $l\geq l(\gamma^{\text{VADU}}_1)$ and $t\geq t(\gamma^{\text{VADU}}_1)$ and $l\geq l(\gamma^{\text{LVA}}_1)$ and $t\geq t(\gamma^{\text{LVA}}_1)$, respectively. 

    In addition, if $\lambda_n(\vecmat{P}^{-\frac{1}{2}}(\vecmat{W}+\tilde{\vecmat{\Sigma}}^{-1})\vecmat{P}^{-\frac{T}{2}})>1$, e.g., for $\vecmat{P}=\vecmat{P}_{\text{LVA}}$, $l\geq l^m(\kappa)$, where $ l^m(\kappa) = \frac{1}{2}\log\left(\frac{4C_{nt}(\log\left(\lambda_n(\kappa + 1 -1/\kappa)\right)+\pi)(\sqrt{2\kappa+1}+1)}{\log(\lambda_n)\epsilon}\right)/\log\left(\frac{\sqrt{2\kappa+1}+1}{\sqrt{2\kappa+1}-1}\right)$, and $t\geq \frac{32}{\epsilon^2}\log\left(\frac{4}{\eta}\right)$, then
    \begin{equation}\label{mult_bound_slq}
    P(|\widehat \Gamma - \log\det(\vecmat{P}^{-\frac{1}{2}}(\vecmat{W}+\tilde{\vecmat{\Sigma}}^{-1})\vecmat{P}^{-\frac{T}{2}})|\leq \epsilon\log\det(\vecmat{P}^{-\frac{1}{2}}(\vecmat{W}+\tilde{\vecmat{\Sigma}}^{-1})\vecmat{P}^{-\frac{T}{2}}) ) \geq 1 - \eta.
    \end{equation}
    For the LVA preconditioner, the bound \eqref{mult_bound_slq} holds also if $l\geq l^m(\gamma^{\text{LVA}}_1)$.
\end{theorem}

\subsection{Predictive covariance matrices} \label{section:iterativePred}
The computational bottleneck for prediction is the calculation of the posterior predictive (co-)variances $\vecmat{\Omega_p}$ in \eqref{latentpostpredictive}. These are required for predictive distributions of the latent and response variables and also for predictive means of response variables. Calculating $\vecmat{\Omega_p}$ requires solving $n_p$ linear equation systems with the matrix $\vecmat{W}+\tilde{\vecmat{\Sigma}}^{-1}$ and $\vecmat{B}_{po}^T\vecmat{B}_p^{-1}\in \mathbb{R}^{n \times n_p}$ as right-hand side. For large $n_p$ and $n$, solving these linear systems with the CG method can thus be computationally prohibitive. We propose two solutions for this: a simulation-based method and an approach based on preconditioned Lanczos algorithms. %Experiments in Section \ref{section:LanczosVSsim} compare these two prediction methods and demonstrate that the simulation-based approach provides more accurate predictive variances for a given runtime compared to the Lanczos-based method. Further, we show that simulation leads to unbiased variance estimates while the Lanczos-based approach can systematically underestimate the predictive variances. 

\subsubsection{Predictive (co-)variances using simulation} \label{section:simAlgo}
The predictive covariance matrix $\vecmat{\Omega_p}$ can be approximated using simulation with Algorithm \ref{simAlgo}. This algorithm results in an unbiased and consistent approximation for $\vecmat{\Omega_p}$, see Appendix \ref{proofs} for a proof of Proposition \ref{pred_var_sim}. The linear solves $(\vecmat{W}+\tilde{\vecmat{\Sigma}}^{-1})^{-1} \vecmat{z}_i^{(3)}$ can be done using the preconditioned CG method. The computational complexity of Algorithm \ref{simAlgo} is thus $O(sln + sn_p)$, where $l$ denotes the number of CG iterations and $s$ the number of simulation iterations. Furthermore, the algorithm can be trivially parallelized.% Also note that in Algorithm \ref{simAlgo}, we use Gaussian random vectors, but other distributions (e.g., Rademacher) could also be used.

\begin{algorithm}[ht!]
    \caption{Approximate predictive covariance matrices using simulation}\label{simAlgo}
    \begin{algorithmic}[1]
        \Require Matrices $\vecmat{B}, \vecmat{D}^{-1}, \vecmat{W}, \vecmat{B}_p^{-1}, \vecmat{B}_{po}, \vecmat{D}_p$
        \Ensure Approximate predictive covariance matrix $\hat{\vecmat{\Omega}}_p$
        \For{$i \gets 1$ to $s$}
            \State{Sample $\vecmat{z}_i^{(1)}, \vecmat{z}_i^{(2)} \overset{\text{i.i.d.}}{\sim} \mathcal{N}(\vecmat{0}, \vecmat{I}_n)$ and set $\vecmat{z}_i^{(3)} \gets \vecmat{W}^{\frac{1}{2}}\vecmat{z}_i^{(1)} +\vecmat{B}^T\vecmat{D}^{-\frac{1}{2}}\vecmat{z}_i^{(2)}$}
            \State{$\vecmat{z}_i^{(4)} \gets \vecmat{B}_p^{-1}\vecmat{B}_{po}(\vecmat{W}+\tilde{\vecmat{\Sigma}}^{-1})^{-1} \vecmat{z}_i^{(3)}$}
        \EndFor
        \State{$\hat{\vecmat{\Omega}}_p\gets \vecmat{B}_p^{-1}\vecmat{D}_p\vecmat{B}_p^{-T} + \frac{1}{s}\sum_{i=1}^s\vecmat{z}_i^{(4)}(\vecmat{z}_i^{(4)})^T$}
    \end{algorithmic}
\end{algorithm}

\begin{proposition}\label{pred_var_sim}
Algorithm \ref{simAlgo} produces an unbiased and consistent estimator $\hat{\vecmat{\Omega}}_p$ for the predictive covariance matrix $\vecmat{\Omega_p}$ given in \eqref{latentpostpredictive}.
\end{proposition}

\subsubsection{Predictive (co-)variances using preconditioned Lanczos algorithms} \label{section:predLanczos}
\citet{pleiss2018constant} use the Lanczos algorithm to approximate predictive covariance matrices for Gaussian process regression. This is currently a state-of-the-art approach in machine learning. The Lanczos algorithm factorizes a symmetric matrix $\vecmat{A} = \vecmat{Q}\vecmat{T}\vecmat{Q}^T$, where $\vecmat{Q} \in \mathbb{R}^{n\times n}$ is orthonormal and $\vecmat{T}\in \mathbb{R}^{n\times n}$ is a tridiagonal matrix. This decomposition is computed iteratively, and $k$ iterations produce an approximation $\vecmat{A} \approx \tilde{\vecmat{Q}}\tilde{\vecmat{T}} \tilde{\vecmat{Q}}^{T}$, where $\tilde{\vecmat{Q}} \in \mathbb{R}^{n\times k}$ contains the first $k$ columns of $\vecmat{Q}$ and $\tilde{\vecmat{T}} \in \mathbb{R}^{k\times k}$ the corresponding coefficients of $\vecmat{T}$. In the following, we adapt and extend the approach of \citet{pleiss2018constant} to Vecchia-Laplace approximations. We run $k$ steps of the Lanczos algorithm with the matrix $\vecmat{W}+\tilde{\vecmat{\Sigma}}^{-1}$ using the average of the column vectors of $\vecmat{B}_{po}^T\vecmat{B}_p^{-1}$ as initial value and approximate $(\vecmat{W}+\tilde{\vecmat{\Sigma}}^{-1})^{-1}\vecmat{B}_{po}^T\vecmat{B}_p^{-1} \approx \tilde{\vecmat{Q}} \tilde{\vecmat{T}}^{-1} \tilde{\vecmat{Q}}^{T}\vecmat{B}_{po}^T\vecmat{B}_p^{-1}$. This gives the approximation $\vecmat{\Omega}_p \approx \vecmat{B}_p^{-1}\vecmat{D}_p\vecmat{B}_p^{-T} + \vecmat{B}_p^{-1}\vecmat{B}_{po}\tilde{\vecmat{Q}} \tilde{\vecmat{T}}^{-1} \tilde{\vecmat{Q}}^{T}\vecmat{B}_{po}^T\vecmat{B}_p^{-1}$. Computing a rank $k$ Lanczos approximation requires one matrix-vector multiplication with $\vecmat{W}+\tilde{\vecmat{\Sigma}}^{-1}$ in each iteration. Since the Lanczos algorithm suffers from numerical stability issues due to loss of orthogonality, we use a full reorthogonalization scheme \citep{wu2000thick}. Computing the Lanczos approximation then has complexity $O(k^2n)$.

Let $\lambda_n\leq \dots\leq \lambda_1$ denote the eigenvalues of a symmetric matrix $\vecmat{A}\in \mathbb{R}^{n\times n}$, here $\vecmat{A} = \vecmat{W}+\tilde{\vecmat{\Sigma}}^{-1}$. The inverse function $(\vecmat{A})^{-1}$ is mostly determined by the small eigenvalues of $\vecmat{A}$. It is known that the Lanczos algorithm approximates the smallest eigenvalue $\lambda_n$ the faster, the smaller $\lambda_1 - \lambda_{n}$ and the larger $\frac{\lambda_{n-1} - \lambda_{n}}{\lambda_1 - \lambda_{n-1}}$ \citep{golub2013matrix}. Similar results can also be obtained for other eigenvalues. This suggests to apply the following way of preconditioning. Choose a symmetric preconditioner matrix $\vecmat{P}$ such that for $\vecmat{P}^{-\frac{1}{2}}\vecmat{A}\vecmat{P}^{-\frac{T}{2}}$, $\lambda_1 - \lambda_{n}$ is smaller and $\frac{\lambda_{n-1} - \lambda_{n}}{\lambda_1 - \lambda_{n-1}}$ is larger compared to $\vecmat{A}$. Then, apply the partial Lanczos algorithm to $\vecmat{P}^{-\frac{1}{2}}\vecmat{A}\vecmat{P}^{-\frac{T}{2}} \approx \tilde{\vecmat{Q}}\tilde{\vecmat{T}} \tilde{\vecmat{Q}}^{T}$ and correct for this by using $\vecmat{P}^{-\frac{T}{2}}\tilde{\vecmat{Q}} \tilde{\vecmat{T}}^{-1} \tilde{\vecmat{Q}}^{T}\vecmat{P}^{-\frac{1}{2}}\vecmat{B}_{po}^T\vecmat{B}_p^{-1} \approx  \vecmat{A}^{-1}\vecmat{B}_{po}^T\vecmat{B}_p^{-1}$ as approximation. Note that the Lanczos algorithm requires a symmetric matrix and in contrast to the preconditioned CG method, we need to explicitly calculate a factor $\vecmat{P}^{-\frac{1}{2}}$, where $\vecmat{P}=\vecmat{P}^{\frac{1}{2}}\vecmat{P}^{\frac{T}{2}}$. Below, we suggest two preconditioners. Experiments in Section \ref{section:LanczosVSsim} show that this form of preconditioning leads to more accurate predictive variance approximations.

\paragraph{Lanczos preconditioner 1}

We consider the preconditioner $\vecmat{P}_{L1} = \vecmat{B}^T\left(\vecmat{W} + \vecmat{D}^{-1}\right)\vecmat{B}$. Note that the matrix $\vecmat{P}_{L1}^{-\frac{1}{2}}\left(\vecmat{W} + \tilde{\vecmat{\Sigma}}^{-1}\right)\vecmat{P}_{L1}^{-\frac{T}{2}}$, which is factorized, is given in \eqref{P_VADU_prec}.

\paragraph{Lanczos preconditioner 2}

We write $\vecmat{W}+\tilde{\vecmat{\Sigma}}^{-1}=\tilde{\vecmat{\Sigma}}^{-1}(\tilde{\vecmat{\Sigma}} + \vecmat{W}^{-1})\vecmat{W}$ and use the Lanczos algorithm to approximate $\vecmat{P}_{L2}^{-\frac{1}{2}}(\tilde{\vecmat{\Sigma}} + \vecmat{W}^{-1})\vecmat{P}_{L2}^{-\frac{T}{2}} \approx \tilde{\vecmat{Q}} \tilde{\vecmat{T}} \tilde{\vecmat{Q}}^T$, where $\vecmat{P}_{L2} \in \mathbb{R}^{n\times n}$ is a diagonal preconditioner with entries $(\vecmat{P}_{L2})_{ii} = \left(\tilde{\vecmat{\Sigma}}+\vecmat{W}^{-1}\right)_{ii}$. We then use the approximation $\left(\vecmat{W}+\tilde{\vecmat{\Sigma}}^{-1}\right)^{-1}\vecmat{B}_{po}^T\vecmat{B}_p^{-1} \approx \vecmat{W}^{-1}\vecmat{P}_{L2}^{-\frac{T}{2}}\tilde{\vecmat{Q}} \tilde{\vecmat{T}}^{-1} \tilde{\vecmat{Q}}^T\vecmat{P}_{L2}^{-\frac{1}{2}}\tilde{\vecmat{\Sigma}} \vecmat{B}_{po}^T\vecmat{B}_p^{-1}$.

\subsection{Software implementation}\label{software}
The iterative methods presented in this article are implemented in the \if1\blind{\texttt{GPBoost} }\else{\#hidden\# }\fi library written in C++ with Python and R interface packages\if1\blind{, see \url{https://github.com/fabsig/GPBoost}}\fi.\footnote{Iterative methods can be enabled via the parameter \texttt{matrix\_inversion\_method = "iterative"} and  the preconditioner is chosen via the parameter \texttt{cg\_preconditioner\_type}.} For linear algebra calculations, we use the \texttt{Eigen} library version 3.4.99 and its sparse matrix algebra operations whenever possible. Multi-processor parallelization is done using \texttt{OpenMP}. %For the iterative methods, we use sparse matrices in row-major format in order that matrix multiplications can be parallelized, which is not possible with column-major sparse matrices in \texttt{Eigen}. For Cholesky-based calculations, we use sparse matrices in the default column-major format.

\section{Experiments using simulated data}
In the following, we conduct several experiments using simulated data. First, we compare different preconditioners and methods for calculating predictive variances. We then analyze the proposed iterative methods concerning the accuracy for parameter estimation, prediction, likelihood evaluation, and runtime. Among other things, we compare this to calculations based on the Cholesky decomposition. Note, however, that the Cholesky decomposition is also not exact since round-off errors can accumulate when using finite precision arithmetic. Code to reproduce the simulated and the real-world data experiments of this article is available at \url{https://github.com/pkuendig/iterativeVL}.% For simplicity, we pretend that the Cholesky decomposition is exact in situations when there is no alternative.

\subsection{Experimental setting}\label{simSetting}
We simulate from a zero-mean GP with a Matérn covariance function with smoothness parameter $\nu=1.5$, $c_{\theta}(\vecmat{s},\vecmat{s}')=\sigma_1^2 \left(1 +\frac{\sqrt{3}||\vecmat{s}-\vecmat{s}'||}{\rho}\right)\exp{\left(-\frac{\sqrt{3}||\vecmat{s}-\vecmat{s}'||}{\rho}\right)}$, marginal variance $\sigma_1^2=1$, range $\rho=0.05$, and the locations $\vecmat{s}$ are sampled uniformly from $[0,1]^2$. The response variable $y \in \{0,1\}$ follows a Bernoulli likelihood with a logit link function. 

We measure both the accuracy of point and probabilistic predictions for the latent process $\vecmat{b}$. Point predictions in the form of predictive means are evaluated using the root mean squared error (RMSE), and for probabilistic predictions, we use the log score (LS) $-\sum_{i=1}^{n_p}\log\left(\mathcal{N}\left(b_{p,i};\omega_{p,i}, (\vecmat{\Omega}_p)_{ii}\right)\right)$, where $b_{p,i}$ is the simulated Gaussian process at the test locations, $\omega_{p,i}$ and $(\vecmat{\Omega}_p)_{ii}$ are the predictive latent means and variances given in \eqref{latentpostpredictive}, and $\mathcal{N}\left(b_{p,i};\omega_{p,i}, (\vecmat{\Omega}_p)_{ii}\right)$ denotes a Gaussian density evaluated at $b_{p,i}$. We also measure the runtime in seconds. All calculations are done on a laptop with an Intel i7-12800H processor and 32 GB of random-access memory. Calculations with Vecchia-Laplace approximations based on the Cholesky decomposition (`Cholesky-VL') and iterative methods (`Iterative-VL') and Laplace approximations without a Vecchia approximation are done using the \if1\blind{\texttt{GPBoost} }\else{\#hidden\# }\fi library version 1.4.0 compiled with the MinGW64 compiler version 10.0.0. In addition, we compare our methodology to the approach of \citet{zilber2021vecchia} (`GPVecchia') using the \texttt{R} package \texttt{GPVecchia} version 0.1.7.

We use a random ordering for the Vecchia approximation in Iterative-VL and Cholesky-VL since this often gives very good approximations \citep{guinness2018permutation}. For prediction, the observed locations appear first in the ordering, and the prediction locations are only conditioned on the observed locations. We could not detect any differences in our results when also conditioning on the prediction locations. For GPVecchia, we use the maxmin ordering in combination with the response-first ordering and full conditioning (``RF-full") for prediction as recommended by \citet{zilber2021vecchia} and \citet{katzfuss2020vecchia}. If not noted otherwise, we use $m=20$ number of neighbors for training and prediction.

For the iterative methods, if not stated otherwise, we use the VADU preconditioner given in \eqref{P_VADU}, set the number of random vectors for STE and SLQ to $t=50$, and use a rank of $k=50$ for the LRAC preconditioner. For the CG algorithm, we use a convergence tolerance of $10^{-2}$ for calculating marginal likelihoods and gradients, and a tolerance of $10^{-3}$ for predictive variances. Further decreasing the CG tolerance does not change the results. For predictive distributions, we use simulation-based predictive variances with $s=2000$ samples unless stated otherwise. Further, we adopt a sample average approximation \citep{kim2015guide} when maximizing the marginal likelihood. As optimization method, we use the L-BFGS algorithm, except for GPVecchia, for which we follow \citet{zilber2021vecchia} and use the Nelder-Mead algorithm provided in the \texttt{R} function \texttt{optim}. We have also tried using the BFGS algorithm in \texttt{optim} for GPVecchia, but this often crashed due to singular matrices. The calculation of gradients is not supported in \texttt{GPVecchia}. 
%In our preconditioner comparison in Section \ref{CG_P_Comparison}, the VADU preconditioner achieves the fastest runtime, and $t=50$ random vectors correspond to a good tradeoff between accuracy and runtime.

\subsection{Preconditioner comparison}\label{CG_P_Comparison}
First, we compare the proposed preconditioners with regard to the runtime and the variance of log-marginal likelihoods. We consider varying numbers of random vectors $t$ for the SLQ approximation. For each preconditioner and number of random vectors, the calculation of the likelihood is done at the true covariance parameters and repeated $100$ times with different random vectors. Since with some preconditioners computations are rather slow as many CG iterations are required, we use a relatively small sample size of $n=5'000$. Figure \ref{fig:P_comparison_CG_5000} reports the results. The runtime reported is the average wall-clock time. The VADU, LRAC, and LVA preconditioners clearly have the fastest runtime and a small variance. The ZIRC preconditioner has an even smaller variance but a longer runtime than the VADU, LRAC, and LVA preconditioners. For the ZIRC preconditioner, we use the sparsity pattern of $\tilde{\vecmat{\Sigma}}^{-1}$ since breakdowns occur when using the sparsity pattern of $\vecmat{B}$. For larger data sets, we also observe breakdowns when using the sparsity pattern of $\tilde{\vecmat{\Sigma}}^{-1}$. Note that the runtime of the ZIRC preconditioner is mainly determined by the construction of the preconditioner and not by the CG method itself, since this preconditioner is quite accurate in our experiments. The three remaining preconditioners have much larger runtimes and variances. In Figure \ref{fig:P_comparison_CG_100000} in Appendix \ref{appendix:add_res_sim}, we additionally report the results when repeating the experiment using a large sample size of $n=100'000$ for the VADU and LRAC preconditioners. We observe only minor differences between these two preconditioners. In Figures \ref{fig:P_comparison_CG_rho=0.01} and \ref{fig:P_comparison_CG_rho=0.25} in Appendix \ref{appendix:add_res_sim}, we also report analogous results when using a small and large range of $\rho=0.01$ and $\rho=0.25$, respectively. We find that the VADU preconditioner outperforms the LRAC preconditioner for the small range, while the LRAC preconditioner outperforms the VADU preconditioner for the large range. In contrast to the VADU, LVA, and ZIRC preconditioners, the LRAC has a tuning parameter corresponding to the rank $k$. \citet{maddox2021iterative} recommend a rank of $k=50$. In Figure \ref{fig:P_LRAC_ranks} in Appendix \ref{appendix:add_res_sim}, we briefly analyze the impact of the rank $k$. In line with \citet{maddox2021iterative}, we find that $k=50$ is a good choice. Note that the experiments in this subsection and the following Section \ref{section:LanczosVSsim} are conducted on only one simulated data set since we do not want to mix sampling variability and variability of the simulation-based approximations. However, the results do not change when using other samples (results not reported).

\begin{figure}[ht!]
\centering
    \includegraphics[width=0.7\linewidth]{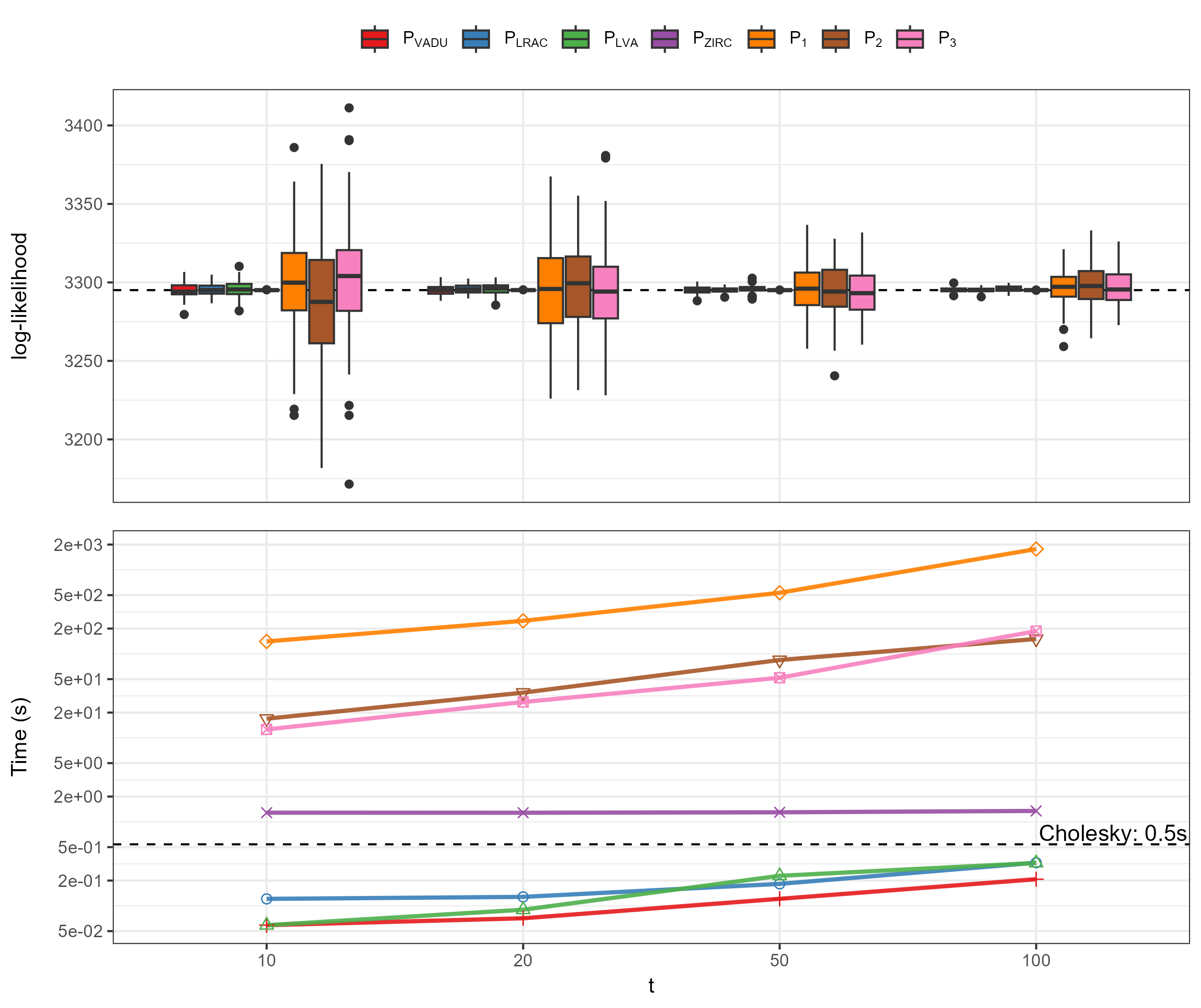}
    \caption{Negative log-marginal likelihood and runtime for different preconditioners and numbers of random vectors $t$. The dashed lines are the results for the Cholesky decomposition.}
    \label{fig:P_comparison_CG_5000}
\end{figure}

\subsection{Comparing simulation and Lanczos-based predictive variances} \label{section:LanczosVSsim}
Next, we compare the accuracy of the different approximations for predictive variances introduced in Section \ref{section:iterativePred} using simulated data with $n=n_p=100'000$ training and test points. We calculate the RMSE of the predictive variances compared to the ``exact" Cholesky-based results, and we measure the runtime for prediction, which includes the calculation of the mode, the latent predictive means, and the latent predictive variances. We use the true covariance parameters to calculate predictive distributions.

\begin{figure}[ht!]
\centering
    \includegraphics[width=0.7\linewidth]{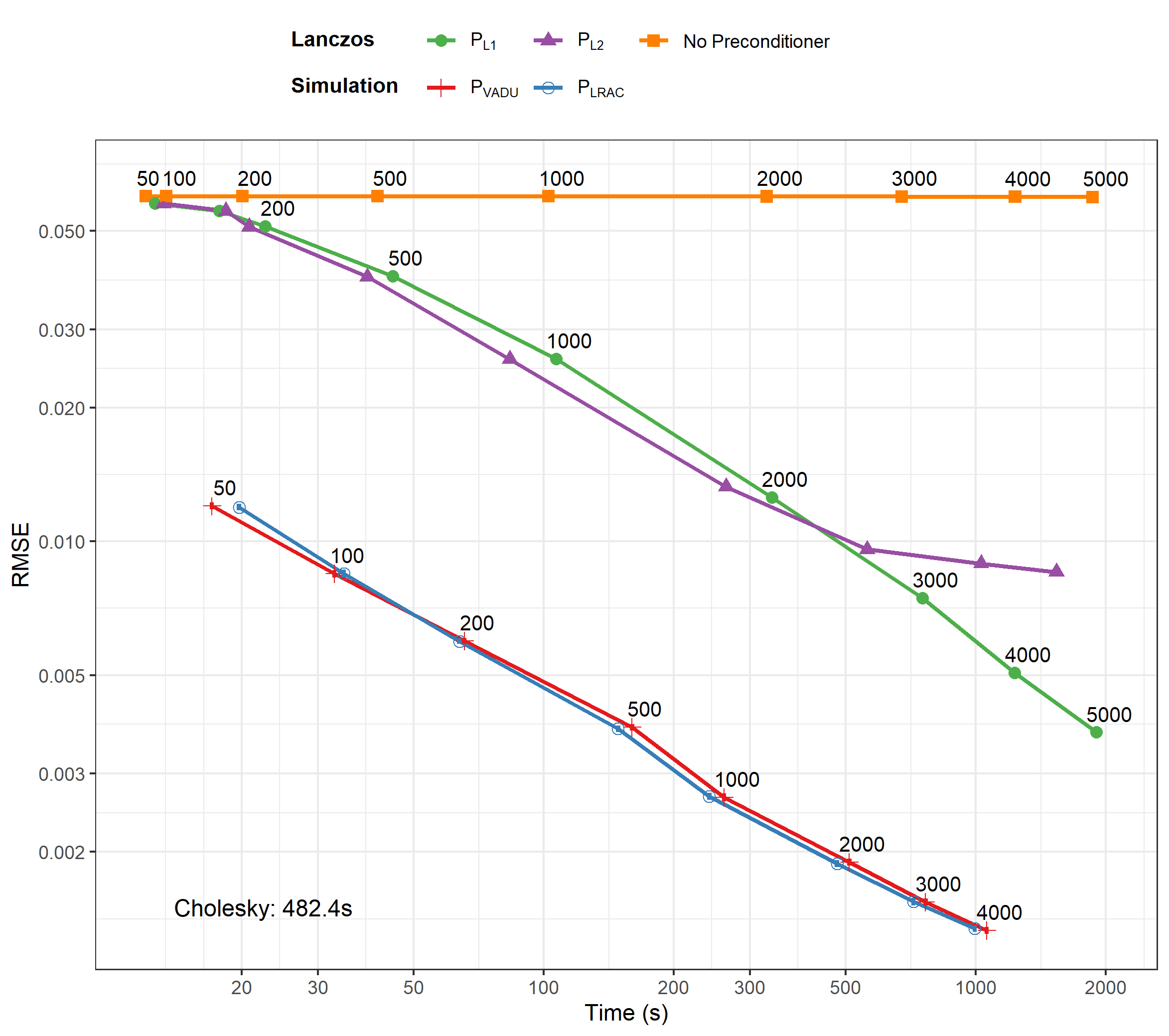}
    \caption{Comparison of simulation- and Lanczos-based methods for predictive variances. The number of random vectors $s$ and the Lanczos rank $k$ are annotated in the plot.}
    \label{fig:Lanczos_vs_simulation1e5}
\end{figure}

Figure \ref{fig:Lanczos_vs_simulation1e5} shows the RMSE versus the wall-clock time for different ranks $k$ for the Lanczos-based approximations and different numbers of random vectors $s$ for the simulation-based methods. For the latter, we average the results over $10$ independent repetitions for every $s$ and add whiskers to the plot representing confidence intervals for the RMSE obtained as $\pm 2 \times$ standard errors, but the corresponding whiskers are not visible since they are very small. Figure \ref{fig:sim1e5} in Appendix \ref{appendix:add_res_sim_var} reports box-plots of the RMSE of the stochastic approximation versus different $s$'s to visualize the variability. The simulation-based methods are often faster by more than a factor of ten for the same accuracy compared to the Lanczos-based approximations. Further, as the rank of the Lanczos decomposition increases, the RMSE decreases faster when a preconditioner $\vecmat{P}_{L1}$ or $\vecmat{P}_{L2}$ is used compared to when no preconditioning is applied. For the simulation-based predictions, the choice of the preconditioner for the CG method does not play a major role. In Figure \ref{fig:latentpostpredictivepath} in Appendix \ref{appendix:add_res_sim_var}, we additionally show the predictive variance $(\vecmat{\Omega}_p)_{i,i}$ for a single data point for simulation- and Lanczos-based approximations with different numbers of random vectors $s$ and Lanczos ranks $k$. The Lanczos-based predictions systematically underestimate the variance, but the bias decreases as $k$ increases. As expected, simulation-based variances are unbiased. Very similar results are obtained for other prediction points (results not reported). Figure \ref{fig:Lanczos_vs_simulation5000} in Appendix \ref{appendix:add_res_sim_var} also reports qualitatively very similar results for a smaller sample size of $n=n_p=5000$.

\subsection{Parameter estimation}\label{param_est}
In the following, we analyze the properties of covariance parameter estimates obtained as maximizers of the log-marginal likelihood. Estimation is repeated $100$ times on data sets with a size of $n=20'000$. This is, approximately, the largest $n$ such that we can run GPVecchia and Cholesky-based calculations in a reasonable amount of time. Figure \ref{fig:estimates20000} reports the estimates for the marginal variance and range parameter and Table \ref{table:estimates20000} in Appendix \ref{appendix:add_res_sim_est} reports the RMSE and bias. Comparing Cholesky-based and iterative methods-based estimates, we observe almost no differences. GPVecchia, on the other hand, produces estimates with large upward biases and large RMSEs. 
\begin{figure}[ht!]
    \centering
    \includegraphics[width=\linewidth]{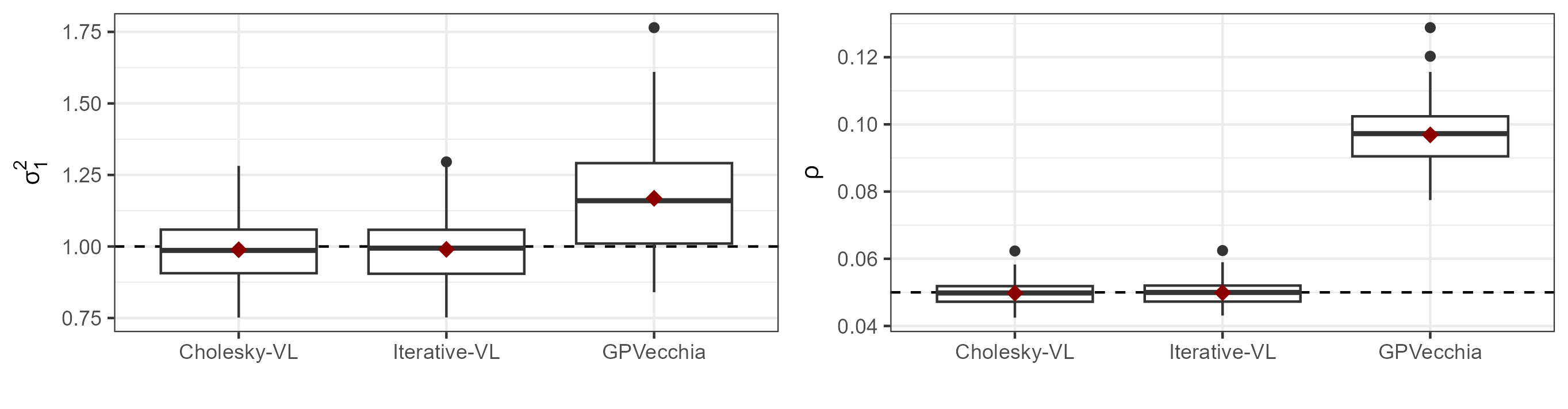}
    \caption{Estimated marginal variance $\sigma_1^2$ and range $\rho$ parameter. The red rhombi represent means. The dotted lines indicate the true values.}
    \label{fig:estimates20000}
\end{figure}

\subsection{Prediction accuracy}
Next, we analyze the prediction accuracy using simulated training and test data with $n=n_p=20'000$ and $100$ simulation repetitions. Predictive distributions are calculated using estimated covariance parameters. Figure \ref{fig:predictionRMSE20000} and Table \ref{table:prediction20000} in Appendix \ref{appendix:add_res_sim_est} report the results. We observe that Cholesky- and iterative methods-based predictions have virtually the same accuracy in terms of both the RMSE and the log score. Predictive means obtained with GPVecchia are less accurate. Further, GPVecchia produces negative predictive variances in each simulation run, and we could thus not calculate the log score. In Section \ref{varying_m}, we additionally analyze predictions for other sample sizes and numbers of neighbors using true covariance parameters instead of estimated ones, and we obtain very similar results.
\begin{figure}[ht!]
    \centering
    \includegraphics[width=0.55\linewidth]{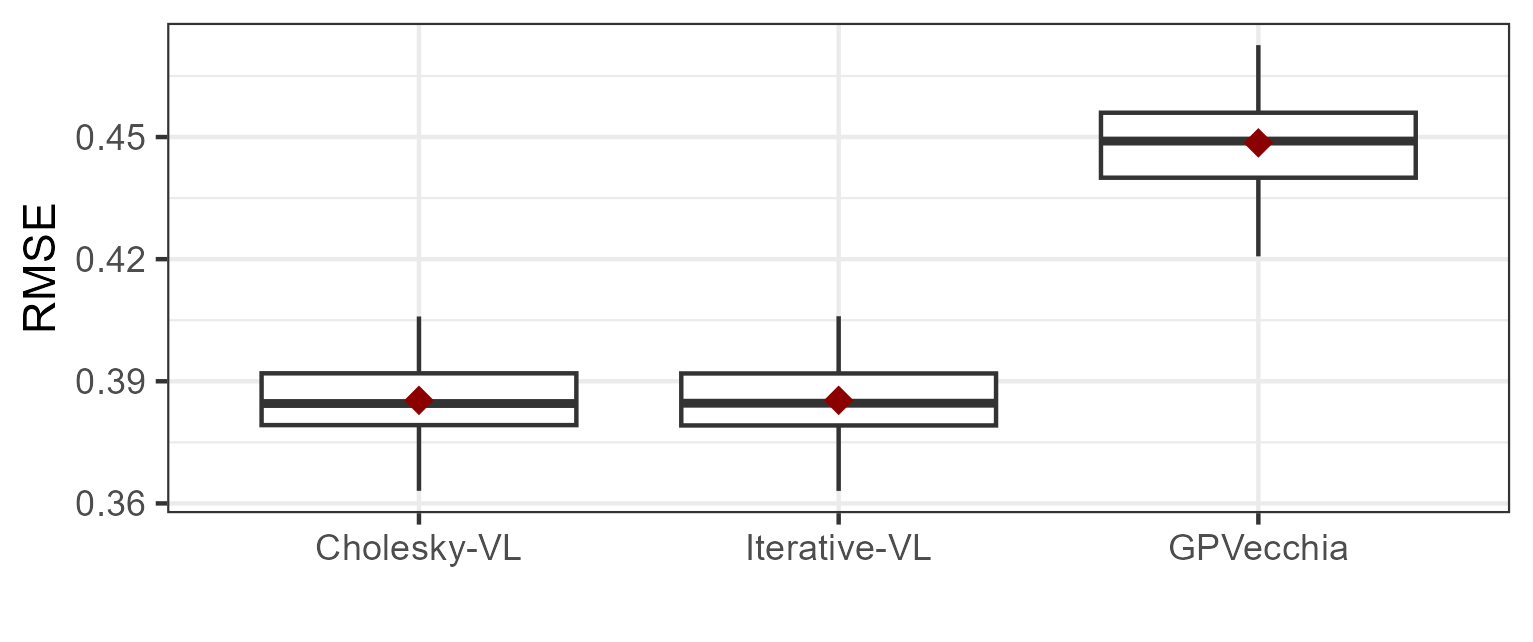}
    \includegraphics[width=0.42\linewidth]{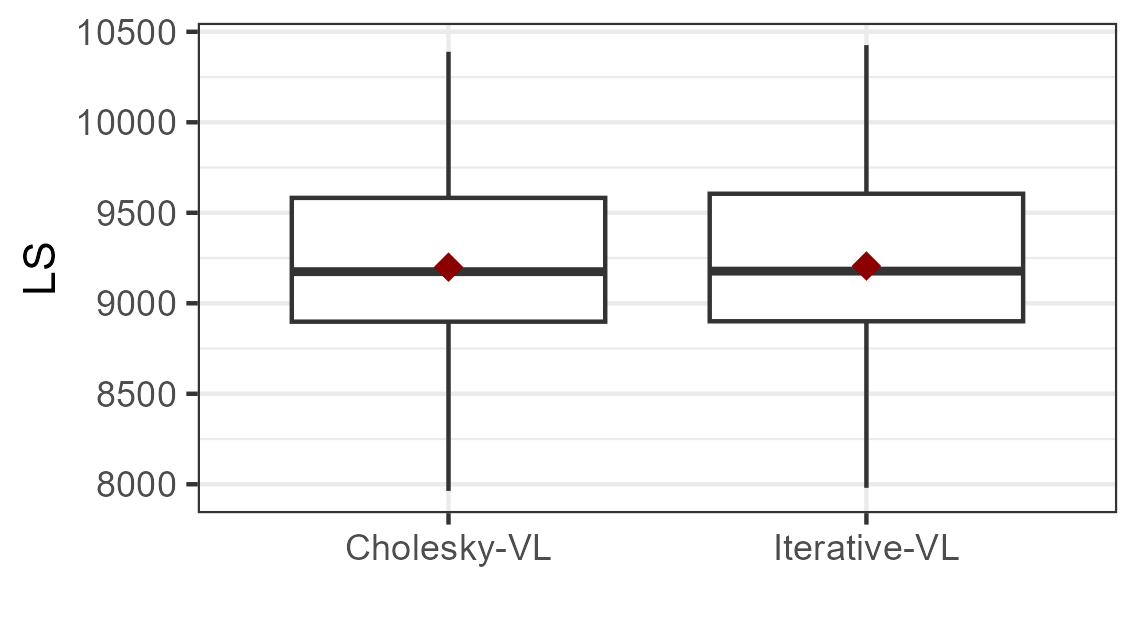}
    \caption{RMSE for predictive means and log score (LS) for probabilistic predictions. For GPVecchia, the log score could not be calculated due to negative predictive variances.}
    \label{fig:predictionRMSE20000}
\end{figure}

\subsection{Likelihood evaluation}
We simulate $100$ times data sets with a sample size of $n=20'000$ and evaluate the likelihood at the true parameters. We then calculate relative differences in the negative log-marginal likelihood of the Iterative-VL and GPVecchia approximations compared to the ``exact" results obtained with a Cholesky decomposition. Figure \ref{fig:relLikelihood20000} reports the results. The log-marginal likelihood calculated with iterative methods shows only minor deviations compared to Cholesky-based calculations with no systematic upward or downward bias. GPVecchia, on the other hand, overestimates the negative log-marginal likelihood.
\begin{figure}[ht!]
    \centering
    \includegraphics[width=0.4\linewidth]{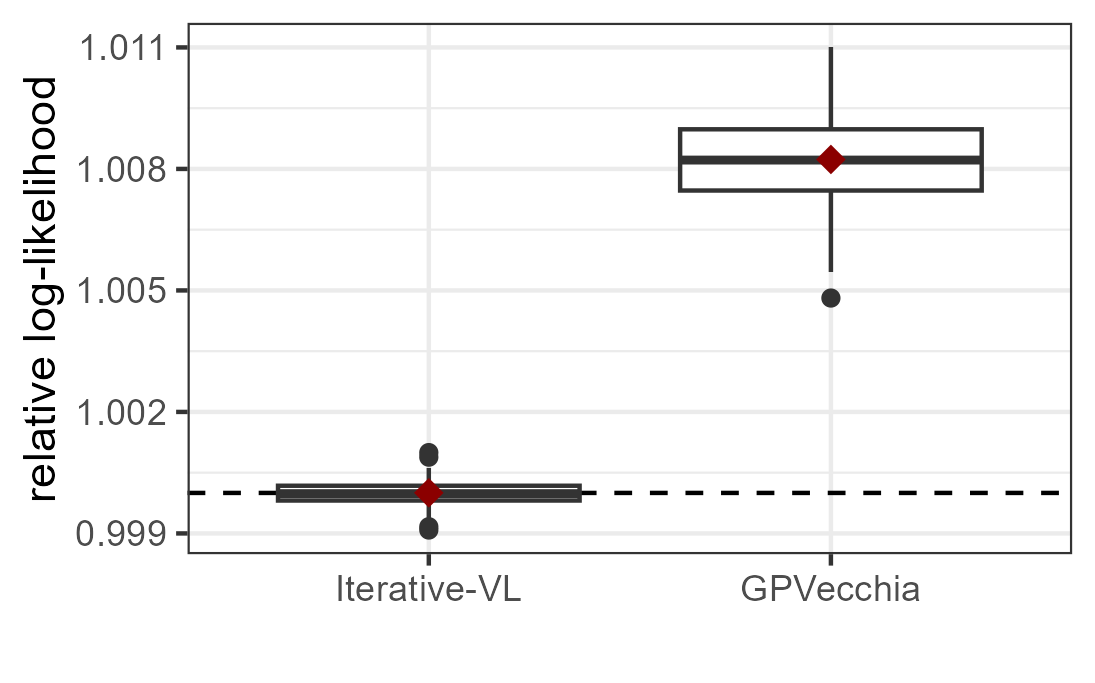}
    \caption{Relative differences in log-likelihood compared to Cholesky-based calculations.}
    \label{fig:relLikelihood20000}
\end{figure}

\subsection{Runtime comparison}
In the following, we compare the runtime to calculate marginal likelihoods of the different methods. We consider sample sizes ranging from $n=1'000$ up to $n=200'000$ and evaluate the likelihood at the true covariance parameters.
\begin{figure}[ht!]
\centering
    \includegraphics[width=0.7\linewidth]{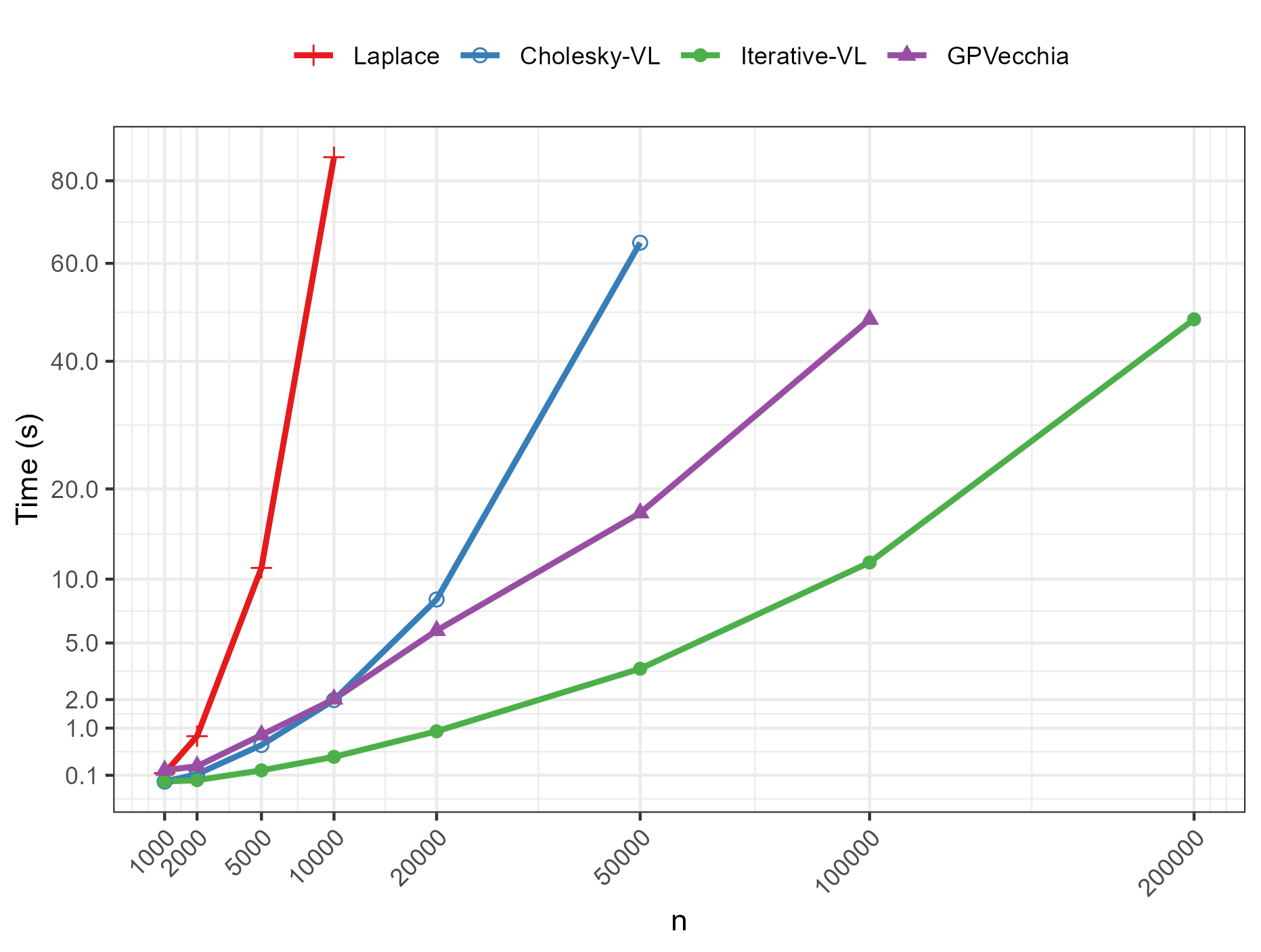}
    \caption{Runtime for one evaluation of the marginal likelihood for varying sample sizes.}
    \label{fig:runtime}
\end{figure}
The results in Figure \ref{fig:runtime} show that the iterative methods have the smallest runtime for all sample sizes $n \geq 2000$. For large $n$'s, e.g., $n=20'000$ and $n=50'000$, iterative methods are faster than Cholesky-based calculations by approximately one order of magnitude. GPVecchia has the second-fastest runtime for large sample sizes but is slower than the iterative methods. In Table \ref{table:runtime20000}, we also report the average runtime for estimating covariance parameters on data with $n=20'000$ as done in Section \ref{param_est}. Iterative-VL is more than seventeen times faster than GPVecchia and approximately six times faster than Cholesky-VL.

\begin{table}[ht!]
    \centering
    \begin{tabular}{llll}
        \hline
        \hline
        Cholesky-VL & Iterative-VL & GPVecchia\\ 
        \hline
        117.3s & 18.1s & 314.9s\\ 
        \hline
        \hline
    \end{tabular}
    \caption{Average wall clock times (s) for estimation on simulated data with $n=20'000$.} 
    \label{table:runtime20000}
\end{table}

\subsection{Varying number of nearest neighbors and sample sizes}\label{varying_m}
Next, we compare the iterative methods and GPVecchia for varying numbers of neighbors $m$ in the Vecchia approximation. We consider the marginal likelihood, prediction accuracy, and runtime for the likelihood evaluation for $m \in \{10,20,40,60\}$ on simulated data with $n=n_p=5'000$ and $n=n_p=50'000$. The likelihood is evaluated at the true covariance parameters.

\begin{figure}[ht!]
\centering
    \includegraphics[width=\linewidth]{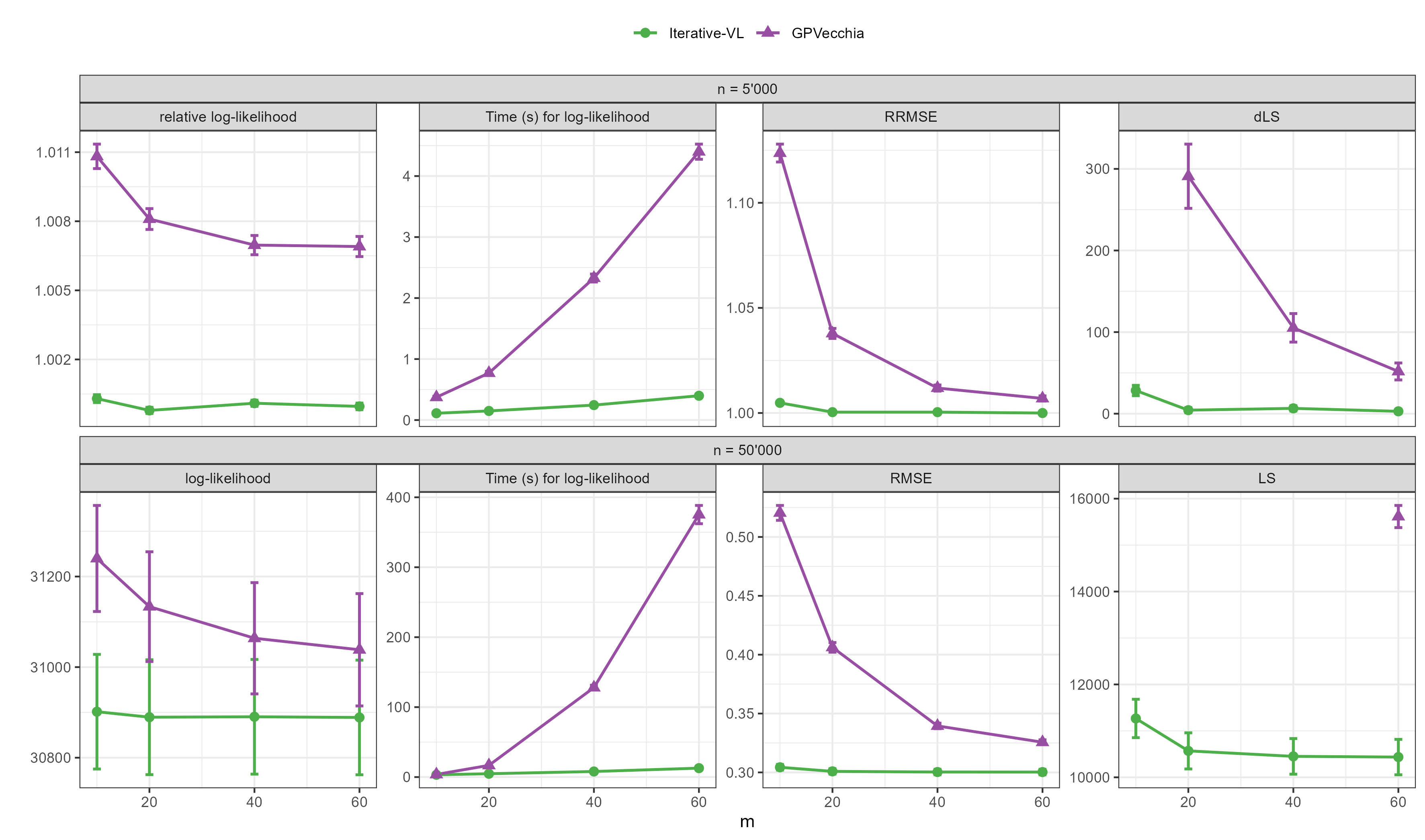}
    \caption{Negative log-marginal likelihood, RMSE, log score (LS), and runtime for likelihood calculation for varying numbers of neighbors $m$ and sample sizes $n=n_p$. For $n=n_p=5'000$, we report the relative difference in the likelihood and RMSE (RRMSE) and the difference in the log score (dLS) compared to a Laplace approximation. The whiskers represent confidence intervals obtained as $\pm 2 \times$ standard errors.}
    \label{fig:NN}
\end{figure}

Figure \ref{fig:NN} reports the results for $n=n_p=5'000$ and $100$ simulation runs and $n=n_p=50'000$ and $50$ simulation repetitions. For $n=n_p=5'000$, we report the average relative difference in the negative log-marginal likelihood, the relative root mean squared error (RRMSE), and the difference in the log score (dLS) compared to a Laplace approximation, and prediction is done using estimated parameters. For $n=n_p=50'000$, it is infeasible to run the experiment in a reasonable amount of time with a Laplace or a Cholesky-based Vecchia-Laplace approximation. We thus simply report the average negative log-marginal likelihood, the RMSE, and the log score. Further, prediction is done using the true covariance parameters for $n=n_p=50'000$, since estimation is very slow for GPVecchia, and for $n=n_p=5'000$, we use estimated ones. The results in Figure \ref{fig:NN} show that log-marginal likelihoods and predictions are again highly accurate for Iterative-VL also when the number of neighbors is small, and the runtime remains low as $m$ increases. GPVecchia is less accurate regarding likelihood evaluation and prediction even with a lot more neighbors, say $m=60$, compared to the iterative methods with $m=10$ neighbors. Further, a larger number of neighbors increases the runtime of GPVecchia in a non-linear manner. Note that GPVecchia produces negative predictive variances and the log score cannot be calculated in the following number of simulation runs: $50$, $50$, $43$, and $9$ for $m=10, 20, 40$, and $60$, respectively, for $n=50'000$ as well as $100$, $61$, $7$, and $2$ for $m=10, 20, 40$, and $60$, respectively, for $n=5'000$.

\section{Real-world applications}
In the following, we present experiments on two real-world data sets. We compare the proposed iterative methods (`Iterative-VL') with Cholesky-based calculations (`Cholesky-VL'), the approach of \citet{zilber2021vecchia} (`GPVecchia'), and a modified predictive process approximation \citep{finley2009improving}  (`LowRank') concerning prediction accuracy and runtime for calculating marginal likelihoods.

%\subsection{Water vapor data}\label{MODIS}
%https://ladsweb.modaps.eosdis.nasa.gov/missions-and-measurements/products/MYD05_L2#overview
We first consider a large, spatial, and non-Gaussian data set of column water vapor data collected from the NASA Earthdata website \url{https://www.earthdata.nasa.gov}. This data set contains measurements of column water vapor amounts over cloud-free ocean areas that were collected by NASA's Moderate-Resolution Imaging Spectroradiometer (MODIS) on the Aqua satellite between 13:45 and 13:50 on March 28, 2019, over a region with west, north, east, and south bounding coordinates -42.249, 67.444, 4.261, and 45.260, respectively. We consider data produced by the near-infrared algorithm with a 1 km spatial resolution. The same data was also used in \citet{zilber2021vecchia}. 

Since water vapor amounts are strictly positive, we use a gamma likelihood as in \citet{zilber2021vecchia}, $y_i|\mu_i \sim \Gamma(\alpha, \alpha \exp{(-\mu_i)})$, where $\alpha$ denotes the shape parameter, and $\beta=\alpha \exp{(-\mu_i)}$ denotes the rate parameter such that $\mathbb{E}[y_i|\mu_i] = \exp{(\mu_i)}$. We include a linear predictor term, $\vecmat{F}(\vecmat{X}) = \vecmat{X}\vecmat{\beta}$, and use an intercept and scaled longitude and latitude coordinates as predictor variables. Following \citet{zilber2021vecchia}, the latter ones are obtained by projecting the coordinates onto a rectangular grid with equally-spaced 1 km by 1 km cells, see Figure \ref{fig:realmean}. We use a Matérn covariance function with smoothness parameter $\nu=1.5$. 
%In the following, we compare the proposed iterative methods, denoted as `Iterative-VL', to the approach of \citet{zilber2021vecchia}, denoted as `GPVecchia'. In addition, we consider a low-rank approximation, denoted as `LowRank', and a Vecchia-Laplace approximation based on the Cholesky decomposition, denoted as `Cholesky-VL'. 
For the Vecchia approximation and the iterative methods, we use the same settings as described in Section \ref{simSetting}, except that we choose $s=1000$ samples for the simulation-based prediction variances. Regarding the low-rank approximation, we set the number of inducing points to $500$. For measuring the prediction accuracy, we randomly generate two disjoint training and test sets with $n=n_p=250'000$. We use the test RMSE to evaluate predictive response means $\mathbb{E}[\vecmat{y}_p|\vecmat{y},\vecmat{\theta},\vecmat{\xi}]=\mathbb{E}[\exp{(\vecmat{\mu}_p)}|\vecmat{y},\vecmat{\theta},\vecmat{\xi}]=\exp{(\vecmat{\omega}_{p} + \text{diag}(\vecmat{\Omega}_p)/2)}$ and the continuous ranked probability score (CRPS) to measure the accuracy of posterior predictive distributions $p(\vecmat{y}_p|\vecmat{y},\vecmat{\theta},\vecmat{\xi})$. The latter are obtained using simulation. Parameter estimation is done on a random sub-sample of the training data with $n=50'000$. We use a random sub-sample since the gamma marginal likelihood is very flat in some directions, which translates into long optimization times when jointly estimating the shape, covariance, and $\beta$ parameters. For GPVecchia, we could not estimate all parameters jointly because the estimation crashed and singular matrices occurred. Therefore, we adopt the same estimation procedure as in \citet{zilber2021vecchia}. For measuring the time for calculating the log-likelihood, we use the random subsample of size $n=50'000$, the covariance parameters at the optimum, $\alpha=28$, $\sigma_1^2=0.3$, $\rho=15$, and $\vecmat{\beta}$ is set to the generalized linear model estimate. Note that we do not report the total estimation time since we must use a different and simpler estimation approach for GPVecchia. However, Iterative-VL is approximately six times faster in terms of total estimation time compared to GPVecchia (results not tabulated).
%This highlights the need for a robust and fast method that can do proper joint estimation of parameters.

In Table \ref{table:realPrediction}, we report the test RMSE and CRPS, and the runtime to calculate the marginal likelihood. The iterative approach has a CRPS that is approximately three times lower compared to GPVecchia and more than two times lower than the low-rank approximation. The RMSE of the iterative methods is also considerably lower compared to GPVecchia and LowRank. Note that for GPVecchia, we can reproduce the results of \citet{zilber2021vecchia} and obtain almost the same RMSE and CRPS up to the third digit. Comparing iterative and Cholesky-based calculations, we obtain virtually identical CRPS and RMSE values. Furthermore, the iterative methods are about fifteen times faster in evaluating the marginal likelihood than Cholesky-based calculations and over three times faster than the GPVecchia approach.

\begin{table}[ht!]
	\centering
	\begin{tabular}{lllll}
		\hline
		\hline
		& Cholesky-VL & Iterative-VL & GPVecchia & LowRank\\ 
		\hline
		RMSE & 0.0812 & 0.0812 & 0.1207 & 0.1774\\ 
        \hline
		CRPS & 0.0444 & 0.0444 & 0.1412 & 0.0987\\ 
        \hline
		Time (s) & 89.6s & 5.9s & 21.4s & 5.2s\\ 
		\hline
		\hline
	\end{tabular}
	\caption{Prediction accuracy (lower = better) and runtime (s) for calculating marginal likelihoods on the MODIS column water vapor data set.}
    \label{table:realPrediction}
\end{table}

The estimated parameters are reported in Table \ref{table:realEstimates} in Appendix \ref{appendix_rw}. The results are almost identical when using other random sub-samples with $n=50'000$ for estimation (results not tabulated). In Table \ref{table:realLarge} in Appendix \ref{appendix_rw}, we also report the results for Iterative-VL when the parameters are estimated on the entire training data set when fixing the shape parameter to the estimate obtained on the random sub-sample. Both the parameter estimates and the prediction accuracy measures change only marginally.
%The shape parameter estimate obtained with the iterative Vecchia-Laplace approximation is $\alpha=28.13$. I.e., the likelihood is relatively close to a normal distribution. In contrast, the estimated shape parameter of GPVecchia is $\alpha=0.95$, which is almost the same value as obtained in \citet{zilber2021vecchia}. This corresponds to a gamma distribution that is much more right-skewed.

Figures \ref{fig:realmean} and \ref{fig:realvar} in Appendix \ref{appendix_rw} show the test data and maps of posterior predictive means $\mathbb{E}[\vecmat{y}_p|\vecmat{y},\vecmat{\theta},\vecmat{\xi}]$ and variances $\text{Var}[\vecmat{y}_p|\vecmat{y},\vecmat{\theta},\vecmat{\xi}]$. These maps visually confirm the results from Table \ref{table:realPrediction}. In particular, Iterative-VL and Cholesky-VL predict the fine-scale structure more accurately than GPVecchia, while the low-rank approximation misses the fine-scale structure completely. %Furthermore, the posterior predictive variance is generally the smallest for Iterative-VL and Cholesky-VL, followed by the larger variances of LowRank and GPVecchia.

In Appendix \ref{precip_data}, we additionally apply our methods to a spatial precipitation data set, and we obtain similar findings with an even larger outperformance of the iterative methods compared to GPVecchia and the low-rank approximation in terms of the CRPS and RMSE. The speed-up over the Cholesky decomposition is approximately four fold. This smaller speed-up is likely due to the fact that this data set is smaller ($n=10'901$), and we generally expect the VADU preconditioner to be less efficient on smaller data sets, see Sections \ref{section:preconditioners} and \ref{section:conv_theory}.

\section{Conclusion}
We present iterative methods for Vecchia-Laplace approximations whose main computational operations are matrix-vector multiplications with sparse triangular matrices. In particular, we introduce and analyze several preconditioners, derive novel theoretical results, and propose different methods for accurately approximating predictive variances. In experiments, we find that the current state-of-the-art approach for large sample sizes \citep{zilber2021vecchia} can be inaccurate in marginal likelihoods and can breakdown in joint parameter estimation. We observe that our novel methods are faster, more accurate, and more stable. Further, we find that the iterative methods achieve a substantial reduction in runtime compared to Cholesky-based calculations with negligible loss in accuracy.

%In our simulations, the VADU and LRAC preconditioners have a small variance of the log-marginal likelihood and a fast runtime and are therefore preferred. For predictive variances, the simulation-based method outperforms the Lanczos-based method, often achieving a certain accuracy with a ten times faster runtime.

%Limitations
An open challenge for the proposed iterative methods is that their convergence rates remain dependent on the covariance and fixed effects parameters despite preconditioning. Future research can systematically analyze in which situations which preconditioner is preferred. Finally, we mention that for very large data sets, computations can still become slow with the current software implementation. Future research can investigate whether and how graphics processing units (GPU)-based calculations can further speed-up calculations with iterative methods. 

%Further, it would be interesting to analyze modification schemes \citep{scott2014positive} for the zero fill-in reverse incomplete Cholesky preconditioner which could potentially lead to numerically stable, fast, and accurate preconditioners. Future research can also investigate how our iterative methods perform in a fully Bayesian framework. 

\section*{Acknowledgments}
This research was partially supported by the Swiss Innovation Agency - Innosuisse (grant number `55463.1 IP-ICT').

\clearpage
\bibliographystyle{abbrvnat}
\bibliography{bib_Iterative_VL.bib}
\clearpage

\appendix
\section{Appendix}

\subsection{Proofs of Theorem 3.1, Theorem 3.2, and Proposition 3.3}\label{proofs}
\begin{proof}[Proof of Theorem \ref{conv_VADU}]
We first observe that
\begin{equation*}
\begin{split}
\vecmat{P}_{\text{VADU}}^{-\frac{1}{2}}(\vecmat{W}+\tilde{\vecmat{\Sigma}}^{-1})\vecmat{P}_{\text{VADU}}^{-\frac{T}{2}} =& (\vecmat{W} + \vecmat{D}^{-1})^{-\frac{1}{2}}\vecmat{B}^{-T}\vecmat{W}\vecmat{B}^{-1}(\vecmat{W} + \vecmat{D}^{-1})^{-\frac{1}{2}} + (\vecmat{D}\vecmat{W} + \vecmat{I}_n)^{-1}\\
=&(\vecmat{D}\vecmat{W} + \vecmat{I}_n)^{-\frac{1}{2}}(\vecmat{D}^{\frac{1}{2}}\vecmat{B}^{-T}\vecmat{W}\vecmat{B}^{-1}\vecmat{D}^{\frac{1}{2}} + \vecmat{I}_n)(\vecmat{D}\vecmat{W} + \vecmat{I}_n)^{-\frac{1}{2}},
\end{split}
\end{equation*}
where $\vecmat{I}_n$ denotes the identity matrix of dimension $n$. It follows that for all $l$, $1\leq l \leq n$,
\begin{equation}\label{EV_PC}
    \begin{split}
    \lambda_l(\vecmat{P}_{\text{VADU}}^{-\frac{1}{2}}(\vecmat{W}+\tilde{\vecmat{\Sigma}}^{-1})\vecmat{P}_{\text{VADU}}^{-\frac{T}{2}}) &= \lambda_l(\vecmat{D}^{\frac{1}{2}}\vecmat{B}^{-T}\vecmat{W}\vecmat{B}^{-1}\vecmat{D}^{-\frac{1}{2}} + \vecmat{I}_n)\theta_l \\
        &= (\lambda_l(\vecmat{D}^{\frac{1}{2}}\vecmat{B}^{-T}\vecmat{W}\vecmat{B}^{-1}\vecmat{D}^{-\frac{1}{2}}) + 1)\theta_l\\
        &= (\lambda_l(\tilde{\vecmat{\Sigma}})\tilde\theta_l + 1)\theta_l,
    \end{split}
\end{equation}
where $\min((D_iW_{ii} + 1)^{-1}) \leq \theta_l \leq \max((D_iW_{ii} + 1)^{-1}) \leq 1$ and $\min(W_{ii}) \leq \tilde \theta_l \leq \max(W_{ii})$. The first line in \eqref{EV_PC} is due to Ostrowski's theorem \citep[][Theorem 4.5.9]{horn2012matrix}. The second line holds since for any symmetric matrix $\vecmat{A}$, we have $\lambda_l(\vecmat{A}+\vecmat{I}_n) = \lambda_l(\vecmat{A}) + 1$. And the third line holds again due to Ostrowski's theorem and since the two matrices $\vecmat{D}^{\frac{1}{2}}\vecmat{B}^{-T}\vecmat{W}\vecmat{B}^{-1}\vecmat{D}^{\frac{1}{2}}$ and $\vecmat{W}^{\frac{1}{2}}\vecmat{B}^{-1}\vecmat{D}\vecmat{B}^{-T}\vecmat{W}^{\frac{1}{2}} = \vecmat{W}^{\frac{1}{2}}\tilde{\vecmat{\Sigma}} \vecmat{W}^{\frac{1}{2}}$ have the same eigenvalues. The latter statement can be seen as follows. If $\lambda_l$ and $x_l$ are an eigenvalue and eigenvector of $\vecmat{D}^{\frac{1}{2}}\vecmat{B}^{-T}\vecmat{W}\vecmat{B}^{-1}\vecmat{D}^{\frac{1}{2}}$, then
$$
\vecmat{W}^{\frac{1}{2}}\vecmat{B}^{-1}\vecmat{D}^{\frac{1}{2}}\vecmat{D}^{\frac{1}{2}}\vecmat{B}^{-T}\vecmat{W}^{\frac{1}{2}}\vecmat{W}^{\frac{1}{2}}\vecmat{B}^{-1}\vecmat{D}^{\frac{1}{2}}x_l = \lambda_l \vecmat{W}^{\frac{1}{2}}\vecmat{B}^{-1}\vecmat{D}^{\frac{1}{2}}x_l
$$
which shows that $\lambda_l$ and $\vecmat{W}^{\frac{1}{2}}\vecmat{B}^{-1}\vecmat{D}^{\frac{1}{2}}x_l$ are an eigenvalue and eigenvector of $\vecmat{W}^{\frac{1}{2}}\tilde{\vecmat{\Sigma}} \vecmat{W}^{\frac{1}{2}}$. 

From \eqref{EV_PC}, it follows that 
\begin{equation*}
 (\lambda_l(\tilde{\vecmat{\Sigma}})\min(W_{ii}) + 1)\min((D_iW_{ii} + 1)^{-1}) \leq \lambda_l(\vecmat{P}_{\text{VADU}}^{-\frac{1}{2}}(\vecmat{W}+\tilde{\vecmat{\Sigma}}^{-1})\vecmat{P}_{\text{VADU}}^{-\frac{T}{2}})
\end{equation*}
and 
\begin{equation*}
\lambda_l(\vecmat{P}_{\text{VADU}}^{-\frac{1}{2}}(\vecmat{W}+\tilde{\vecmat{\Sigma}}^{-1})\vecmat{P}_{\text{VADU}}^{-\frac{T}{2}}) \leq (\lambda_l(\tilde{\vecmat{\Sigma}})\max(W_{ii}) + 1)\max((D_iW_{ii} + 1)^{-1}).
\end{equation*}
The statement in \eqref{CG_conv_VADU} then follows from existing results about the convergence of conjugate gradient algorithms, see, e.g., \citet[][Theorem B.4]{nishimura2022prior} or \citet[][Section 5.3]{van2003iterative}.

For the LVA preconditioner it holds that
$$
\vecmat{P}_{\text{LVA}}^{-\frac{1}{2}}(\vecmat{W}+\tilde{\vecmat{\Sigma}}^{-1})\vecmat{P}_{\text{LVA}}^{-\frac{T}{2}} = \vecmat{D}^{\frac{1}{2}}\vecmat{B}^{-T}\vecmat{W}\vecmat{B}^{-1}\vecmat{D}^{\frac{1}{2}} + \vecmat{I}_n.
$$
Analogous arguments as for the VADU preconditioner above then lead to
$$
\lambda_l(\vecmat{P}_{\text{LVA}}^{-\frac{1}{2}}(\vecmat{W}+\tilde{\vecmat{\Sigma}}^{-1})\vecmat{P}_{\text{LVA}}^{-\frac{T}{2}}) = \lambda_l(\tilde{\vecmat{\Sigma}})\tilde\theta_l + 1,
$$
where again $\min(W_{ii}) \leq \tilde \theta_l \leq \max(W_{ii})$, and thus
\begin{equation*}
\lambda_l(\tilde{\vecmat{\Sigma}})\min(W_{ii}) + 1 \leq \lambda_l(\vecmat{P}_{\text{LVA}}^{-\frac{1}{2}}(\vecmat{W}+\tilde{\vecmat{\Sigma}}^{-1})\vecmat{P}_{\text{LVA}}^{-\frac{T}{2}}) \leq \lambda_l(\tilde{\vecmat{\Sigma}})\max(W_{ii}) + 1.
\end{equation*}

\end{proof}

\begin{proof}[Proof of Theorem \ref{acc_SLQ}]

%$\frac{\lambda_1}{\lambda_1 + \lambda_n}$ and $\frac{\lambda_n}{\lambda_1 + \lambda_n}$

We first proof the multiplicative bound in \eqref{mult_bound_slq}. This is done using similar arguments as in \citet{ubaru2017fast} and the fact that Theorem \ref{conv_VADU} shows that the smallest eigenvalue $\lambda_n(\vecmat{P}_{\text{LVA}}^{-\frac{1}{2}}(\vecmat{W}+\tilde{\vecmat{\Sigma}}^{-1})\vecmat{P}_{\text{LVA}}^{-\frac{T}{2}})$ is larger than $1$. However, Theorem 4.1 of \citet{ubaru2017fast} holds for Rademacher vectors, but we use Gaussian random vector. The latter implies, among other things, that the inequality in (15) in \citet{ubaru2017fast} does not hold since the norm of Gaussian random vectors is stochastic and unbounded whereas the norm of Rademacher vectors of size $n$ is $n$. Further, despite the fact that the smallest eigenvalue is larger than $1$, we could not directly apply Theorem 4.1 of \citet{ubaru2017fast} even if we used Rademacher vectors since the logarithm function diverges at the boundary of the Bernstein ellipse used in Theorem 4.1 of \citet{ubaru2017fast}. We thus use a smaller ellipse. In addition, \citet{cortinovis2021randomized} point out two errors in the proof of Theorem 4.1 of \citet{ubaru2017fast} whose corrections we use below.

In the following, we denote $\vecmat{A} = \vecmat{P}^{-\frac{1}{2}}(\vecmat{W}+\tilde{\vecmat{\Sigma}}^{-1})\vecmat{P}^{-\frac{T}{2}}$ and 
$$
\hat\tr_n(\log(\vecmat{A})) =\frac{1}{t}\sum_{i=1}^t\vecmat{z}_i^T \vecmat{P}^{-\frac{T}{2}}\log(\vecmat{A}) \vecmat{P}^{-\frac{1}{2}} \vecmat{z}_i.
$$
Similarly as in the proof of Theorem 4.1 of \citet{ubaru2017fast}, we consider the function 
$$
g(t) = \log\left(\frac{\lambda_1-\lambda_n}{2}t+\frac{\lambda_1+\lambda_n}{2}\right),
$$
which is analytic in $[-1,1]$. Further, denote by $E_\rho$ a Bernstein ellipse with semimajor axis length $\alpha = \frac{\kappa+1}{\kappa}$, sum of the two semiaxis $\rho = \frac{\sqrt{2\kappa + 1} + 1}{\sqrt{2\kappa + 1} - 1}>1$, and foci $\pm 1$. It follows that
\begin{equation*}
\begin{split}
    \sup_{z\in E_\rho}|g(z)| &\leq \sqrt{\sup_{z\in E_\rho} \left(\log\left(\frac{\lambda_1-\lambda_n}{2}|z|+\frac{\lambda_1+\lambda_n}{2}\right)\right)^2+\pi^2} = \sqrt{\left(\log\left(\lambda_n(\kappa + 1 -1/\kappa)\right)\right)^2+\pi^2}\\
    &\leq \log\left(\lambda_n(\kappa + 1 -1/\kappa)\right)+\pi= M_\rho,
\end{split}
\end{equation*}
where we use the fact that the function $g(t)$ attains its maximum at $\alpha$ and $g(\alpha) = \log\left(\lambda_n(\kappa + 1 -1/\kappa)\right)$. The integral representation used in \citet[][Equation (5)]{ubaru2017fast} and Theorem 3 of \citet{cortinovis2021randomized}, instead of the erroneous Theorem 4.2 of \citet{ubaru2017fast}, applied to $g(t)$, with $\rho$ and $M_\rho$ are given above and, e.g., $\rho_0 = \frac{\sqrt{\kappa}+1}{\sqrt{\kappa}-1}$, then lead to
$$
|\vecmat{z}_i^T \vecmat{P}^{-\frac{T}{2}}\log(\vecmat{A}) \vecmat{P}^{-\frac{1}{2}} \vecmat{z}_i - \vecmat{z}_i^T \vecmat{P}^{-\frac{T}{2}} \tilde{\vecmat{Q}}_i\log(\tilde{\vecmat{T}}_i) \tilde{\vecmat{Q}}_i^T \vecmat{P}^{-\frac{1}{2}} \vecmat{z}_i| \leq \| \vecmat{z}_i^T \vecmat{P}^{-\frac{T}{2}}\|_2^2 \frac{C_\rho}{\rho^{2l}},
$$
where 
\begin{equation}\label{defC}
C_\rho = 4\frac{M_\rho}{1-\rho^{-1}} = 2(\log\left(\lambda_n(\kappa + 1 -1/\kappa)\right)+\pi)(\sqrt{2\kappa+1}+1),
\end{equation}
and thus
\begin{equation}\label{ineq1}
|\frac{1}{t}\sum_{i=1}^t\left(\hat\tr_n(\log(\vecmat{A})) - \vecmat{z}_i^T \vecmat{P}^{-\frac{T}{2}} \tilde{\vecmat{Q}}_i\log(\tilde{\vecmat{T}}_i) \tilde{\vecmat{Q}}_i^T \vecmat{P}^{-\frac{1}{2}} \vecmat{z}_i\right)| \leq \frac{1}{t}\sum_{i=1}^t\| \vecmat{z}_i^T \vecmat{P}^{-\frac{T}{2}}\|_2^2 \frac{C_\rho}{\rho^{2l}}.
\end{equation}
Note that, in contrast to \citet[][page 1087]{ubaru2017fast}, there is no factor $\frac{\lambda_1-\lambda_n}{2}$ in $C_\rho$ since \citet[][page 889]{cortinovis2021randomized} point out that this factor is erroneous. Next, the fact that $\tr\log(\vecmat{A}) \geq n \log(\lambda_n)$, \eqref{ineq1}, and \eqref{defC} imply that
\begin{equation}\label{ineq2}
    \begin{split}
P\left(|\hat\tr_n(\log(\vecmat{A})) -\widehat \Gamma | \leq \frac{\epsilon}{2}\tr\log(\vecmat{A})\right) & \geq P\left(|\hat\tr_n(\log(\vecmat{A})) -\widehat \Gamma | \leq \frac{\epsilon n\log(\lambda_n)}{2}\right)\\
&\geq P\left( \frac{1}{t}\sum_{i=1}^t\| \vecmat{z}_i^T \vecmat{P}^{-\frac{T}{2}}\|_2^2 \frac{C_\rho}{\rho^{2l}}\leq \frac{\epsilon n\log(\lambda_n)}{2}\right)\\
&\geq P\left( \frac{1}{nt}\sum_{i=1}^t\| \vecmat{z}_i^T \vecmat{P}^{-\frac{T}{2}}\|_2^2\leq C_{nt}\right) \\
&\geq 1 - \eta/2,
    \end{split}
\end{equation}
where the second last inequality holds since $\frac{C_\rho}{\rho^{2l}}\leq \frac{\epsilon \log(\lambda_n)}{2C_{nt}}$ is equivalent to
\begin{equation*}
l\geq \frac{1}{2}\log\left(\frac{4C_{nt}(\log\left(\lambda_n(\kappa + 1 -1/\kappa)\right)+\pi)(\sqrt{2\kappa+1}+1)}{\log(\lambda_n)\epsilon}\right)/\log\left(\frac{\sqrt{2\kappa+1}+1}{\sqrt{2\kappa+1}-1}\right),
\end{equation*}
and the last inequality in \eqref{ineq2} holds due to the definition of $C_{nt}$ and by noting that $\sum_{i=1}^t\| \vecmat{z}_i^T \vecmat{P}^{-\frac{T}{2}}\|_2^2$ follows a chi-squared distribution with $nt$ degrees of freedom.

Finally, plugging everything together gives
\begin{equation*}
    \begin{split}
        P\left((|\tr\log(\vecmat{A}) - \widehat \Gamma |\leq \epsilon\tr\log(\vecmat{A})\right) \geq & P\left(|\tr\log(\vecmat{A}) - \hat\tr_n(\log(\vecmat{A}))| + |\hat\tr_n(\log(\vecmat{A})) -\widehat \Gamma | \leq \epsilon\tr\log(\vecmat{A})\right) \\
        \geq & P\left(|\tr\log(\vecmat{A}) - \hat\tr_n(\log(\vecmat{A}))|\leq \frac{\epsilon}{2}\tr\log(\vecmat{A})\right) \\
        & ~~ + P\left(|\hat\tr_n(\log(\vecmat{A})) -\widehat \Gamma | \leq \frac{\epsilon}{2}\tr\log(\vecmat{A})\right) - 1\\
        \geq & 1 - \eta,
    \end{split}
\end{equation*}
where in the last inequality, we used \eqref{ineq2} and the fact that $P\left(|\tr\log(\vecmat{A}) - \hat\tr_n(\log(\vecmat{A}))|\leq \frac{\epsilon}{2}\tr\log(\vecmat{A})\right) \geq 1 - \eta/2$ if $t\geq  32 / \epsilon^2 \log(4/\eta)$. The latter statement concerning the accuracy of a stochastic trace estimator can be shown by applying Theorem 3 of \citet{roosta2015improved} with $\varepsilon = \epsilon/2$ and $\delta = \eta/2$. This completes the proof for the multiplicative bound in \eqref{mult_bound_slq}. Clearly, \eqref{mult_bound_slq} holds also for $l\geq l^m(\gamma^{\text{LVA}}_1)$ and $t\geq t^m(\gamma^{\text{LVA}}_1)$ since $l^m(\kappa)$ and $t^m(\kappa)$ are increasing and  $\gamma^{\text{LVA}}_1\geq \kappa$ and $\gamma^{\text{LVA}}_1 \geq \kappa$ by Theorem \ref{conv_VADU}.

The additive bound in \eqref{bound_slq} follows using the same arguments as in the proof of Corollary 4.5 of \citet{ubaru2017fast}, however, applying the analogous modifications as above to account for the fact that we use Gaussian random vectors instead of Rademacher vectors and making the two above-mentioned corrections observed by \citet{cortinovis2021randomized}. The former means that, in contrast to \citet[][Corrolary 4.5]{ubaru2017fast}, the factor $C_{nt}$ is added in $l(\kappa)$ and that for $t(\kappa)$, the multiplicative constant is $32$ instead of $24$ and $\eta$ is replaced by $\eta/2$. The corrections by \citet{cortinovis2021randomized} imply that the factor $\frac{5\kappa\log(2(\kappa+1))}{\sqrt{2\kappa+1}}$ in \citet{ubaru2017fast} is replaced by $20\log(2(\kappa+1))\sqrt{2\kappa+1}$.

\end{proof}

\begin{proof}[Proof of Proposition \ref{pred_var_sim}]
First, observe that $\vecmat{z}_i^{(3)} = \vecmat{W}^{\frac{1}{2}}\vecmat{z}_i^{(1)} +\vecmat{B}^T\vecmat{D}^{-\frac{1}{2}}\vecmat{z}_i^{(2)} \sim \mathcal{N}(\vecmat{0}, (\vecmat{W}+\tilde{\vecmat{\Sigma}}^{-1}))$. It follows that $\vecmat{z}_i^{(4)} = \vecmat{B}_p^{-1}\vecmat{B}_{po}\vecmat{z}_i^{(3)}\sim \mathcal{N}(\vecmat{0},  \vecmat{B}_p^{-1}\vecmat{B}_{po}(\vecmat{W} + \tilde{\vecmat{\Sigma}}^{-1})^{-1} \vecmat{B}_{po}^T \vecmat{B}_p^{-1})$. By standard results, the empirical covariance matrix $\frac{1}{s}\sum_{i=1}^s\vecmat{z}_i^{(4)}(\vecmat{z}_i^{(4)})^T$ in Algorithm \ref{simAlgo} is an unbiased and consistent estimator for $\vecmat{B}_p^{-1}\vecmat{B}_{po}(\vecmat{W} + \tilde{\vecmat{\Sigma}}^{-1})^{-1} \vecmat{B}_{po}^T \vecmat{B}_p^{-1}$, and the claim in Proposition \ref{pred_var_sim} thus follows.
\end{proof}

\subsection{Derivations for stochastic Lanczos quadrature for log-determinants and stochastic trace estimation for gradients}\label{deriv_SLQ}
The approximation in \eqref{def_SLQ} is obtained as follows:
\begin{equation*}
\begin{split}
    \log\det(\vecmat{P}^{-\frac{1}{2}}(\vecmat{W}+\tilde{\vecmat{\Sigma}}^{-1})\vecmat{P}^{-\frac{T}{2}}) & = \tr(\mathbb{E}[\vecmat{P}^{-\frac{1}{2}} \vecmat{z}_i \vecmat{z}_i^T \vecmat{P}^{-\frac{T}{2}}]\log(\vecmat{P}^{-\frac{1}{2}}(\vecmat{W}+\tilde{\vecmat{\Sigma}}^{-1})\vecmat{P}^{-\frac{T}{2}}))\\
    & = \mathbb{E}[\vecmat{z}_i^T \vecmat{P}^{-\frac{T}{2}}\log(\vecmat{P}^{-\frac{1}{2}}(\vecmat{W}+\tilde{\vecmat{\Sigma}}^{-1})\vecmat{P}^{-\frac{T}{2}}) \vecmat{P}^{-\frac{1}{2}} \vecmat{z}_i]\\
    %&\approx \tr\log(\tilde{\vecmat{Q}}_i \tilde{\vecmat{T}}_i \tilde{\vecmat{Q}}_i^T) \\
    &\approx \mathbb{E}[\vecmat{z}_i^T \vecmat{P}^{-\frac{T}{2}} \tilde{\vecmat{Q}}_i \log(\tilde{\vecmat{T}}_i) \tilde{\vecmat{Q}}_i^T \vecmat{P}^{-\frac{1}{2}} \vecmat{z}_i] \\
    &\approx \frac{1}{t} \sum_{i=1}^t \vecmat{z}_i^T \vecmat{P}^{-\frac{T}{2}} \tilde{\vecmat{Q}}_i\log(\tilde{\vecmat{T}}_i) \tilde{\vecmat{Q}}_i^T \vecmat{P}^{-\frac{1}{2}} \vecmat{z}_i\\
    &= \frac{1}{t} \sum_{i=1}^t\|\vecmat{P}^{-\frac{1}{2}}\vecmat{z}_i\|_2^2 \vecmat{e}_1^T \log(\tilde{\vecmat{T}}_i) \vecmat{e}_1 \approx \frac{n}{t} \sum_{i=1}^t \vecmat{e}_1^T \log(\tilde{\vecmat{T}}_i) \vecmat{e}_1,
\end{split}
\end{equation*}
The equality in the last line holds because all columns of $\tilde{\vecmat{Q}}_i$ except the first one are orthogonal to the start vector: $\tilde{\vecmat{Q}}_i^T \vecmat{P}^{-\frac{1}{2}}\vecmat{z}_i=\|\vecmat{P}^{-\frac{1}{2}}\vecmat{z}_i\|_2\vecmat{e}_1\approx\sqrt{n}\vecmat{e}_1$.

The approximation in \eqref{grad_STE} is obtained as follows:
\begin{equation*}
    \begin{split}
    \frac{\partial\log\det(\vecmat{W}+\tilde{\vecmat{\Sigma}}^{-1})}{\partial\theta_k} 
    &= \tr\left((\vecmat{W}+\tilde{\vecmat{\Sigma}}^{-1})^{-1}\frac{\partial (\vecmat{W}+\tilde{\vecmat{\Sigma}}^{-1})}{\partial\theta_k}\mathbb{E}[\vecmat{P}^{-1}\vecmat{z}_i\vecmat{z}_i^T]\right)\\
    &= \mathbb{E}\left[\vecmat{z}_i^T (\vecmat{W}+\tilde{\vecmat{\Sigma}}^{-1})^{-1}\frac{\partial (\vecmat{W}+\tilde{\vecmat{\Sigma}}^{-1})}{\partial\theta_k}\vecmat{P}^{-1} \vecmat{z}_i\right]\\
    &\approx \frac{1}{t}\sum_{i=1}^t ((\vecmat{W}+\tilde{\vecmat{\Sigma}}^{-1})^{-1}\vecmat{z}_i)^T\frac{\partial (\vecmat{W}+\tilde{\vecmat{\Sigma}}^{-1})}{\partial\theta_k}\vecmat{P}^{-1} \vecmat{z}_i.
    \end{split}
\end{equation*}

\subsection{Gradients of Vecchia-Laplace approximations}\label{gradientsVLA}
The gradients with respect to $\vecmat{F}$, $\vecmat{\theta}$, and $\vecmat{\xi}$ of the approximate negative logarithmic marginal likelihood of a Vecchia-Laplace approximation $L^{LAV}(\vecmat{y},\vecmat{F},\vecmat{\theta},\vecmat{\xi})$ in \eqref{VLA_loss} are given by
\begin{equation*}
\begin{split}
\frac{\partial L^{VLA}(\vecmat{y},\vecmat{F},\vecmat{\theta},\vecmat{\xi})}{\partial F_i}=&-\frac{\partial \log p(y_i|\mu^*_i,\vecmat{\xi})}{\partial \mu^*_i} + \frac{1}{2}\frac{\partial\log\det(\tilde{\vecmat{\Sigma}}\vecmat{W}+\vecmat{I}_n)}{\partial\mu^*_i}\\
& + \left(\frac{\partial L^{VLA}(\vecmat{y},\vecmat{F},\vecmat{\theta},\vecmat{\xi})}{\partial \vecmat{b}^*}\right)^T\frac{\partial \vecmat{b}^*}{\partial F_i},~~~~ i=1,\dots,n,
\end{split}
\end{equation*}
\begin{equation*}
\begin{split}
\frac{\partial L^{VLA}(\vecmat{y},\vecmat{F},\vecmat{\theta},\vecmat{\xi})}{\partial \theta_k}=& \frac{1}{2}{\vecmat{b}^*}^T \frac{\partial \tilde{\vecmat{\Sigma}}^{-1}}{\partial  \theta_k} \vecmat{b}^* + \frac{1}{2}\frac{\partial\log\det(\tilde{\vecmat{\Sigma}}\vecmat{W}+\vecmat{I}_n)}{\partial\theta_k}\\
&+ \left(\frac{\partial L^{VLA}(\vecmat{y},\vecmat{F},\vecmat{\theta},\vecmat{\xi})}{\partial \vecmat{b}^*}\right)^T\frac{\partial \vecmat{b}^*}{\partial \theta_k}, ~~~~ k=1,\dots,q,
\end{split}
\end{equation*}
\begin{equation*}
\begin{split}
\frac{\partial L^{VLA}(\vecmat{y},\vecmat{F},\vecmat{\theta},\vecmat{\xi})}{\partial \xi_l}=&-\frac{\partial \log p(\vecmat{y}|\vecmat{\mu}^*,\vecmat{\xi})}{\partial \xi_l}  + \frac{1}{2}\frac{\partial\log\det(\tilde{\vecmat{\Sigma}}\vecmat{W}+\vecmat{I}_n)}{\partial\xi_l}\\
&+ \left(\frac{\partial L^{VLA}(\vecmat{y},\vecmat{F},\vecmat{\theta},\vecmat{\xi})}{\partial \vecmat{b}^*}\right)^T\frac{\partial \vecmat{b}^*}{\partial \xi_l}, ~~~~ l=1,\dots,r,
\end{split}
\end{equation*}
where
\begin{equation*}
\frac{\partial L^{VLA}(\vecmat{y},\vecmat{F},\vecmat{\theta},\vecmat{\xi})}{\partial b^*_i}
=\frac{1}{2}\frac{\partial\log\det(\tilde{\vecmat{\Sigma}}\vecmat{W}+\vecmat{I}_n)}{\partial\mu^*_i},
\end{equation*} 
\begin{equation*}
\frac{\partial \vecmat{b}^*}{\partial F_i} =- \left(\vecmat{W} + \tilde{\vecmat{\Sigma}}^{-1}\right)^{-1}\vecmat{W}_{\cdot i},
\end{equation*}
\begin{equation}\label{grad_mode_par}
\frac{\partial \vecmat{b}^*}{\partial \theta_k} =-\left(\vecmat{W} + \tilde{\vecmat{\Sigma}}^{-1}\right)^{-1} \frac{\partial \tilde{\vecmat{\Sigma}}^{-1}}{\partial \theta_k} \vecmat{b}^*,
\end{equation}
\begin{equation*}
\frac{\partial \vecmat{b}^*}{\partial \xi_l} = \left(\vecmat{W} + \tilde{\vecmat{\Sigma}}^{-1}\right)^{-1} \frac{\partial^2 \log p(\vecmat{y}|\vecmat{\mu}^*,\vecmat{\xi})}{\partial  \xi_l\partial \vecmat{\mu}^*}.
\end{equation*}
$\vecmat{W}_{\cdot i}$ denotes column $i$ of $\vecmat{W}$. To reduce the number of linear systems in the implementation, one should calculate the implicit derivatives by first solving for $\left(\vecmat{W}+\tilde{\vecmat{\Sigma}}^{-1}\right)^{-1}\frac{\partial L^{VLA}(\vecmat{y}, \vecmat{F},\vecmat{\theta},\vecmat{\xi})}{\partial \vecmat{b}^*}$. Further, it is preferred to work with $\frac{\partial \tilde{\vecmat{\Sigma}}^{-1}}{\partial \theta_k}$ instead of $\frac{\partial \tilde{\vecmat{\Sigma}}}{\partial \theta_k}$ since then $\vecmat{B}^{-1}$ is not required, see  Equations \eqref{vecchiaSigmaderiv} and \eqref{vecchiaSigmaIderiv} in Appendix \ref{vecchiaAppendix} for expressions for $\frac{\partial \tilde{\vecmat{\Sigma}}^{-1}}{\partial \theta_k}$ and $\frac{\partial \tilde{\vecmat{\Sigma}}}{\partial \theta_k}$. Therefore we use in \eqref{grad_mode_par} that $-\frac{\partial \tilde{\vecmat{\Sigma}}^{-1}}{\partial \theta_k} \vecmat{b}^* = \tilde{\vecmat{\Sigma}}^{-1} \frac{\partial \tilde{\vecmat{\Sigma}}}{\partial \theta_k} \frac{\partial \log p(\vecmat{y}|\vecmat{\mu}^*,\vecmat{\xi})}{\partial  \vecmat{\mu}^*}$ which holds because $\vecmat{b}^*=\argmax_{\vecmat{b}}\log p(\vecmat{y}|\vecmat{\mu},\vecmat{\xi}) - \frac{1}{2} \vecmat{b}^T\tilde{\vecmat{\Sigma}}^{-1} \vecmat{b}$.

\subsection{Derivatives of log-determinants using stochastic trace estimation and variance reduction for different preconditioners} \label{PAppendix}
In the following, $\frac{\partial \vecmat{W}}{\partial \mu^*_i}$ denotes a matrix of $0$'s except for the entry $\left(\frac{\partial \vecmat{W}}{\partial \mu^*_i}\right)_{ii} = -\frac{\partial^3 \log p(y_i|\mu^*_i,\vecmat{\xi})}{\partial {\mu^*_i}^3}$, and $\frac{\partial \vecmat{W}}{\partial\xi_l}$ denotes a diagonal matrix with entries $\left(\frac{\partial \vecmat{W}}{\partial \xi_l}\right)_{ii}=-\frac{\partial^3 \log p(y_i|\mu^*_i,\vecmat{\xi})}{\partial \xi_l \partial{\mu^*_i}^2}$. Further, we recall that $\frac{\partial \tilde{\vecmat{\Sigma}}^{-1}}{\partial \theta_k}$ is given in \eqref{vecchiaSigmaderiv} in Appendix \ref{vecchiaAppendix}.

\subsubsection{Preconditioner VADU} \label{P_VADUAppendix}
The following holds:
\begin{equation*}
\begin{split}
    \frac{\partial\log\det\left(\tilde{\vecmat{\Sigma}} \vecmat{W}+\vecmat{I}_n\right)}{\partial\mu^*_i}
    =&\;\frac{\partial}{\partial\mu^*_i}\left(c\log\det(\vecmat{P}_{\text{VADU}})+\log\det\left(\vecmat{W} +\tilde{\vecmat{\Sigma}}^{-1}\right)-c\log\det(\vecmat{P}_{\text{VADU}})\right) \\
    =&\;c\;\tr\left(\left(\vecmat{W}+\vecmat{D}^{-1}\right)^{-1}\frac{\partial \vecmat{W}}{\partial\mu^*_i}\right) \\
    &\;+\frac{1}{t}\sum_{i=1}^t \left(\underbrace{\left(\left(\vecmat{W} +\tilde{\vecmat{\Sigma}}^{-1}\right)^{-1}\vecmat{z}_i\right)^{T}\frac{\partial \vecmat{W}}{\partial\mu^*_i}\vecmat{P}_{\text{VADU}}^{-1}\vecmat{z}_i}_\text{$h(\vecmat{z}_i)$} \right.\\
    &\left.-\;c\underbrace{\left(\vecmat{P}_{\text{VADU}}^{-1} \vecmat{z}_i\right)^T\vecmat{B}^T\frac{\partial \vecmat{W}}{\partial\mu^*_i}\vecmat{B}\vecmat{P}_{\text{VADU}}^{-1}\vecmat{z}_i}_\text{$r(\vecmat{z}_i)$}\right),
\end{split}
\end{equation*}
\begin{equation*}
\begin{split}
    \frac{\partial\log\det\left(\tilde{\vecmat{\Sigma}} \vecmat{W}+\vecmat{I}_n\right)}{\partial\theta_k}
    =&\;\frac{\partial}{\partial\theta_k}\left(\log\det(\tilde{\vecmat{\Sigma}}) + c\log\det(\vecmat{P}_{\text{VADU}})+\log\det\left(\vecmat{W} +\tilde{\vecmat{\Sigma}}^{-1}\right)\right. \\
    &\left.\;-c\log\det(\vecmat{P}_{\text{VADU}})\right)\\
    =&\;\tr(\vecmat{D}^{-1}\frac{\partial \vecmat{D}}{\partial\theta_k}) - c\;\tr\left(\left(\vecmat{W}+\vecmat{D}^{-1}\right)^{-1}\vecmat{D}^{-1}\frac{\partial \vecmat{D}}{\partial\theta_k}\vecmat{D}^{-1}\right)\\
    &\;+\frac{1}{t}\sum_{i=1}^t \left(\underbrace{\left(\left(\vecmat{\vecmat{W}} +\tilde{\vecmat{\Sigma}}^{-1}\right)^{-1}\vecmat{z}_i\right)^{T}\frac{\partial\tilde{\vecmat{\Sigma}}^{-1}}{\partial\theta_k}\vecmat{P}_{\text{VADU}}^{-1}\vecmat{z}_i}_\text{$h(\vecmat{z}_i)$}\right.\\
    &\left.\;-\;c\underbrace{\left(\vecmat{P}_{\text{VADU}}^{-1} \vecmat{z}_i\right)^T\left(\frac{\partial\tilde{\vecmat{\Sigma}}^{-1}}{\partial\theta_k}+\frac{\partial \vecmat{B}^T}{\partial\theta_k}\vecmat{W}\vecmat{B}+\vecmat{B}^T\vecmat{W}\frac{\partial \vecmat{B}}{\partial\theta_k}\right)\vecmat{P}_{\text{VADU}}^{-1}\vecmat{z}_i}_\text{$r(\vecmat{z}_i)$}\right),
\end{split}
\end{equation*}
\begin{equation*}
\begin{split}
    \frac{\partial\log\det\left(\tilde{\vecmat{\Sigma}} \vecmat{W}+\vecmat{I}_n\right)}{\partial\xi_l}
    =&\;\frac{\partial}{\partial\xi_l}\left(c\log\det(\vecmat{P}_{\text{VADU}})+\log\det\left(\vecmat{W} +\tilde{\vecmat{\Sigma}}^{-1}\right)-c\log\det(\vecmat{P}_{\text{VADU}})\right) \\
    =&\;c\;\tr\left(\left(\vecmat{W}+\vecmat{D}^{-1}\right)^{-1}\frac{\partial \vecmat{W}}{\partial\xi_l}\right) \\
    &\;+\frac{1}{t}\sum_{i=1}^t \left(\underbrace{\left(\left(\vecmat{W} +\tilde{\vecmat{\Sigma}}^{-1}\right)^{-1}\vecmat{z}_i\right)^{T}\frac{\partial \vecmat{W}}{\partial\xi_l}\vecmat{P}_{\text{VADU}}^{-1}\vecmat{z}_i}_\text{$h(\vecmat{z}_i)$} \right.\\
    &\left.-\;c\underbrace{\left(\vecmat{P}_{\text{VADU}}^{-1} \vecmat{z}_i\right)^T\vecmat{B}^T\frac{\partial \vecmat{W}}{\partial\xi_l}\vecmat{B}\vecmat{P}_{\text{VADU}}^{-1}\vecmat{z}_i}_\text{$r(\vecmat{z}_i)$}\right),
\end{split}
\end{equation*}
where $c=\widehat{\text{Cov}}(h(\vecmat{z}_i),r(\vecmat{z}_i))/\widehat{\text{Var}}(r(\vecmat{z}_i))$ is the optimal weight for the variance reduction.

\subsubsection{Preconditioner LRAC} \label{P_LRACAppendix}
The following holds:
\begin{equation*}
\begin{split}
    \frac{\partial\log\det\left(\tilde{\vecmat{\Sigma}} \vecmat{W}+\vecmat{I}_n\right)}{\partial\mu^*_i}
    =&\;\frac{\partial}{\partial\mu^*_i}\left(\log\det(\vecmat{W}) + c\log\det(\vecmat{P}_{\text{LRAC}})+\log\det\left(\tilde{\vecmat{\Sigma}}+\vecmat{W}^{-1}\right)\right.\\
    &\left.-c\log\det(\vecmat{P}_{\text{LRAC}})\right)\\
    =&\;\frac{\partial}{\partial\mu^*_i}\left(\log\det(\vecmat{W}) + c\left(\log\det\left(\vecmat{L}_k^T\vecmat{W} \vecmat{L}_k+\vecmat{I}_k\right)+\log\det\left(\vecmat{W}^{-1}\right)\right)\right.\\
    &\;+\left.\log\det\left(\tilde{\vecmat{\Sigma}}+\vecmat{W}^{-1}\right)-c\log\det(\vecmat{P}_{\text{LRAC}})\right)\\
    =&\;\tr\left(\vecmat{W}^{-1}\frac{\partial \vecmat{W}}{\partial\mu^*_i}\right) + 
    c\left(\tr\left(\left(\vecmat{L}_k^T\vecmat{W} \vecmat{L}_k+\vecmat{I}_k\right)^{-1}\vecmat{L}_k^T\frac{\partial \vecmat{W}}{\partial\mu^*_i}\vecmat{L}_k\right)\right.\\
    &\left.-\tr\left(\vecmat{W}^{-1}\frac{\partial \vecmat{W}}{\partial\mu^*_i}\right)\right) \\
    &\;-\frac{1}{t}\sum_{i=1}^t \left(\underbrace{\left(\left(\tilde{\vecmat{\Sigma}}+\vecmat{W}^{-1}\right)^{-1}\vecmat{z}_i\right)^{T}\vecmat{W}^{-1}\frac{\partial \vecmat{W}}{\partial\mu^*_i}\vecmat{W}^{-1}\vecmat{P}_{\text{LRAC}}^{-1}\vecmat{z}_i}_\text{$h(\vecmat{z}_i)$} \right.\\
    &\;-\left.c\underbrace{\left(\vecmat{W}^{-1}\vecmat{P}_{\text{LRAC}}^{-1} \vecmat{z}_i\right)^T\frac{\partial \vecmat{W}}{\partial\mu^*_i}\vecmat{W}^{-1}\vecmat{P}_{\text{LRAC}}^{-1}\vecmat{z}_i}_\text{$r(\vecmat{z}_i)$}\right),
\end{split}
\end{equation*}
\begin{equation*}
\begin{split}
    \frac{\partial\log\det\left(\tilde{\vecmat{\Sigma}} \vecmat{W}+\vecmat{I}_n\right)}{\partial\theta_k}
    =&\;\frac{\partial}{\partial\theta_k}\left(c\log\det(\vecmat{P}_{\text{LRAC}})+\log\det\left(\tilde{\vecmat{\Sigma}}+\vecmat{W}^{-1}\right)-c\log\det(\vecmat{P}_{\text{LRAC}})\right)\\
    =&\;\frac{\partial}{\partial\theta_k}\left(c\log\det\left(\vecmat{L}_k^T\vecmat{W} \vecmat{L}_k+\vecmat{I}_k\right)+\log\det\left(\tilde{\vecmat{\Sigma}}+\vecmat{W}^{-1}\right)\right.\\
    &\left. -c\log\det(\vecmat{P}_{\text{LRAC}})\right)\\
    =&\;c\;\tr\left(\left(\vecmat{L}_k^T\vecmat{W} \vecmat{L}_k+\vecmat{I}_k\right)^{-1}\left(\frac{\partial \vecmat{L}_k^T}{\partial\theta_k}\vecmat{W} \vecmat{L}_k+\vecmat{L}_k^T\vecmat{W}\frac{\partial \vecmat{L}_k}{\partial\theta_k}\right)\right)\\
    &\;+\frac{1}{t}\sum_{i=1}^t \left(\underbrace{-\left(\left(\tilde{\vecmat{\Sigma}}+\vecmat{W}^{-1}\right)^{-1}\vecmat{z}_i\right)^{T}\tilde{\vecmat{\Sigma}}\frac{\partial\tilde{\vecmat{\Sigma}}^{-1}}{\partial\theta_k}\tilde{\vecmat{\Sigma}} \vecmat{P}_{\text{LRAC}}^{-1}\vecmat{z}_i}_\text{$h(\vecmat{z}_i)$} \right.\\
    &\;-\left.c\underbrace{\left(\vecmat{P}_{\text{LRAC}}^{-1} \vecmat{z}_i\right)^T\left(\frac{\partial \vecmat{L}_k}{\partial\theta_k}\vecmat{L}_k^T+\vecmat{L}_k\frac{\partial \vecmat{L}_k^T}{\partial\theta_k}\right)\vecmat{P}_{\text{LRAC}}^{-1}\vecmat{z}_i}_\text{$r(\vecmat{z}_i)$}\right),
\end{split}
\end{equation*}
%$\frac{\partial \vecmat{L}_k}{\partial\theta_k}$ could be evaluated using automatic differentiation.
\begin{equation*}
\begin{split}
    \frac{\partial\log\det\left(\tilde{\vecmat{\Sigma}} \vecmat{W}+\vecmat{I}_n\right)}{\partial\xi_l}
    =&\;\frac{\partial}{\partial\xi_l}\left(\log\det(\vecmat{W}) + c\log\det(\vecmat{P}_{\text{LRAC}})+\log\det\left(\tilde{\vecmat{\Sigma}}+\vecmat{W}^{-1}\right)\right. \\
    &\left.-c\log\det(\vecmat{P}_{\text{LRAC}})\right)\\
    =&\;\frac{\partial}{\partial\xi_l}\left(\log\det(\vecmat{W}) + c\left(\log\det\left(\vecmat{L}_k^T\vecmat{W} \vecmat{L}_k+\vecmat{I}_k\right)+\log\det\left(\vecmat{W}^{-1}\right)\right)\right.\\
    &\;+\left.\log\det\left(\tilde{\vecmat{\Sigma}}+\vecmat{W}^{-1}\right)-c\log\det(\vecmat{P}_{\text{LRAC}})\right)\\
    =&\;\tr\left(\vecmat{W}^{-1}\frac{\partial \vecmat{W}}{\partial\xi_l}\right) +
    c\left(\tr\left(\left(\vecmat{L}_k^T\vecmat{W} \vecmat{L}_k+\vecmat{I}_k\right)^{-1}\vecmat{L}_k^T\frac{\partial \vecmat{W}}{\partial\xi_l}\vecmat{L}_k\right)\right.\\
    &\left.-\tr\left(\vecmat{W}^{-1}\frac{\partial \vecmat{W}}{\partial\xi_l}\right)\right) \\
    &\;-\frac{1}{t}\sum_{i=1}^t \left(\underbrace{\left(\left(\tilde{\vecmat{\Sigma}}+\vecmat{W}^{-1}\right)^{-1}\vecmat{z}_i\right)^{T}\vecmat{W}^{-1}\frac{\partial \vecmat{W}}{\partial\xi_l}\vecmat{W}^{-1}\vecmat{P}_{\text{LRAC}}^{-1}\vecmat{z}_i}_\text{$h(\vecmat{z}_i)$} \right.\\
    &\;-\left.c\underbrace{\left(\vecmat{W}^{-1}\vecmat{P}_{\text{LRAC}}^{-1} \vecmat{z}_i\right)^T\frac{\partial \vecmat{W}}{\partial\xi_l}\vecmat{W}^{-1}\vecmat{P}_{\text{LRAC}}^{-1}\vecmat{z}_i}_\text{$r(\vecmat{z}_i)$}\right),
\end{split}
\end{equation*}
where $c=\widehat{\text{Cov}}(h(\vecmat{z}_i),r(\vecmat{z}_i))/\widehat{\text{Var}}(r(\vecmat{z}_i))$ is the optimal weight for the variance reduction.

\subsection{Derivatives of the Vecchia-approximated precision and covariance matrices} \label{vecchiaAppendix}
\begin{align}
    \frac{\partial \tilde{\vecmat{\Sigma}}^{-1}}{\partial \theta_k} &= \frac{\partial \vecmat{B}^T}{\partial \theta_k}\vecmat{D}^{-1}\vecmat{B} + \vecmat{B}^T\vecmat{D}^{-1}\frac{\partial \vecmat{B}}{\partial \theta_k} - \vecmat{B}^T\vecmat{D}^{-1}\frac{\partial \vecmat{D}}{\partial \theta_k}\vecmat{D}^{-1}\vecmat{B} \label{vecchiaSigmaderiv} \\
    \frac{\partial \tilde{\vecmat{\Sigma}}}{\partial \theta_k}      &= -\vecmat{B}^{-1}\frac{\partial \vecmat{B}}{\partial \theta_k}\vecmat{B}^{-1}\vecmat{D}\vecmat{B}^{-T} - \vecmat{B}^{-1}\vecmat{D}\vecmat{B}^{-T}\frac{\partial \vecmat{\vecmat{B}}^T}{\partial \theta_k}\vecmat{B}^{-T} + \vecmat{B}^{-1}\frac{\partial \vecmat{D}}{\partial \theta_k}\vecmat{B}^{-T} \label{vecchiaSigmaIderiv} \\
    \left(\frac{\partial \vecmat{B}}{\partial \theta_k}\right)_{i,N(i)} &= -\frac{\partial \vecmat{A}_i}{\partial \theta_k} \nonumber \\
    & = -\left(\frac{\partial \vecmat{\Sigma}}{\partial \theta_k}\right)_{i,N(i)}\vecmat{\Sigma}_{N(i)}^{-1} + \vecmat{\Sigma}_{i,N(i)}\vecmat{\Sigma}^{-1}_{N(i)}\left(\frac{\partial \vecmat{\Sigma}}{\partial \theta_k}\right)_{N(i)}\vecmat{\Sigma}^{-1}_{N(i)} \nonumber \\
    \frac{\partial \vecmat{D}}{\partial \theta_k} &= \left(\frac{\partial \vecmat{\Sigma}}{\partial \theta_k}\right)_{i,i} -\frac{\partial \vecmat{A}_i}{\partial \theta_k} \vecmat{\Sigma}_{N(i),i} - \vecmat{A}_i\left(\frac{\partial \vecmat{\Sigma}}{\partial \theta_k}\right)_{N(i),i} \nonumber
\end{align}
Calculating these gradients costs $O(nm^3).$
\clearpage
\subsection{Preconditioned conjugate gradient algorithm}\label{appendix:CGalgo}
\begin{algorithm}[ht!]
    \caption{Preconditioned conjugate gradient algorithm with Lanczos tridiagonal matrix}
    \begin{algorithmic}[1]
        \Require Matrix $\vecmat{A}$, preconditioner matrix $\vecmat{P}$, vector $\vecmat{b}$
        \Ensure $\vecmat{u}_{l+1} \approx \vecmat{A}^{-1}\vecmat{b}$, tridiagonal matrix $\tilde{\vecmat{T}}$
        \State{early-stopping $\gets$ false}
        \State{$\alpha_0 \gets 1$}
        \State{$\beta_0 \gets 0$}
        \State{$\vecmat{u}_0 \gets 0$}
        \State{$\vecmat{r}_0 \gets \vecmat{b} - \vecmat{A}\vecmat{u}_0$}
        \State{$\vecmat{z}_0 \gets \vecmat{P}^{-1}\vecmat{r}_0$}
        \State{$\vecmat{h}_0 \gets \vecmat{z}_0$}
        \For{$l \gets 0$ to $L$}
            \State{$\vecmat{v}_l \gets \vecmat{A}\vecmat{h}_l$}
            \State{$\alpha_{l+1} \gets \frac{\vecmat{r}_l^T\vecmat{z}_l}{\vecmat{h}_l^T\vecmat{v}_l}$}
            \State{$\vecmat{u}_{l+1} \gets \vecmat{u}_l + \alpha_{l+1} \vecmat{h}_l$}
            \State{$\vecmat{r}_{l+1} \gets \vecmat{r}_l - \alpha_{l+1} \vecmat{v}_l$}
            \If{$||\vecmat{r}_{l+1}||_2 <$ tolerance}
                \State{early-stopping $\gets$ true}
            \EndIf
            \State{$\vecmat{z}_{l+1} \gets \vecmat{P}^{-1}\vecmat{r}_{l+1}$}
            \State{$\beta_{l+1} \gets \frac{\vecmat{r}_{l+1}^T\vecmat{z}_{l+1}}{\vecmat{r}_l^T\vecmat{z}_l}$}
            \State{$\vecmat{h}_{l+1} \gets \vecmat{z}_{l+1} + \beta_{l+1} \vecmat{h}_l$}
            \State{$\tilde{\vecmat{T}}_{l+1,l+1} \gets \frac{1}{\alpha_{l+1}} + \frac{\beta_l}{\alpha_{l}}$}
            \If{$l > 0$}
                \State{$\tilde{\vecmat{T}}_{l,l+1}, \tilde{\vecmat{T}}_{l+1,l} \gets \frac{\sqrt{\beta_l}}{\alpha_{l}}$}
            \EndIf
            \If{early-stopping}
                \State{return $\vecmat{u}_{l+1}, \tilde{\vecmat{T}}$}
            \EndIf
        \EndFor
    \end{algorithmic}
\end{algorithm}
\subsection{Zero fill-in reverse incomplete Cholesky factorization}\label{appendix:ReverseIncompChol}
\begin{algorithm}[ht!]
    \caption{Zero fill-in reverse incomplete Cholesky algorithm}\label{ZIRC_algo}
    \begin{algorithmic}[1]
        \Require Matrix $\vecmat{A}\in\mathbb{R}^{n\times n}$, matrix $\vecmat{S}\in\mathbb{R}^{n\times n}$ with sparsity pattern
        \Ensure Sparse lower triangular matrix $\vecmat{L}$ with $\vecmat{A} \approx \vecmat{L}^T\vecmat{L}$

        \For{$j \gets n$ to $1$}
            \For{$i \gets n$ to $1$}
                \If{$(i,j) \in \vecmat{S}$ and $j \leq i$}
                    \State{$s \gets \vecmat{L}_{\cdot i}^T \vecmat{L}_{\cdot j}$}
                    \If{$j = i$}
                        \State{$L_{jj} \gets \sqrt{A_{jj}-s}$}
                    \Else
                        \State{$L_{ij} \gets \frac{A_{ji}-s}{L_{ii}}$}
                    \EndIf
                \EndIf
            \EndFor
        \EndFor
    \end{algorithmic}
\end{algorithm}

\clearpage

\subsection{Additional results for preconditioner comparison}\label{appendix:add_res_sim}

In the following, we present additional results comparing the VADU and the LRAC preconditioners on simulated data using the simulation setting described in Section \ref{CG_P_Comparison}. Based on the experimental results in this subsection and Section \ref{CG_P_Comparison} as well as the theoretical results in Sections \ref{section:preconditioners} and \ref{section:conv_theory}, we generally prefer the VADU preconditioner in practice. One benefit of the VADU preconditioner is that its construction involves essentially no computational cost. Further, we showed in Sections \ref{section:preconditioners} and \ref{section:conv_theory} that the VADU preconditioner removes the small eigenvalues which tend to delay the convergence of the CG method more than the large ones \citep{nishimura2022prior}. However, the LRAC preconditioner can be advantageous in situations when there is a lot of large-scale variation.

\begin{figure}[ht!]
\centering
    \includegraphics[width=0.7\linewidth]{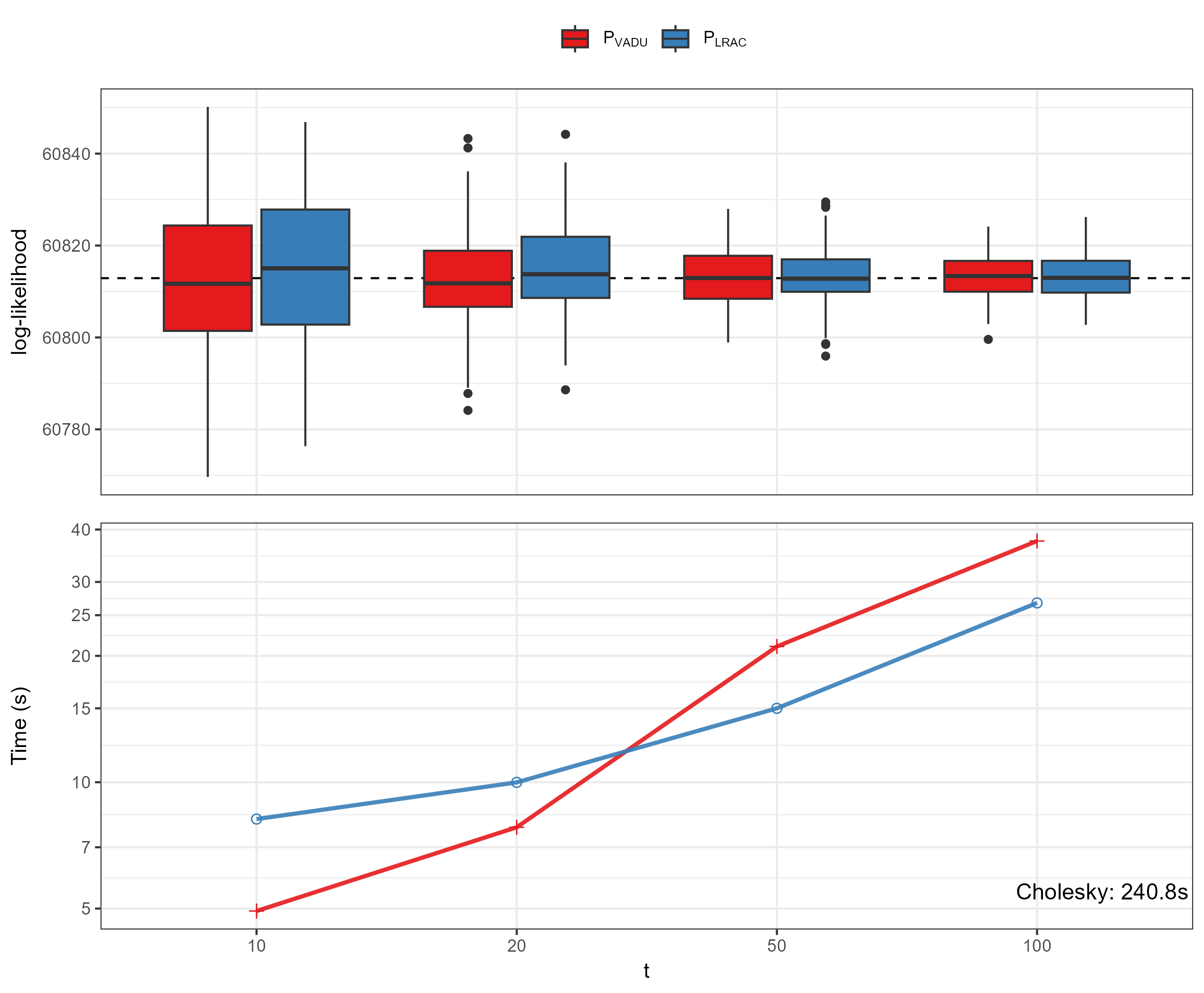}
    \caption{Negative log-marginal likelihood and runtime for the VADU and LRAC preconditioners for different numbers of random vectors $t$ on simulated data with $n=100'000$ and a range $\rho=0.05$. The dashed lines are the results for the Cholesky decomposition.} 
    \label{fig:P_comparison_CG_100000}
\end{figure}

\begin{figure}[ht!]
\centering
    \includegraphics[width=0.7\linewidth]{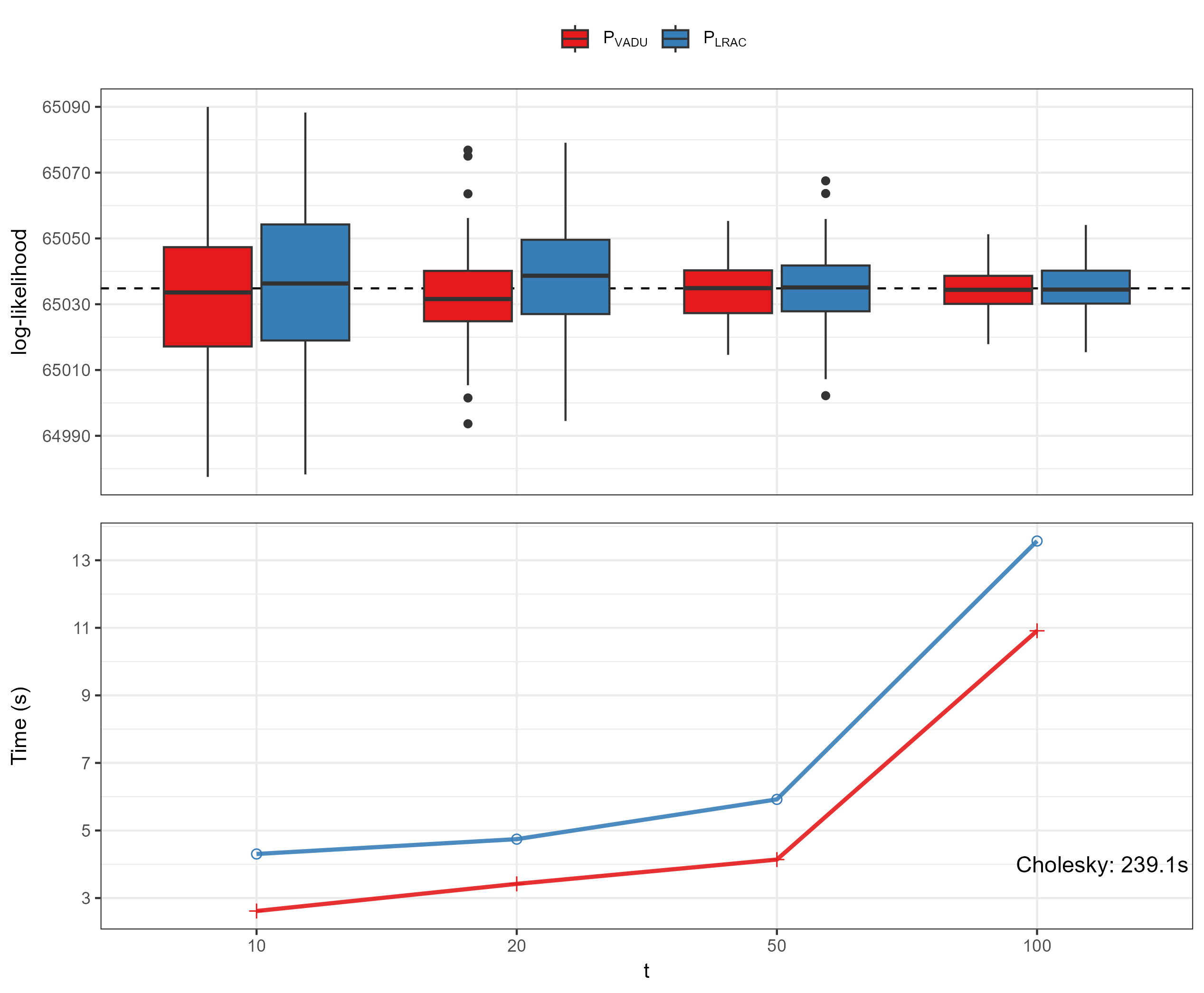}
    \caption{Negative log-marginal likelihood and runtime for the VADU and LRAC preconditioners for different numbers of random vectors $t$ on simulated data with $n=100'000$ and a range $\rho=0.01$.} 
    \label{fig:P_comparison_CG_rho=0.01}
\end{figure}

\begin{figure}[ht!]
\centering
    \includegraphics[width=0.7\linewidth]{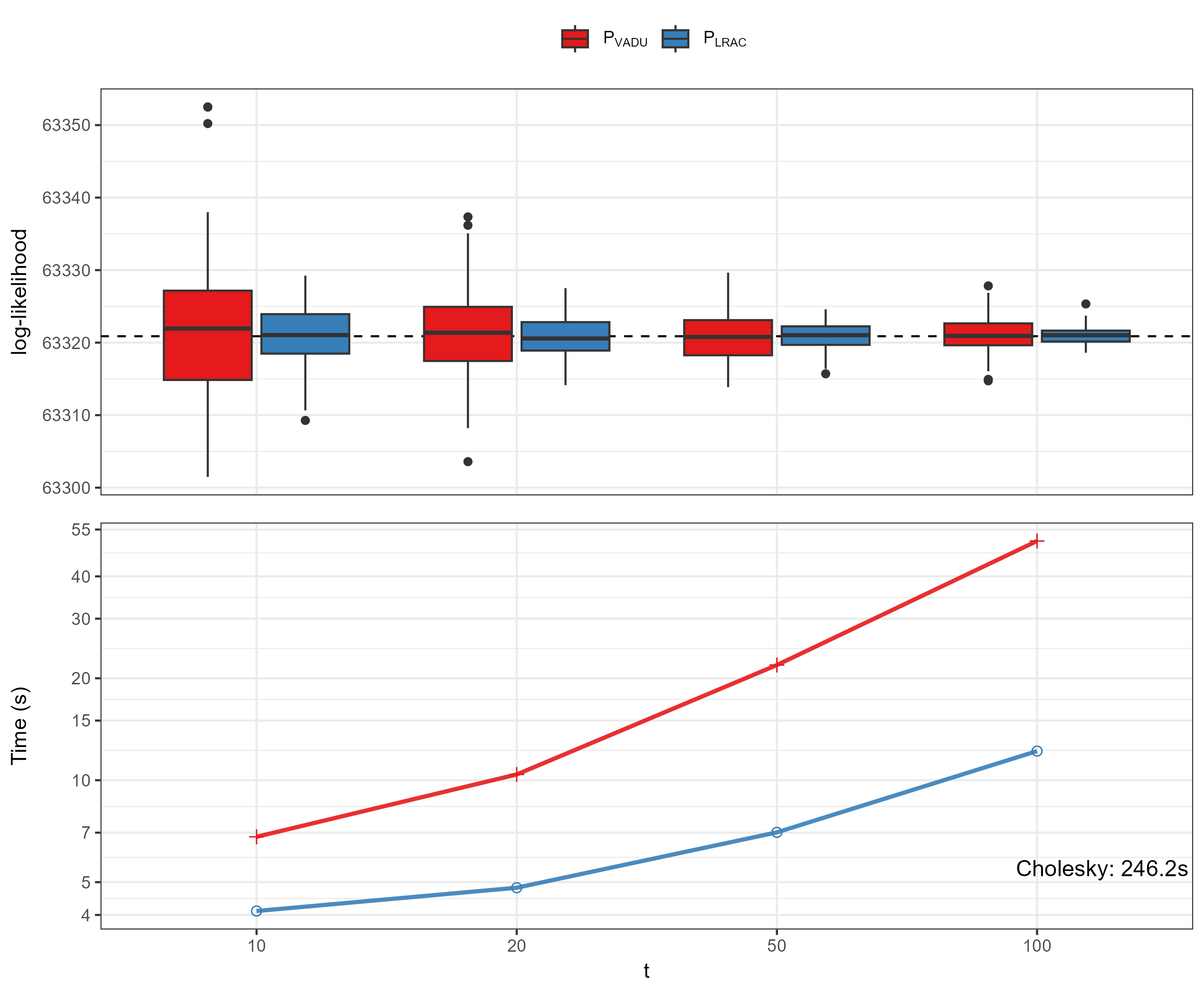}
    \caption{Negative log-marginal likelihood and runtime for the VADU and LRAC preconditioners for different numbers of random vectors $t$ on simulated data with $n=100'000$ and a range $\rho=0.25$.} 
    \label{fig:P_comparison_CG_rho=0.25}
\end{figure}

\begin{figure}[ht!]
\centering
    \includegraphics[width=0.7\linewidth]{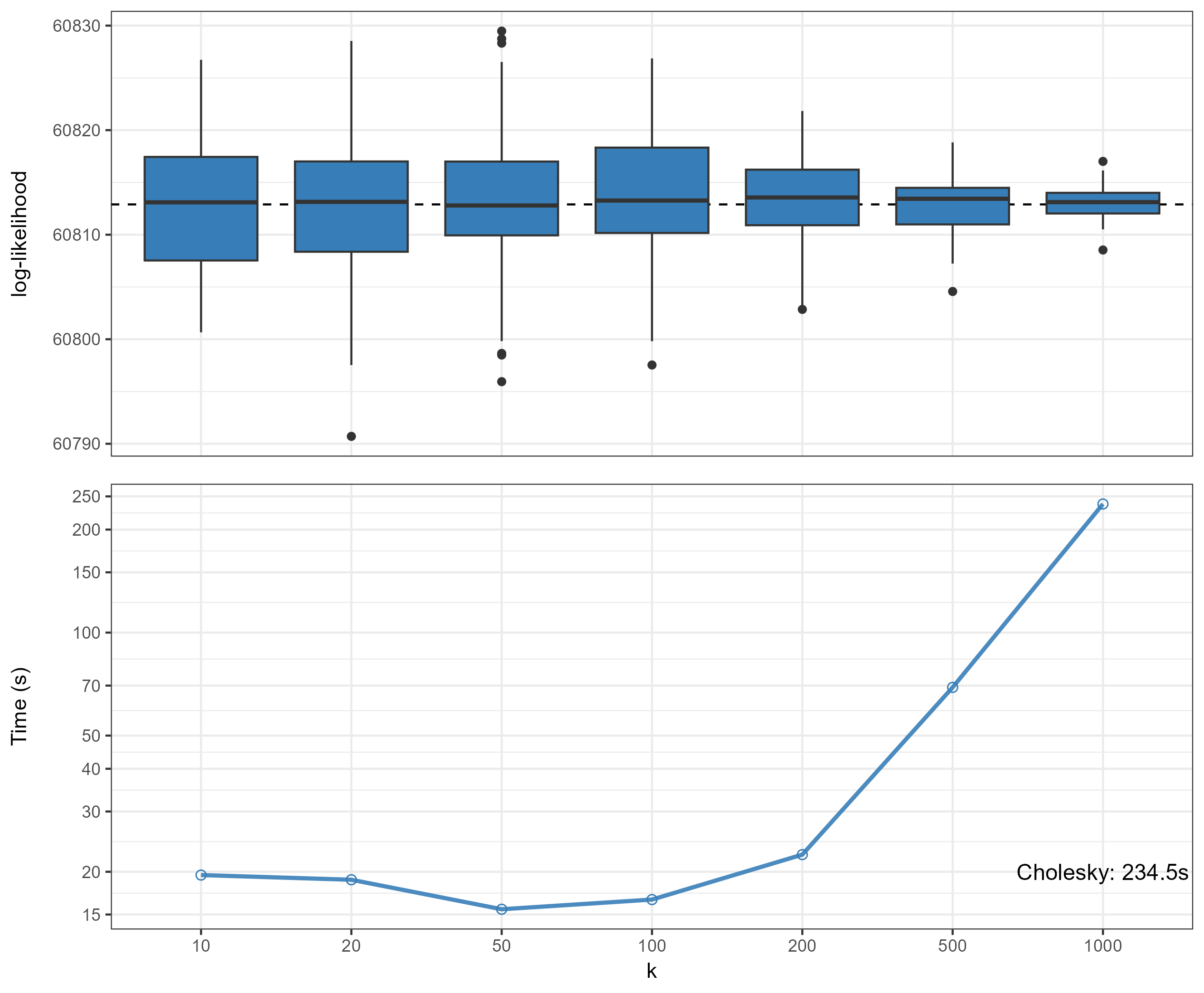}
    \caption{Negative log-marginal likelihood and runtime for the LRAC preconditioner for different ranks $k$ on simulated data with $n=100'000$ and a range of $\rho=0.05$.} 
    \label{fig:P_LRAC_ranks}
\end{figure}

\clearpage
\subsection{Additional results for predictive variances}\label{appendix:add_res_sim_var}
\begin{figure}[ht!]
\centering
    \includegraphics[width=0.6\linewidth]{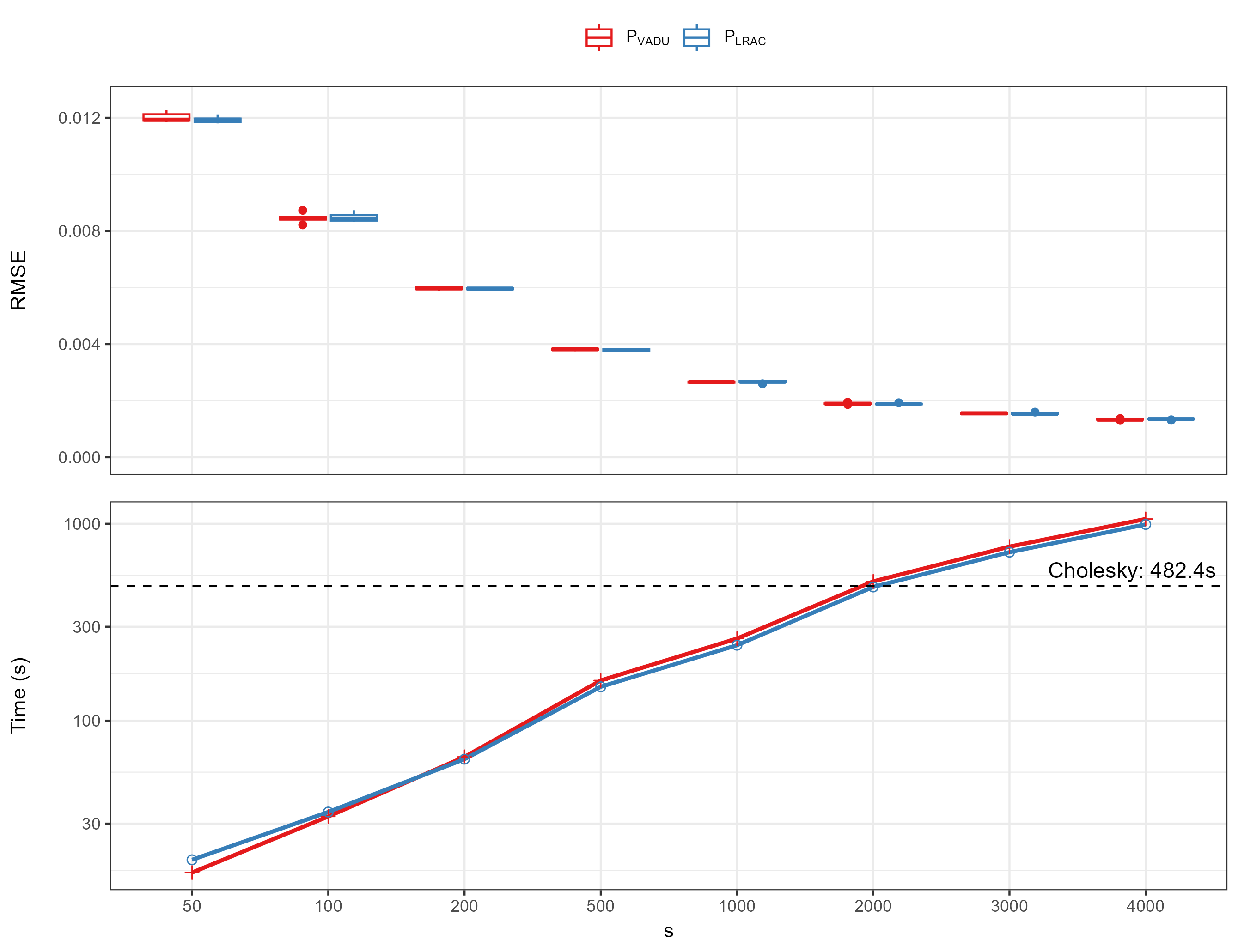}
        \caption{Accuracy and time for simulation-based predictive variances on simulated data with $n=n_p=100'000$. For each preconditioner (VADU, LRAC) and number of random vectors $s$, $10$ simulation-based predictions with different random vectors are performed. We report the root mean squared error (RMSE) for the predictive variances compared to Cholesky-based computations and the average runtime (s) for prediction.}
    \label{fig:sim1e5}
\end{figure}
\begin{figure}[ht!]
\centering
    \includegraphics[width=0.6\linewidth]{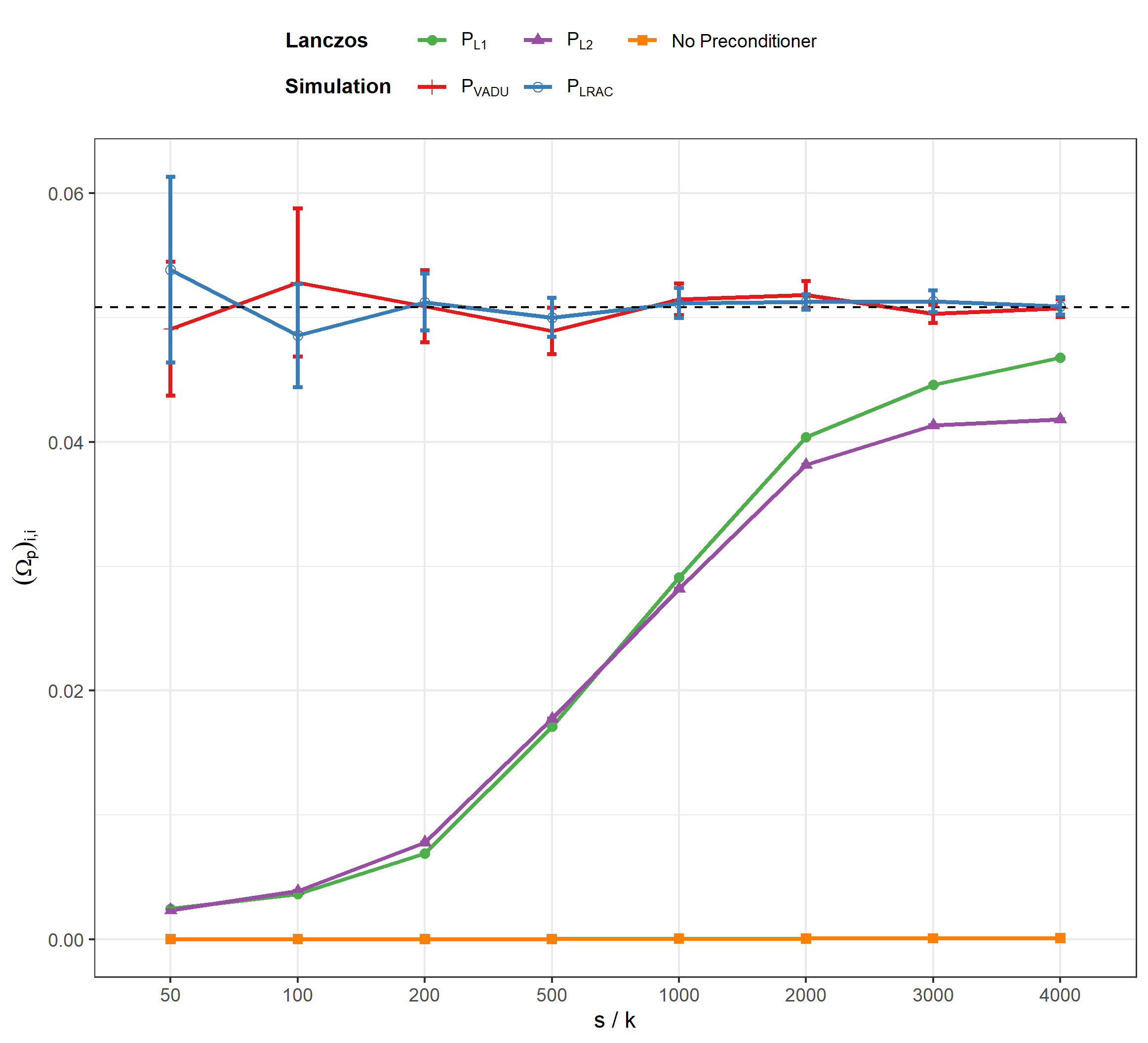}
    \caption{Predictive variance for a single observation $(\vecmat{\Omega}_p)_{i,i}$ for simulation- and Lanczos-based approximations with different numbers of random vectors $s$ and Lanczos ranks $k$ on simulated data with $n_p=n=100'000$. For each preconditioner and number of random vectors $s$, $10$ simulation-based predictions are performed with different random vectors, and the whiskers represent confidence intervals obtained as plus and minus twice the standard error of the mean. The dashed horizontal line shows the predictive variance based on the Cholesky decomposition.}
    \label{fig:latentpostpredictivepath}
\end{figure}
\begin{figure}[ht!]
\centering
    \includegraphics[width=0.6\linewidth]{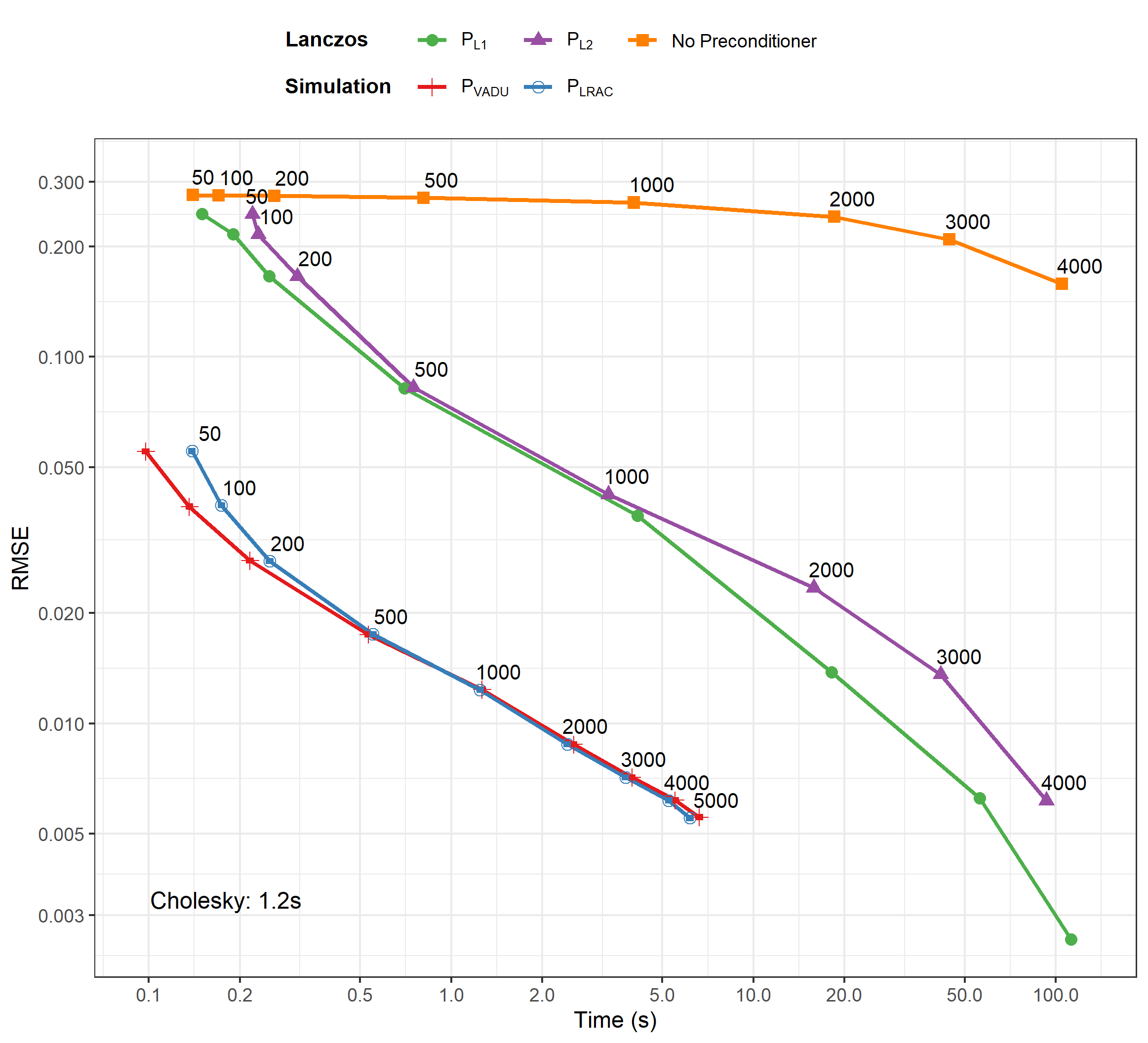}
    \caption{Comparison of simulation- and Lanczos-based methods for predictive variances: accuracy (= RMSE w.r.t. Cholesky-based predictive variances) vs. time (s) on simulated data of size $n=n_p=5'000$. The number of random vectors $s$ and the Lanczos rank $k$ are annotated in the plot beside the points with the results. The results for the simulation-based prediction are averaged over $50$ repetitions, but the corresponding whiskers with confidence intervals for the RMSE are not visible due to very small variances.}
    \label{fig:Lanczos_vs_simulation5000}
\end{figure}

\clearpage
\subsection{Additional results for parameter estimation and prediction in simulated experiments}\label{appendix:add_res_sim_est}
\begin{table}[ht!]
    \centering
    \begin{tabular}{l|lllll}
        \hline
        \hline
        && Cholesky-VL & Iterative-VL & GPVecchia\\ 
        \hline
        $\sigma_1^2$ & RMSE & 0.1160 & 0.1179 & 0.2484\\ 
                     & Bias & -0.0113 & -0.0101 & 0.1684\\ 
        \hline
        $\rho$ & RMSE & 0.0036& 0.0037 & 0.0479\\ 
               & Bias & -0.0002 & -0.0001 & 0.0469\\ 
        \hline
        \hline
    \end{tabular}
    \caption{Root mean squared error (RMSE) and bias of the covariance parameter estimates.} 
    \label{table:estimates20000}
\end{table}

\begin{table}[ht!]
    \centering
    \begin{tabular}{llllll}
        \hline
        \hline
        & Cholesky-VL & Iterative-VL & GPVecchia\\ 
        \hline
        $\overline{\text{RMSE}}$ & 0.3853 & 0.3853 & 0.4486\\ 
        sd($\overline{\text{RMSE}}$) & 0.0009 & 0.0009 & 0.0010\\ 
        \hline
        $\overline{\text{LS}}$ & 9199.3698 &  9205.9101 & -\\ 
        sd($\overline{\text{LS}}$)& 47.8820 & 48.0022 & -\\ 
        \hline
        \hline
    \end{tabular}
    \caption{Average root mean squared error (RMSE) for predictive means and average log score (LS) for probabilistic predictions with corresponding standard errors. For GPVecchia, the log score could not be calculated due to negative predictive variances.} 
    \label{table:prediction20000}
\end{table}

\clearpage
\subsection{Additional results for the MODIS column water vapor data application}\label{appendix_rw}
\begin{table}[ht!]
	\centering
	\begin{tabular}{lcccccc} 
		\hline
		\hline
		& $\alpha$ & $\sigma_1^2$ & $\rho$ & $\beta_0$ & $\beta_1$ & $\beta_2$\\ 
    	\hline
  	Cholesky-VL & 28.21 & 0.30 & 15.32 km & $-2.0 \times 10^{-2}$ & $-4.9 \times 10^{-5}$ & $-8.5 \times 10^{-4}$\\ 
  	\hline
  	Iterative-VL & 28.13 & 0.30 & 15.35 km & $-2.0 \times 10^{-2}$ & $-4.9 \times 10^{-5}$ & $-8.5 \times 10^{-4}$\\ 
		\hline
	  GPVecchia & 0.95 & 0.25 & 149.15 km & $+8.2 \times 10^{-2}$ & $-7.0 \times 10^{-5}$ & $-7.6 \times 10^{-4}$ \\
        \hline
	  LowRank & 6.47 & 1.03 & 153.40 km & $+2.7 \times 10^{-1}$ & $-3.9 \times 10^{-5}$ & $-7.8 \times 10^{-4}$ \\
        \hline
		\hline
	\end{tabular}
	\caption{Estimated parameters on the MODIS column water vapor data set.}
    \label{table:realEstimates}
\end{table}
\begin{table}[ht!]
	\centering
	\begin{tabular}{lllllll}
		\hline
		\hline
		$\sigma_1^2$ & $\rho$ & $\beta_0$ & $\beta_1$ & $\beta_2$ & RMSE & CRPS\\ 
		\hline
		0.28 & 11.86km & $-2.1\times 10^{-2}$ & $-4.6\times 10^{-5}$ & $-8.6\times 10^{-4}$ & 0.0791 & 0.0441 \\ 
		\hline
		\hline
	\end{tabular}
	\caption{Estimated parameters, RMSE, and CRPS for Iterative-VL when parameter estimation is executed on the training data set for fixed shape $\alpha=28.13$.} 
    \label{table:realLarge}
\end{table}

\begin{figure}[ht!]
\centering
    \includegraphics[width=1\linewidth]{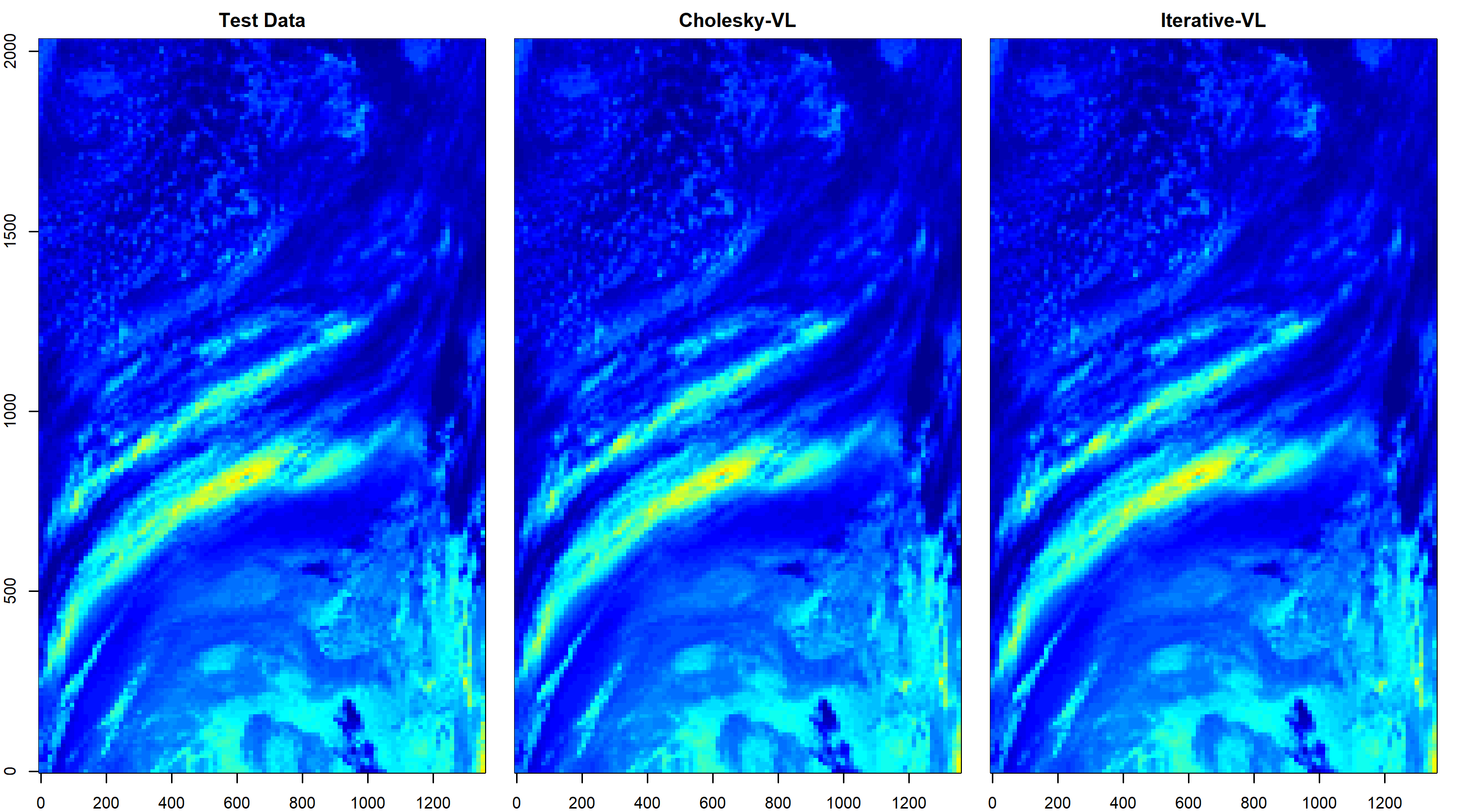}
    \includegraphics[width=0.77\linewidth]{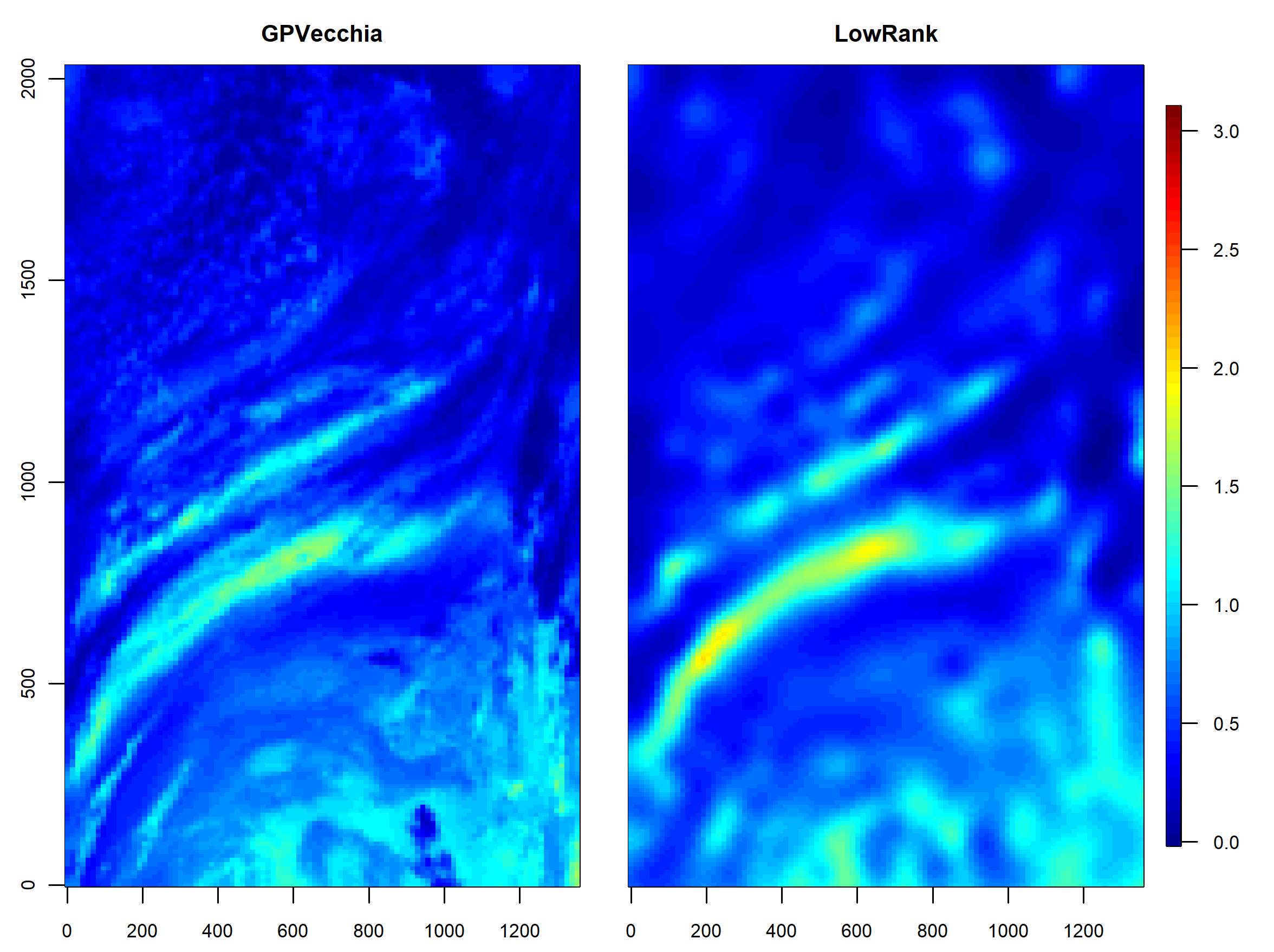}
    \caption{Test data and posterior predictive response means for the MODIS column water vapor data set.}
    \label{fig:realmean}
\end{figure}

\begin{figure}[ht!]
\centering
    \includegraphics[width=1\linewidth]{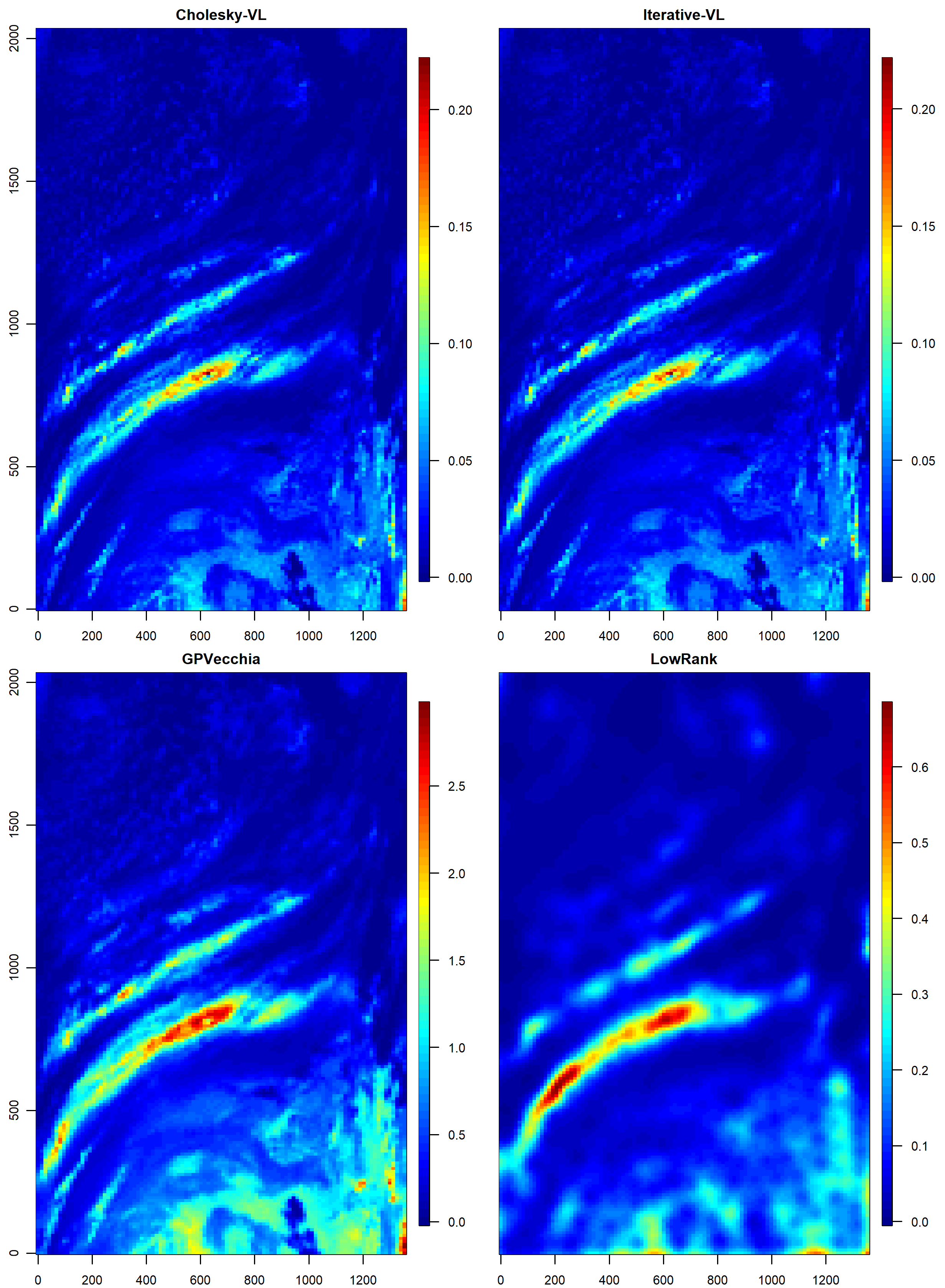}
    \caption{Posterior predictive variances for the MODIS column water vapor data set.}
    \label{fig:realvar}
\end{figure}

\clearpage

\subsection{Precipitation data application}\label{precip_data}
%https://ftp.cpc.ncep.noaa.gov/precip/CPC_UNI_PRCP/GAUGE_CONUS/DOCU/PRCP_CU_GAUGE_V1.0CONUS_0.25deg.README
We apply our proposed methods to a spatial precipitation data set from the National Oceanic and Atmospheric Administration (NOAA). Daily precipitation values in mm (divided by 10) are provided for a 0.25° longitude-latitude grid over the contiguous United States with coordinates from 235.375° to 292.875° east and 25.125° to 49.125° north. The data can be accessed using the R function \texttt{rnoaa::cpc\_prcp}, and in our analysis we use the accumulated precipitation values for the month of November 2019. 
%gauge network // interpolating the corresponding station values using the optimal interpolation method

Since the accumulated precipitation values are strictly positive, we follow \citet{quiroz2023fast} and use a gamma likelihood $y_i|\mu_i \sim \Gamma(\alpha, \alpha \exp{(-\mu_i)})$. We use a Matérn covariance function with smoothness parameter $\nu=1.5$ and a linear predictor term, $\vecmat{F}(\vecmat{X}) = \vecmat{X}\vecmat{\beta}$, with an intercept and the longitude and latitude coordinates as predictor variables. 
%In this application, we compare Vecchia-Laplace approximations based on the Cholesky decomposition (`Cholesky-VL') and iterative methods (`Iterative-VL') with the approach of \citet{zilber2021vecchia} (`GPVecchia') and a low-rank approximation (`LowRank'). 
For the low-rank approximation, we set the number of inducing points to 500, and for the simulation-based prediction variances we choose $s=500$ samples. Otherwise, for the Vecchia approximation and the iterative methods, we use the same settings as described in Section \ref{simSetting}.

To evaluate the prediction accuracy, we randomly split the total 13'626 observations into a training and a test set with $n=10'901$ and $n_p=2'725$ respectively. The joint estimation of all parameters with GPVecchia crashed and singular matrices occurred. Therefore, we again use the same estimation procedure as in \citet{zilber2021vecchia} and report the runtime to calculate the marginal likelihood. The latter is evaluated on the training set at $\alpha=5$, $\sigma_1^2=0.4$, $\rho=0.86$, and for $\vecmat{\beta}$ we use the generalized linear model estimate. Table \ref{table:NOAAPrediction} reports the runtime, test RMSE, and CRPS. The RMSE for Iterative-VL and Cholesky-VL is identical up to the fourth digit and almost four times smaller compared to GPVecchia and less than half the RMSE of the low-rank approximation. The CRPS for Iterative-VL and Cholesky-VL are also very similar and about ten and three times smaller compared to GPVecchia and the low-rank approximation, respectively. Further, iterative methods for the Vecchia-Laplace approximation are approximately four times faster in evaluating the marginal likelihood than Cholesky-based calculations and the GPVecchia approach. Estimated parameters are reported in Table \ref{table:NOAAEstimates}.
%The estimate of the shape parameter obtained with Iterative-VL is $\alpha=6.25 \times 10^5$ and with Cholesky-VL $\alpha=1.32 \times 10^6$. I.e., the likelihood is close to a normal distribution. In contrast, the estimated shape parameter of GPVecchia is $\alpha=0.243$, which corresponds to a gamma distribution that is much more right-skewed.

Figure \ref{fig:NOAAmean} shows the test data and maps of posterior predictive means $\mathbb{E}[\vecmat{y}_p|\vecmat{y},\vecmat{\theta},\vecmat{\xi}]$ and Figure \ref{fig:NOAAvar} shows the posterior predictive variances $\text{Var}[\vecmat{y}_p|\vecmat{y},\vecmat{\theta},\vecmat{\xi}]$. Iterative-VL and Cholesky-VL predict the precipitation values very accurately, while GPVecchia fails to predict the high precipitation values and instead oversmoothes. Further, some predictive variances for GPVecchia are very large. The latter also holds for the low-rank approximation. Iterative-VL and Cholesky-VL generally have smaller predictive variances.

\begin{table}[ht!]
	\centering
	\begin{tabular}{lllll}
		\hline
		\hline
		& Cholesky-VL & Iterative-VL & GPVecchia & LowRank\\ 
		\hline
		RMSE & 4.0836 & 4.0837 & 15.3807 & 11.8916\\ 
        \hline
		CRPS & 2.0587 & 2.0576 & 22.6582 & 6.1480\\ 
        \hline
		Time(s) & 2.6s & 0.7s & 2.8s & 1.1s\\ 
		\hline
		\hline
	\end{tabular}
	\caption{Prediction accuracy (lower = better) and runtime (s) for calculating marginal likelihoods on the NOAA precipitation data set.}
    \label{table:NOAAPrediction}
\end{table}

\begin{figure}[ht!]
\centering
    \includegraphics[width=1\linewidth]{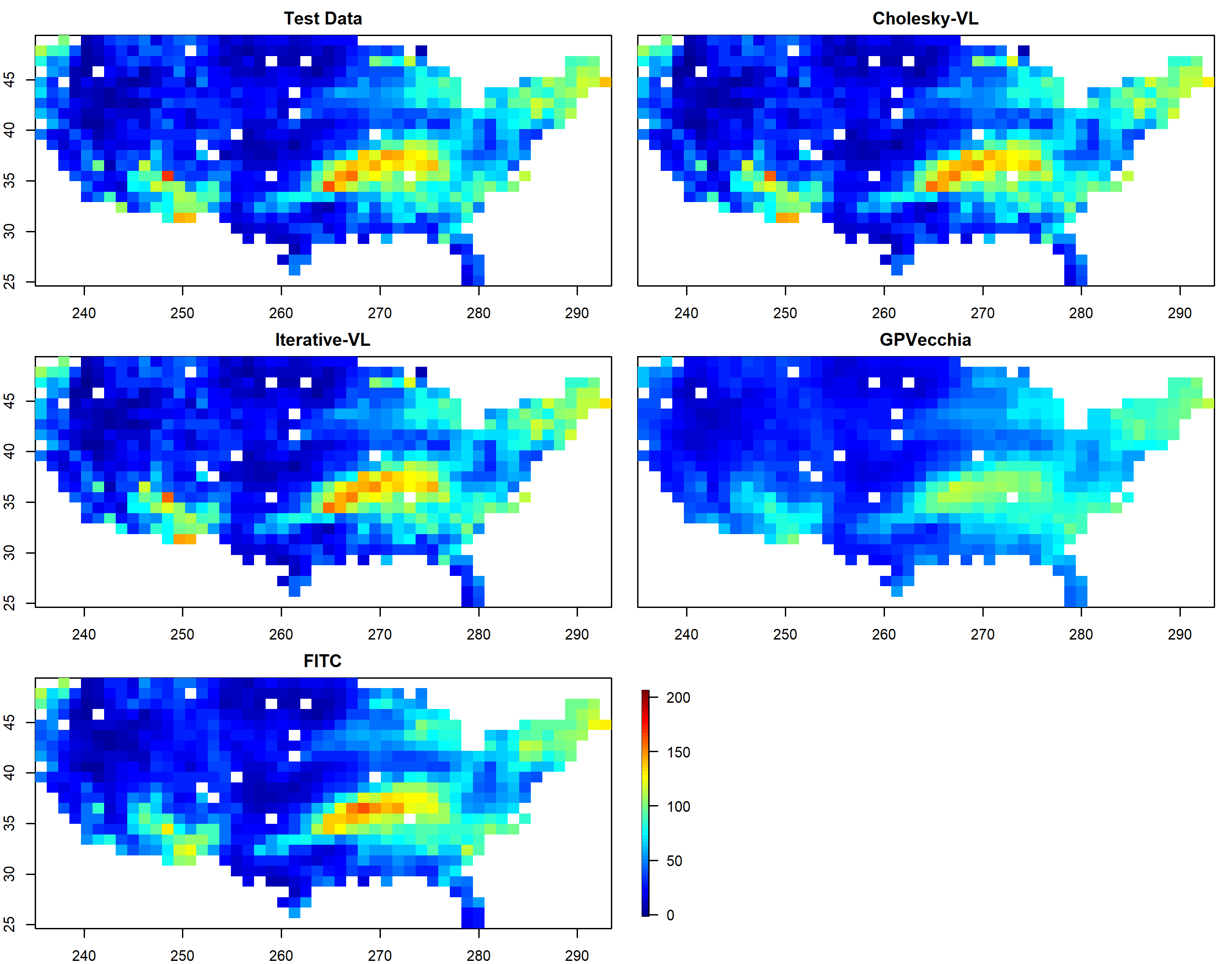}
    \caption{Test data and posterior predictive response means for the NOAA precipitation data set.}
    \label{fig:NOAAmean}
\end{figure}

\begin{table}[ht!]
	\centering
	\begin{tabular}{lcccccc} 
		\hline
		\hline
		&           $\alpha$ & $\sigma_1^2$ & $\rho$ & $\beta_0$ & $\beta_1$ & $\beta_2$\\ 
    	\hline
  	Cholesky-VL & $1.32 \times 10^{6}$ & 0.41 & 0.86 & -1.13 & $+2.1 \times 10^{-2}$ & $-2.0 \times 10^{-2}$\\ 
  	\hline
  	Iterative-VL & $6.26 \times 10^{5}$ & 0.41 & 0.86 & -1.27 & $+2.1 \times 10^{-2}$ & $-1.9 \times 10^{-2}$\\ 
		\hline
	  GPVecchia & $2.43 \times 10^{-1}$ & 0.42 & 11.32 & -0.29 & $+1.9 \times 10^{-2}$ & $-2.9 \times 10^{-2}$ \\
        \hline
	  LowRank & $1.75 \times 10^{1}$ & 0.50 & 2.04 & +0.54 & $+1.8 \times 10^{-2}$ & $-3.0 \times 10^{-2}$ \\
        \hline
		\hline
	\end{tabular}
	\caption{Estimated parameters on the NOAA precipitation data set.}
    \label{table:NOAAEstimates}
\end{table}

\begin{figure}[ht!]
\centering
    \includegraphics[width=1\linewidth]{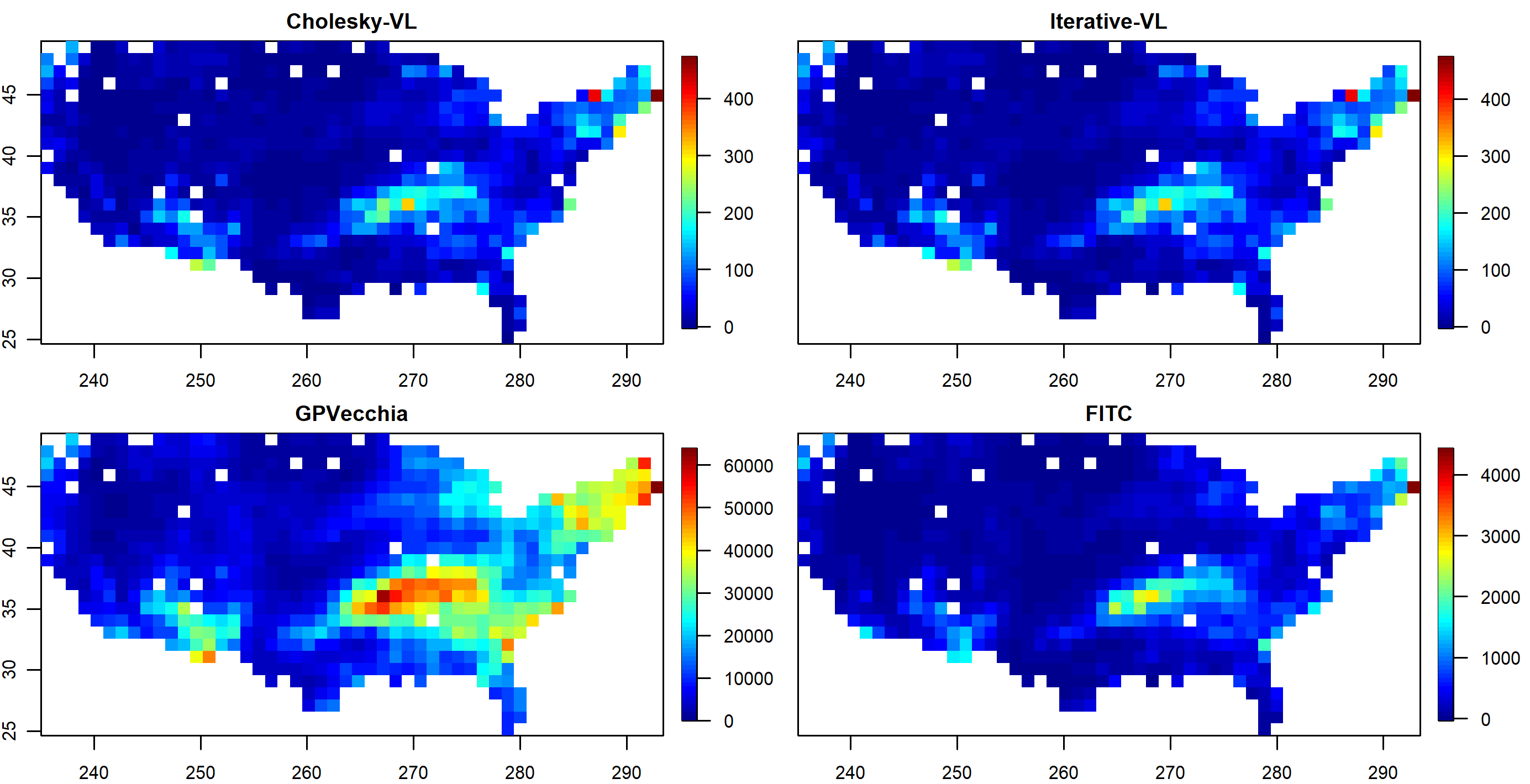}
    \caption{Posterior predictive variances for the NOAA precipitation data set.}
    \label{fig:NOAAvar}
\end{figure}

\end{document}